\documentclass{article}

\usepackage[english]{babel}
\usepackage{enumerate}
\usepackage{amsthm} 

\usepackage[a4paper,top=3cm,bottom=3cm,left=3cm,right=3cm,marginparwidth=1.75cm]{geometry}

\usepackage[numbers]{natbib}
\usepackage{pifont}

\usepackage{amsmath}

\usepackage{amsfonts}
\usepackage{bbm}
\usepackage{graphicx}
\usepackage[colorlinks=true, allcolors=blue]{hyperref}
\usepackage{mathtools}
\usepackage{breqn}
\usepackage{float}
\usepackage{xcolor}
\newtheorem{theorem}{Theorem}[section]

\newtheorem{proposition}[theorem]{Proposition}

\newtheorem{example}[theorem]{Example}

\numberwithin{equation}{section}

\def\E{{\mathbb{E}}}
\newcommand{\R}{{\mathbb R}}
\newcommand{\Mid}{{\ \Big|\ }}

\newcommand{\cmark}{\ding{51}}%
\newcommand{\xmark}{\ding{55}}%
\definecolor{blue0}{RGB}{0,77,153} 
\definecolor{red0}{RGB}{179,0,77} 
\definecolor{green0}{RGB}{134,219,76} 
\definecolor{gray0}{RGB}{84,97,110}

\usepackage{mathtools}
\usepackage{authblk}
\mathtoolsset{showonlyrefs}
\title{Joint SPX \& VIX calibration with Gaussian polynomial volatility models:
deep pricing with quantization hints}

\author[1]{Eduardo Abi Jaber\thanks{eduardo.abi-jaber@polytechnique.edu. The first author is grateful for the financial support from the Chaires FiME-FDD, Financial Risks, Deep Finance \& Statistics at Ecole Polytechnique. We would like to thank Stefano De Marco for fruitful discussions.}}
\author[2]{Camille Illand}
\author[3]{Shaun (Xiaoyuan) Li\thanks{{shaunlinz02@gmail.com.The third author acknowledges the financial support from the Amadeus partnership, Partenariat Hubert Curien (PHC), research project "Polynomial Volterra processes and their applications in finance", {and from AXA Investment Managers.}.}}}
\affil[1]{Ecole Polytechnique, CMAP}
\affil[2,3]{AXA Investment Managers}
\affil[3]{Université Paris 1 Panthéon-Sorbonne, CES}

\begin{document}

\maketitle

\begin{abstract}
We consider the joint SPX \& VIX calibration within a general class of \textit{Gaussian polynomial volatility models} in which the volatility of the SPX is assumed to be a polynomial function of a Gaussian Volterra process defined as a stochastic convolution between a kernel and a Brownian motion. By performing joint calibration to daily SPX \& VIX implied volatility surface data between 2011 and 2022, we compare the empirical performance of different kernels and their associated Markovian and non-Markovian models, such as rough and non-rough path-dependent volatility models. To ensure an efficient calibration and fair comparison between the models, we develop a generic unified method in our class of models for fast and accurate pricing of SPX \& VIX derivatives based on functional quantization and Neural Networks. For the first time, we identify a \textit{conventional one-factor Markovian continuous stochastic volatility model} that can achieve remarkable fits of the implied volatility surfaces of the SPX \& VIX together with the term structure of VIX Futures. {What is even more remarkable is that our \textit{conventional one-factor Markovian continuous stochastic volatility model} outperforms, in all market conditions, its rough and
 non-rough path-dependent counterparts with the same number of parameters.}
\end{abstract}

\begin{description}
\item[JEL Classification:] G13, C63, G10.  
\item[Keywords:] SPX \& VIX modeling, Stochastic volatility, Gaussian Volterra processes, Quantization, Neural Networks.



\end{description}

\bigskip

\newpage 
\section{Introduction}

  Launched in 1993 by the CBOE, the VIX has become one of the most widely followed volatility indexes. It represents an estimation of the S\&P 500 index (SPX) expected volatility over a one-month period.  More precisely, the VIX is calculated by aggregating weighted prices of SPX puts and calls over a wide range of strikes and maturities \cite{cboe}. By construction, the VIX expresses an interpolation between several points of the SPX implied volatility term structure. Thus, the modeling and pricing of VIX  options for a given maturity $T$ naturally requires some consistency with SPX options maturing up to one month ahead of $T$.   Furthermore, computing the implied volatility of VIX options using Black’s formula requires VIX Futures that also need to be priced consistently.

By joint SPX \& VIX calibration problem, we mean the calibration of a model across several maturities to European call and put options on SPX \& VIX together with VIX Futures.  Such joint calibration turns out to be quite challenging for several reasons: a multitude of instruments to be calibrated (SPX \& VIX call/put options, VIX Futures) across several maturities (to stay consistent with the construction of the VIX), characterized by a steep upward-sloping VIX volatility smile, together with the important at-the-money (ATM) SPX skew that becomes more pronounced for smaller maturities.

{The very first, although unsuccessful attempt in the literature at joint calibration was made in 2008 in \cite{gatheral2008consistent} using a double constant elasticity of variance (CEV) model. Since then, numerous models with varying degrees of success have been proposed. Early approaches often involved adding jumps to the dynamics of the underlying asset and/or its volatility process. For instance, \cite{baldeaux2014consistent} incorporated Poisson jumps into the $3/2$ model, while \cite{cont2013consistent} modeled the joint dynamics of forward variance and the underlying asset with additional Lévy jumps. Similarly, \cite{kokholm2015joint} extended the classical Heston model by adding jumps, and \cite{papanicolaou2014regime} proposed a model with jumps and regime-switching.

Following the statistical study \cite{jaisson2016rough}, rough volatility models have also been proposed for the joint calibration problem. Examples include the quadratic rough Heston model \cite{gatheral2020quadratic,rosenbaum2021deep} and the rough Heston model with added Hawkes jumps \cite{bondi2022rough}. In parallel, other types of parametric models have emerged, such as the path-dependent volatility model and its four-factor Markovian version \cite{guyon2022volatility},  and the multi-factor model proposed by \cite{romer2022empirical}, which describes the instantaneous volatility of the underlying using a hyperbolic function of two Ornstein-Uhlenbeck processes.

Non-parametric (or over-parametric) methods have also been proposed in recent years. \cite{guo2022joint}, \cite{guyon2020joint}, and \cite{guyon2022dispersion} use optimal transport techniques to model the joint dynamics of SPX \& VIX between two maturities, $T_1$ and $T_2$. Meanwhile, sophisticated functional approximation techniques have also found application in joint calibration, such as the Neural Network approach used in \cite{guyonneural} to learn the coefficients of the SDE of the instantaneous volatility process, and the (time-extended) signature-based methods of \cite{cuchiero2023joint}: A (truncated) linear functional designed to approximate the dynamics of instantaneous volatility driven by three Ornstein-Uhlenbeck processes.}

Although different in their mathematical nature, these aforementioned models share in common the fact that they allow for 1) large price movements of the SPX on very short time scales with some forms of spikes in the ‘instantaneous’ volatility process due to 
a large ‘vol-of-vol’, and 2) fast mean reversions towards relatively low volatility regimes. We believe these are the two crucial ingredients for the joint calibration problem. 

{Despite the significant interest in joint calibration and the abundance of proposed models, there is a noticeable lack of extensive empirical studies assessing their effectiveness in jointly calibrating the SPX \& VIX volatility surface. One exception is the rough quadratic Heston model studied in \cite{rosenbaum2021deep}, though it did not include VIX Futures in its empirical results. In fact, most models discussed above have only demonstrated partial fits, typically being calibrated over a select few days to justify their effectiveness.
}

{{A common theme in the literature mentioned above is that} \textit{conventional one-factor continuous Markovian stochastic volatility models} are not able to achieve a decent joint calibration.  Our main motivations can be stated as follows:
{\vspace{-0.5cm}}
\begin{center}
{\vspace{0.2cm}}
\textit{Can joint calibration be achieved \textbf{without} appealing to   multiple factors, jumps, roughness or path-dependency?} \\
{\vspace{0.2cm}}
\textit{Is joint calibration \textbf{possible} with conventional one-factor continuous Markovian models?}\\
\end{center}

In a nutshell, we show in this paper that the answer to both questions is {a resounding}: \textit{{Yes.}} {While several empirical studies provide evidence of jumps,  path-dependency, and multiple timescales in the underlying asset and/or its volatility process (see for example, \cite{barndorff2006econometrics, ding1993long, guyon2022volatility}), our approach focuses on assessing the joint calibration capabilities of more streamlined models. This can provide insights into how much complexity is truly necessary for joint calibration versus what can be achieved with reduced assumptions. In practice, simpler models may be preferable due to their computational efficiency. Our work is not intended to dismiss these approaches but to complement the existing literature.}

By performing joint calibration on daily SPX \& VIX implied volatility surface data between 2011 and 2022 using a large class of models, we identify for the first time a \textit{conventional one-factor Markovian continuous stochastic volatility model} {that can achieve remarkable fits for a wide range of maturity slices $[T_s,T_e]$ for VIX implied volatility surface and of maturity slices $[T_s, T_e + 1\text{ month}]$ for SPX implied volatility surface, together with the term structure of VIX Futures as shown on Figure \ref{fig:Intro}}. What is even more remarkable is that our \textit{conventional one-factor Markovian continuous stochastic volatility model} outperforms its rough and non-rough path-dependent counterparts with the same number of parameters: 6 effective parameters that govern the dynamics of the model in addition to the usual input curve that allows to match certain term structures of volatility.

  \begin{figure}[H]
    \centering
    \includegraphics[width=0.5\textwidth]{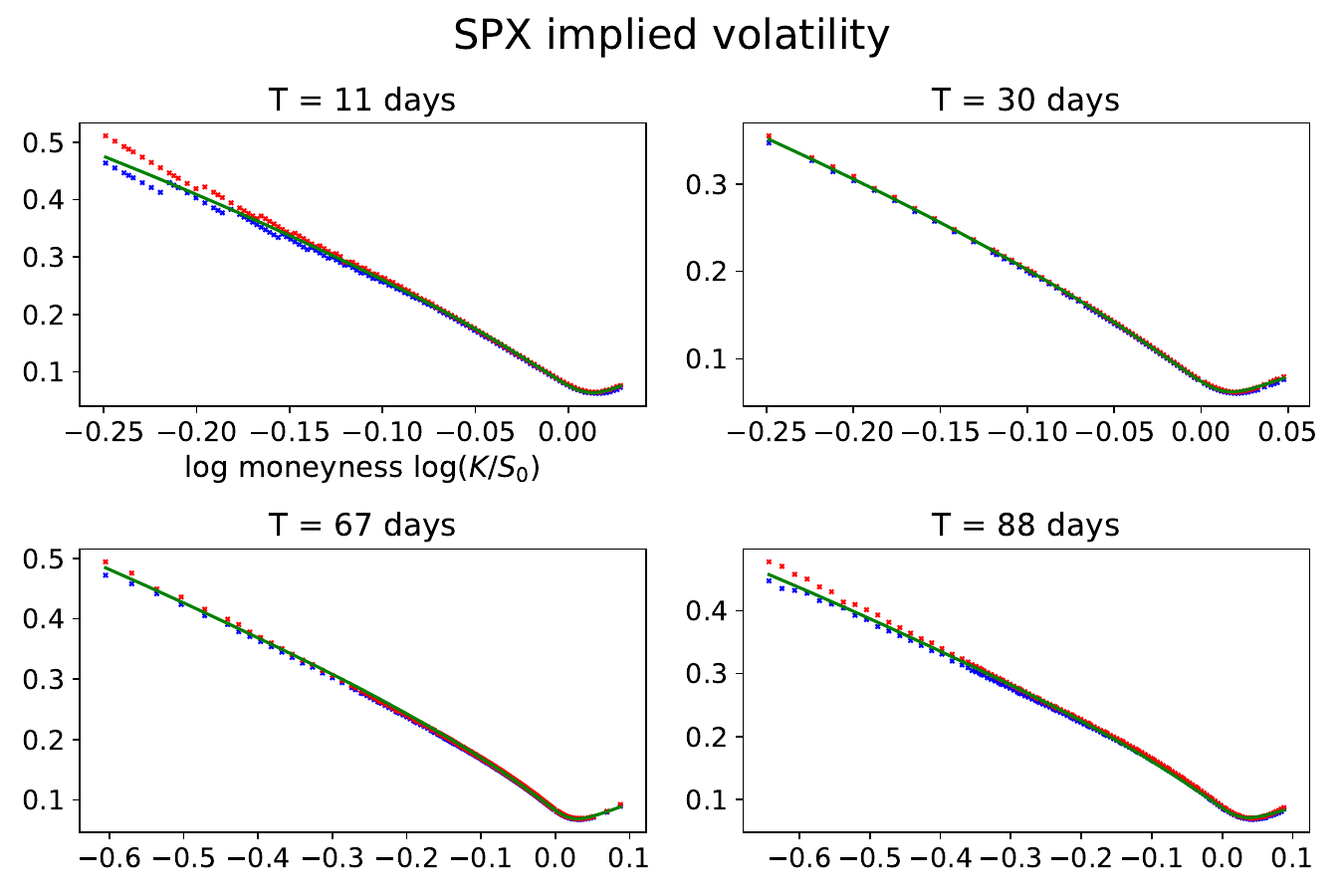}%
    \includegraphics[width=0.5\textwidth]{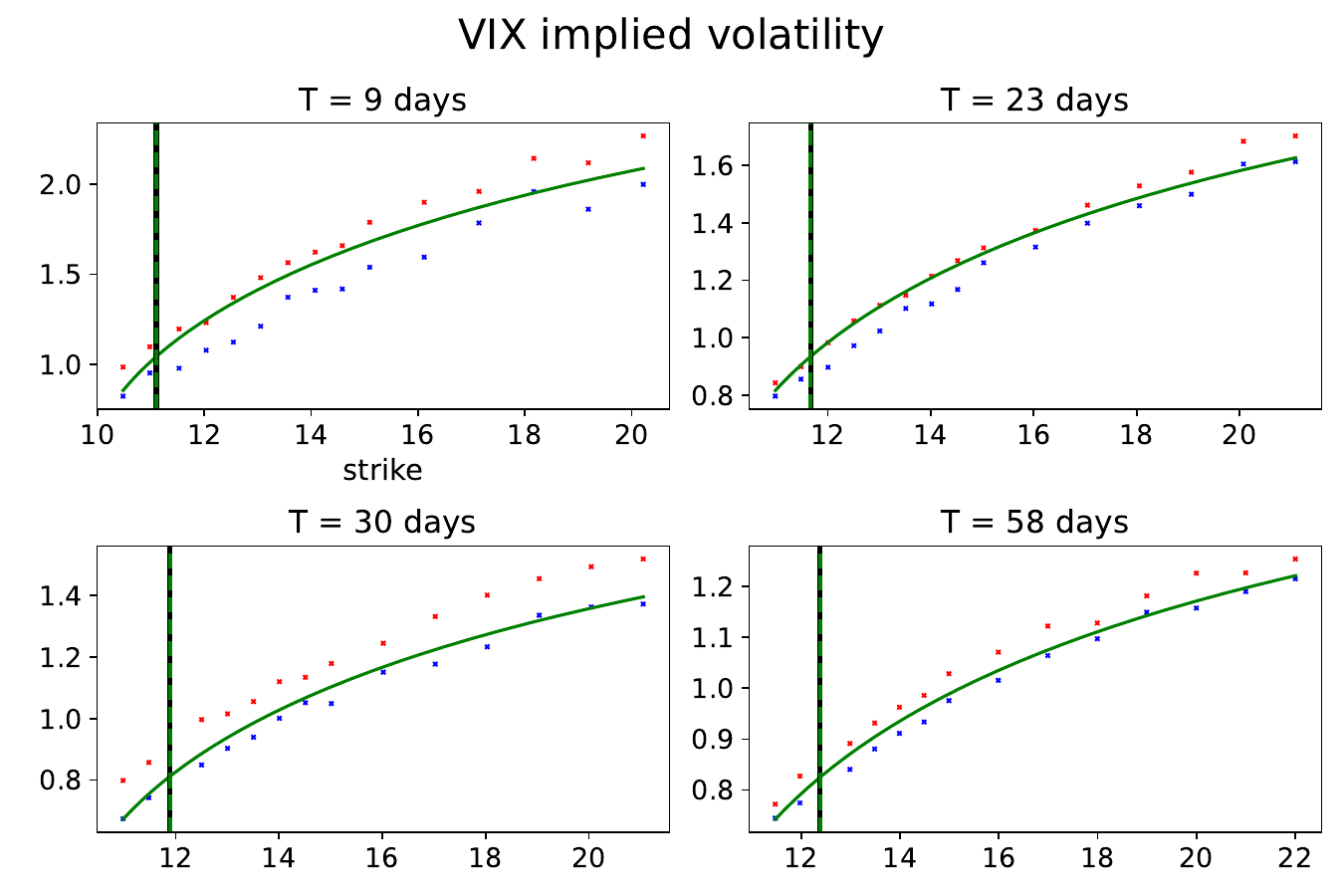}
    \caption{SPX \& VIX smiles (bid/ask in blue/red) and VIX Futures (vertical black lines) jointly calibrated with our conventional continuous stochastic volatility model (green lines) {produced by the exponential kernel $K^{exp}$}, 23 October 2017.}
    \label{fig:Intro}
  \end{figure}
  
  More precisely, our methodology and contributions are summarized as follows:
  
  \paragraph{Gaussian polynomial volatility models.} First, we introduce in Section~\ref{gpvm} a general class of \textit{Gaussian polynomial volatility} models in which the SPX spot price takes the form 
$$    \frac{dS_t}{S_t} = \sigma_t \left(\rho dW_t + \sqrt{1-\rho^2} dW_t^{\perp}\right),$$
where $(W,W^{\perp})$ is two-dimensional  Brownian motion. The SPX spot price $S$ is correlated with the volatility process  $\sigma$ which is, up to a normalizing {deterministic term}, defined as a polynomial function $p(X)$ of a Gaussian Volterra process $X$ in the form
$$ X_t = \int_0^t K(t-s)dW_s,$$
for some locally square-integrable kernel $K$.  The choice of the kernel introduces a good deal of flexibility in the modeling of the volatility process, such as rough volatility \cite{abi2022characteristic, abi2019multifactor, abi2019affine, bayer2016pricing, bennedsen2016decoupling, el2019characteristic,  gatheral2018volatility, 
 gatheral2020quadratic} for singular fractional kernels of the form $K(t)\sim t^{H-1/2}$ with $0<H{\leq}1/2$, or the log-modulated kernel that extends the fractional kernel for the case $H=0$, see \cite{bayer2021log}; path-dependent models with non-singular kernel such as the shifted fractional kernel  $K(t)\sim (t+\varepsilon)^{H-1/2}$; exponential kernels $K(t)\sim e^{-\lambda t}$ for which $X$ is a (Markovian) Ornstein-Uhlenbeck process or weighted sums of exponentials \cite{abi2019lifting,abi2019markovian,cuchiero2020generalized,harms2019affine}, refer to Table~\ref{tablekernels} below.    We will compare the performance of these different kernels on the joint calibration problem. Although it is difficult to decouple the impact of the different input parameters of the model, it turns out, that the choice of $K$ has a major impact on the ATM-skew of the implied volatility of the SPX and the level of the implied volatility of the VIX.  While the choice of the polynomial function $p$ has a prominent impact on the shape of the VIX smile. Taking $p$ a polynomial of order $5$ (and higher) allows us to reproduce the upward slope of the VIX smile.  

\paragraph{Generic, fast, and accurate pricing via quantization and Neural Networks.}
Second, to ensure a fair comparison between the calibrated models with different kernels {across 10 years of daily joint implied volatility surfaces}, we develop a generic unified method that applies to any \textit{Gaussian polynomial volatility model} for pricing SPX \& VIX derivatives efficiently and accurately.  The method is based on functional quantization and Neural Networks. The tractability of the quantization approach highly relies on the Gaussian nature of $X$ combined with the polynomial form of the volatility process $\sigma$. More precisely: 

\begin{itemize}
    \item \textbf{Fast pricing of VIX derivatives via Quantization}: we develop in Section~\ref{S:VIXpricing}  a functional quantization approach for computing VIX derivatives in our class of \textit{Gaussian polynomial volatility models}. 
    
    When computing expectations in the form of $\mathbb{E}[F(Y)]$ where no closed-form solution is available, a fast alternative to Monte Carlo is quantization. 
    The idea is to approximate the random variable $Y$ with a discrete random variable $\widehat Y$  to compute efficiently the (conditional) expectations of suitable functionals of $Y$. Quantization was first developed in the 1950s for signal processing { \cite{gersho2012vector,graf2000foundations}} and more recently has been studied for applications in numerical probability  {\cite{pages2003optimal}} and mathematical finance {\cite{pages2005functional,pham2004optimal}}. We will exploit the Gaussian nature of the process $X$ to develop a functional quantization approach. 
    
    A first attempt to use functional quantization for VIX Futures in the context of the rough Bergomi model appears in \cite{bonesini2021functional}. Unfortunately, the method is not precise enough in practice, especially for the fractional kernel with small values of $H$ even with a lot of quantization trajectories, see \cite[Figure 3]{bonesini2021functional} where the number of quantized points was pushed as far as $N=1,000,000$ but the approximated values for VIX Futures are still well-off the correct values,  see also Figure~\ref{Fig:momentmatching1} below.   It is well known that the convergence of the quantization for fractional processes is very slow of order $1/(\log N)^{H}$, see  \cite{dereich2006high}. 
    
    {Using a crucial \textit{moment-matching trick}, see \eqref{eq:momentmatching}, we can make functional quantization usable in practice by achieving very accurate results for both VIX Futures prices and VIX option smile with only a couple of hundreds of quantization points, even for fractional processes with very low  values of $H$.}

    \item \textbf{Fast pricing of SPX options via Neural Networks with Quantization hints:} 
    
    {In a first step in Section~\ref{S:Quantization}, we extend the previous quantization ideas to quantize SPX. However, the quantization is more delicate whenever $\rho \neq 0$ since it involves the quantization of the stochastic Itô integral $\int_0^t \sigma_s dW_s$. It is well-known since the work of Wong and Zakai {\cite{wong1965convergence}} that the  approximation  $\int_0^t \sigma_s d\widehat W_s$, where $\widehat W$ is some smooth approximation of the Brownian motion, will converge towards the Stratonovich stochastic integral defined by  \begin{align}
\int_0^t \sigma_s \circ dW_s := \int_0^t \sigma_s dW_s + \frac 12 \langle \sigma, W \rangle_t. 
\end{align}
There is an issue whenever the process $\sigma$ is not a semimartingale and has infinite quadratic variation, which is the case for the fractional kernel with $H<1/2$: the quadratic covariation $\langle \sigma, W \rangle$ explodes. 

To solve this issue, we subtract a diverging term, see \eqref{eq:renormalization}, to recover convergence, in the spirit of renormalization theory {\cite{hairer2014theory}}  and the approach in  \cite[Theorem 1.3]{bayer2020regularity}, combined with another  \textit{moment-matching} trick. Once again the \textit{moment-matching} trick is used to improve the accuracy. 

Unfortunately, quantization results for SPX degrade (at a slower rate) as $H$ goes to zero for the SPX derivatives. We therefore develop an approach in Section~\ref{S:NN} with Neural Networks acting as a corrector to the quantization points for the SPX. The Neural Networks approach in our paper has a low input dimension (strikes and the input curve are not part of the Neural Networks' input) and preserves the interpretability by directly modeling the joint density of $\log(S)$ and $\sigma$. It also significantly improves the SPX derivative pricing compared to Monte Carlo simulation, while being extremely fast. 
}
\end{itemize}

 {\paragraph{Extensive empirical study.} Our final contribution is an
extensive empirical joint calibration study detailed in  Section~\ref{S:empiricalcalibration}. A total of 1,422 days of SPX \& VIX joint implied volatility surfaces between August 2011 to September 2022 were calibrated. Interestingly, the \textit{conventional one-factor Markovian continuous stochastic volatility model} outperforms, in all market conditions, its rough and
 non-rough path-dependent counterparts, with the same number of calibrated parameters. A possible explanation for this performance lies in the unconstrained values of $H$ that can be pushed below zero once calibrated, something not possible for the rough fractional kernels. }

\paragraph{Outline of the paper.}  To ease the reading, and to satisfy not-so-patient readers who are (more) interested in our main findings regarding the empirical joint calibration,  we chose to present our empirical performance comparison between different models for the joint calibration problem in Section~\ref{S:empiricalcalibration}, right after the introduction of our class of \textit{Gaussian polynomial volatility} models in Section~\ref{gpvm}.  Our generic fast pricing methods are postponed to later sections:   
Section \ref{S:VIXpricing} develops a generic, fast, and accurate pricing method for VIX derivatives in our class of models based on functional quantization, and Section \ref{S:SPXpricing} extends the previous quantization approach and combines it with  Neural Networks to obtain a generic, fast and accurate pricing method for SPX derivatives.  
{Finally, Appendix \ref{A:formulae} collects some integral formulae, Appendix \ref{B:jc} contains additional 
calibration graphs, and Appendix \ref{C:nn_rss} showcases the Neural Network's capability to reduce Monte Carlo noise}.

\section{Gaussian polynomial volatility models}\label{gpvm}
We define the class of Gaussian polynomial volatility models under a risk-neutral measure as follows. We fix a filtered probability space $(\Omega,\mathcal F, (\mathcal F_t)_{t\geq 0}, \mathbb Q)$ satisfying the usual conditions and supporting a two-dimensional Brownian motion $(W,W^{\perp})$. For $\rho \in [-1,1]$, we set 
\begin{equation}
B=\rho W + \sqrt{1-\rho^2} W^{\perp},
\end{equation}
which is again a Brownian motion. 

\paragraph{The model.} The dynamics of the stock price $S$ are assumed to follow a stochastic volatility model such that the volatility process $\sigma$ is given by a polynomial {(possibly of infinite degree)} of a Gaussian Volterra process $X$
\begin{equation}\label{polynomial_model}
  \begin{aligned}
    \frac{dS_t}{S_t} &= \sigma_tdB_t, \quad  S_0 >0,\\
    \sigma_t &= \sqrt{\xi_0(t)}\frac{p(X_t)}{\sqrt{\mathbb{E}\left[p(X_t)^2\right]}}, \hspace{0.5cm} p(x) = \sum_{k= 0}^{M} \alpha_k x^k,\\
      X_t &=  \int_0^tK(t-s)dW_s,
  \end{aligned}
  \end{equation}
  for some $M \in \mathbb N$ possibly infinite, real coefficients $(\alpha_k)_{k=0,\ldots, M}$, a non-negative square-integrable kernel  $K \in L^2([0,T],\mathbb R_+)$ and input curve $\xi_0 \in L^2([0,T],\mathbb R_{{+}})$ for any $T>0$, with the convention that $0/0 = 1$.  In particular,  $X$ is a Gaussian process such that $\mathbb{E}\left[X_t^2\right] = \int_0^t K(s)^2ds<\infty$, for all $t\geq 0$. Note that $X$ is not necessarily Markovian or a semi-martingale.   We will be chiefly interested in the performance of our class of model for the joint SPX \& VIX calibration problem for four kernels summarized in Table~\ref{tablekernels}. 
  
{\begin{table}[h!]
{
		\centering  
		\resizebox{\textwidth}{!}{\begin{tabular}{c c  c c c} 
				\hline
				Kernel &   $K(t)$ & Domain of H & Semi-martingale & Markovian \\
				\hline 
			Fractional	$K^{frac}$ & $t^{H-1/2}$ & $(0,1/2]$ & \textcolor{red0}{\xmark}   & \textcolor{red0}{\xmark}  \\
		Log-modulated		$K^{log}$ & $t^{H-1/2} (\theta \log(1/t)\vee 1)^{-\beta}$ & $[0,1/2]$ & \textcolor{red0}{\xmark}   & \textcolor{red0}{\xmark}  \\
	Shifted fractional			$K^{shift}$ 	 & $(t+\varepsilon)^{H-1/2}$ & $(-\infty,1/2] $ & \textcolor{blue0}{\cmark}   & \textcolor{red0}{\xmark}  \\
Exponential				$K^{exp}$ & $\varepsilon^{H-1/2} e^{-(1/2-H)\varepsilon^{-1} t}$ & $(-\infty,1/2]$ & \textcolor{blue0}{\cmark}   & \textcolor{blue0}{\cmark}  \\
				\hline 
		\end{tabular}}
		\caption{The different kernels $K$ considered in this paper and the properties of their corresponding process $X_t=\int_0^t K(t-s)dW_s$; $\varepsilon>0$, $\theta>0$ and $\beta>1$.}
		\label{tablekernels} 
}
\end{table}}

The curve $\xi_0$  allows to match certain term structures observed on the market. For instance, the normalization  $\sqrt{\mathbb{E}\left[p(X_t)^2\right]}$ allows $\xi_0$ to match the market forward variance curve since 
\begin{align}\label{eq:fwdvarcalib}
    \mathbb E\left[ \int_0^t \sigma_s^2 ds \right] = \int_0^t \xi_0(s)ds, \quad t\geq 0. 
\end{align}

This family of models captures several well-known models already existing in the literature, such as a particular instance of the Volterra Stein-Stein model of \cite{abi2022characteristic} for $M=1$, $\alpha_0=0$ and $\alpha_1 = 1$; and the Volterra Bergomi model of \cite{bayer2016pricing} in the case $M=\infty$ and $a_k =  \frac{1}{2^k k!}, k\geq 0$ so that $\sigma_t^2 = \sqrt{\xi_0(t)}\exp(X_t - \frac 1 2 \mathbb E[X_t^2] )$. Except for the Volterra Stein-Stein class of models where Fourier inversion techniques can be applied thanks to an explicit expression of the characteristic function, as shown in \cite{abi2022characteristic},  pricing is usually slow in these models and only done using  Monte-Carlo simulation. We will develop in Sections~\ref{S:VIXpricing} and \ref{S:SPXpricing} a generic, efficient, and accurate method for our class of Gaussian polynomial volatility models exploiting the Gaussian nature of the driving process $X$ combined with the polynomial form of the volatility process $\sigma$.

\paragraph{The forward variance process.} The forward variance process $\xi_t(u):=\mathbb{E}\left[\sigma_u^2\mid \mathcal{F}_t\right]$ can be computed explicitly for the Gaussian polynomial volatility model. First, we fix $t\leq u$ and rewrite $X$ as
  \begin{equation}\label{eq6}
X_u = \underbrace{\int_0^t K(u-s)dW_s}_{Z_{t}^u} + \underbrace{\int_t^u K(u-s)dW_s}_{G_{t}^{u}}, 
  \end{equation}
 then, setting 
 $$g(u)=\mathbb E[p(X_u)^2],$$ 
 we have that
  \begin{equation}\label{eq7}
    \xi_t(u) = \mathbb{E}\left[\sigma_u^2\mid \mathcal{F}_t\right] = \frac{\xi_0(u)}{g(u)} \mathbb{E}\left[\left(\sum_{k=0}^M\alpha_k X_u^k\right)^2 \Mid \mathcal{F}_t \right] = \frac{\xi_0(u)}{g(u)} \mathbb{E}\left[\sum_{k=0}^{2M}(\alpha *\alpha)_k X_u^k \Mid \mathcal{F}_t\right],
  \end{equation}
where $(\alpha *\alpha)_k = \sum_{j=0}^{k}\alpha_j\alpha_{k-j}$ is the discrete convolution. Using the Binomial expansion, we can further develop the expression for $\xi_t(u)$ in terms of $Z^u$ and $G^u$ to get
  \begin{equation}\label{fwd_var_process}
    \xi_t(u)=\frac{\xi_0(u)}{g(u)}\sum_{k=0}^{2M}\sum_{i=0}^{k} (\alpha *\alpha)_k \binom{k}{i}(Z_t^u)^{k-i}\mathbb{E}\left[(G_t^u)^i \right],
  \end{equation}
 where we used the fact that $Z_t^u$ is $\mathcal F_t$-measurable and that  $G^u_t$  {is independent of $\mathcal F_t$},  with $k!/((k-i)!i!)$ the binomial coefficient. Furthermore, $G^u_t$ is a Gaussian random variable with mean $0$ and variance $\int_0^{u-t} K(s)^2ds$, the  moments $\mathbb{E}\left[(G_t^u)^i \right]$ can be computed explicitly
\begin{align}\label{eq:momentgaussian}
  \mathbb{E}\left[(G_t^u)^i \right] =
    \begin{cases}
      0 & \text{if $i$ is odd}\\
      \left(\int_0^{u-t}K^2(s)ds\right)^{\frac{i}{2}}(i-1)!! & \text{if $i$ is even}\\
    \end{cases}       
\end{align}
with $i!!$ the double factorial.

We now explicit the dynamics of $(\xi_{t}(u))_{t\in [0,u]}$. By construction for fixed $u$, the process $(Z_t^u)_{t\in [0,u]}$ is a martingale with dynamics 
\begin{align*}
dZ_t^u = K(u-t) dW_t.
\end{align*}
Similarly, $\xi_{\cdot}(u)$ is a martingale on $[0,u]$, so that its part in $dt$ is zero.  
An application of  It\^o's formula leads to the following dynamics for the process $\xi_{\cdot}(u)$
  \begin{equation}\label{eq9}
    d \xi_t(u) = \frac{\xi_0(u)}{g(u)} K(u-t)\sum_{{k=0}}^{2M}{ \sum_{i=0}^{{k}}} (\alpha *\alpha)_k \binom{k}{i}(k-i)(Z_t^u)^{k-i-1} \mathbb{E}\left[(G_t^u)^i  \right]dW_t.
  \end{equation}

 \paragraph{An explicit expression for the VIX.} One major advantage of our class of Gaussian polynomial volatility models is an explicit expression of the VIX. {In continuous time, the VIX can be expressed as
  \begin{equation}
    \mbox{VIX}_T^2 = -\frac {2}{\Delta} \mathbb{E}\left[\log(S_{T+\Delta}/S_T)\mid \mathcal{F}_T \right] \times 100^2,
  \end{equation}
with $\Delta= 30$ days. Combining the above expression with \eqref{polynomial_model} and \eqref{fwd_var_process}, we have an explicit expression of the VIX in the Gaussian polynomial volatility model
\begin{align}
    \mbox{VIX}_T^2 &= \frac{100^2}{\Delta} \mathbb{E}\left[\int_T^{T+\Delta}\sigma_u^2du \Mid \mathcal{F}_T\right] = \frac{100^2}{\Delta} \int_T^{T+\Delta}\xi_T(u) du\\
    &=\frac{100^2}{\Delta} \sum_{k=0}^{2M}\sum_{i=0}^{k} (\alpha *\alpha)_k \binom{k}{i}\int_T^{T+\Delta} \frac{\xi_0(u)}{g(u)}\mathbb{E}\left[(G_T^u)^i \right] (Z_T^u)^{k-i}du, \label{eq:VIXclosed}
\end{align}
with
\begin{align}\label{eq:ZuT}
Z_T^u = \int_0^T K(u-s)dW_s, \quad G_T^u = \int_T^{T+\Delta} K(u-s)dW_s, \quad   T \leq u.
\end{align}
}
 
Recall that the moments $\E[(G^u_T)^i]$ are given explicitly by \eqref{eq:momentgaussian}.

\paragraph{The SPX ATM-Skew and the restriction of the coefficients $(\alpha_k)_{k\geq 0}$.}  Even if there is no theoretical restriction on the domain of $\{\alpha_k\}_{k \leq M}$, it would still be desirable for the SPX at-the-money (ATM) skew to be controlled by the sign of $\rho$, as in all other usual stochastic volatility models.  
Using the Guyon-Bergomi expansion of implied volatility in terms of small volatility of volatility \cite{bergomi2012stochastic} at first order, the  SPX ATM skew $\mathcal{S}_T$ is defined by the quantity
\[
\mathcal{S}_T = \frac{\partial \hat \sigma_{spx}\left(k,T\right)}{\partial k} \mid_{k=0},  \quad k = \log(K/S_0),
\]
where $\hat \sigma_{spx}\left(k,T\right)$ is the implied volatility of SPX. $\mathcal{S}_T$ has the sign of the integrated spot-variance covariance function $C^{X \xi}$ given by 
  \begin{align}\label{eq5}
    C^{X\xi} &:= \int_0^T\int_t^T {\E \left[\frac{d\langle \log S,\xi_{\cdot}(u)\rangle_t} {dt} \right]}  dudt \\
     &=\rho   \int_0^T\int_t^T \frac{\xi_0(u)}{g(u)}\sqrt{\frac{\xi_0(t)}{g(t)}} {K(u-t)} \\
     & \times \sum_{k=0}^{2M}\sum_{i=0}^{k} \sum_{j=0}^{M} (\alpha *\alpha)_k \alpha_j \binom{k}{i}{ (k-i)}\mathbb{E}\left[(G_t^u)^i\right] \mathbb{E}\left[X_t^j (Z_t^u)^{k-i-1}\right]du dt
  \end{align}
  where the second equality follows from  \eqref{eq9}. 

The expression of $C^{X\xi}$ for the Gaussian polynomial volatility model thus requires the computation of $\mathbb{E}\left[X_t^p (Z_t^u)^q\right]$, which can be computed via Isserlis' Theorem \cite{isserlis1918formula}. 

\begin{theorem}[Isserlis' Theorem] If $\left(U_1, \ldots, U_n\right)$ is a zero-mean multivariate normal random vector, then
\[
\mathbb{E}\left[U_1 U_2 \cdots U_n\right]=\sum_{p \in P_n^2} \prod_{\{i, j\} \in p} \mathbb{E}\left[U_i U_j\right],
\]
where the sum is over all the pairings of $\{1, \ldots, n\}$, i.e. all distinct ways of partitioning $\{1, \ldots, n\}$ into pairs $\{i, j\}$, and the product is over the pairs contained in $p$. (When $n$ is odd,  there does not exist any pairing of $\{1,\ldots, n\}$ so that $\mathbb{E}\left[U_1 U_2 \cdots U_n\right] = 0$.)
\end{theorem}

Thus the computation of $\mathbb{E}\left[X_t^p (Z_t^u)^q\right]$ essentially comes down to computing the following quantities
\[
\begin{aligned}
& \mathbb{E}\left[X_t X_t\right] = \int_0^t K(s)^2ds,\\
& \mathbb{E}\left[Z_t^u Z_t^u\right] = \int_{u-t}^u K(s)^2ds,\\
& \mathbb{E}\left[X_t Z_t^u\right] = \int_0^t K(t-s)K(u-s)ds.\\
\end{aligned}
\]

Note that all the quantities above are non-negative so that $\mathbb{E}\left[X_t^p (Z_t^u)^q\right]$ is non-negative. Therefore a sufficient (and simple) condition for the sign of the ATM skew to be the same as $\rho$ is by setting $\alpha_k \geq 0$,  for all $k \leq M$. 

\section{Joint SPX \& VIX calibration: the empirical study}\label{S:empiricalcalibration}

We carried out joint calibration on SPX \& VIX implied volatilities, together with VIX Futures using all four kernels in Table~\ref{tablekernels} for every $\mbox{2}^{nd}$ day between August 2011 to September 2022. That is a total of 1,422 days of SPX \& VIX joint implied volatility surfaces. The VIX is calibrated up to maturity $T=\mbox{2 months}$, and the SPX is calibrated up to maturity $T+\Delta$, i.e.~3 months. Market data was purchased from the CBOE website \url{https://datashop.cboe.com/}.

The objective of joint calibration is to minimize the error between SPX \& VIX implied volatility, together with the VIX Futures prices output from the model and those observed on the market. This amounts to solving the following optimization problem involving the sum of root mean squared error (RMSE)
\begin{equation}\label{calibration_error}
  \begin{aligned}
    &\min_{\Theta} \Bigg\{  c_1 \sqrt{\sum_{i,j}  \Big(\widehat \sigma_{spx}^{\Theta}(T_i,K_j) -\widehat \sigma_{spx}^{mkt}(T_i,K_j) \Big)^2} +    c_2 \sqrt{\sum_{i,j} \Big(\widehat \sigma_{vix}^{\Theta}(T_i,K_j) -\widehat \sigma_{vix}^{mkt}(T_i,K_j) \Big)^2}\\ 
    &\quad \quad \quad \quad \quad  +  c_3 \sqrt{\sum_{i} \Big(F^{\Theta}_{vix}(T_i)-F^{mkt}_{vix}(T_i) \Big)^2}\Bigg\}.
  \end{aligned}
\end{equation}

Here, $\widehat\sigma_{spx}^{mkt}(T_i,K_j)$, $\widehat\sigma_{vix}^{mkt}(T_i,K_j)$ represent market SPX \& VIX implied volatility with maturity $T_i$ and strike $K_j$. $F^{mkt}_{vix}(T_i)$ is the market VIX Futures price maturing at $T_i$. $\widehat\sigma_{spx}^{\Theta}(T_i,K_j)$, $\widehat\sigma_{vix}^{\Theta}(T_i,K_j)$ and $F^{\Theta}_{vix}(T_i)$ represent the same instruments, but coming from the Gaussian polynomial volatility model with parameters denoted collectively as $\Theta$  
 and will be detailed in \eqref{eq:THeta} below. The coefficients $c_1$, $c_2$, and $c_3$ are some positive numbers used to assign different weights to the errors in SPX \& VIX implied volatility and VIX Futures price. We chose $c_1 = 1,c_2 = 0.1, c_3 = 0.5$ for all four kernels. {Of course, these weights can be chosen differently, e.g.~based on liquidity and maturity, etc.}
 
We recall  that the implied volatility of a call option is calculated by inverting  the Black and Scholes  formula, that is, for a given call price $C_0 (K,T)$ with strike $K$ and maturity $T$, we find the unique $\sigma_{K,T}$ such that 
  \begin{equation}
    C_0 (K,T) = F_0^T\mathcal{N}(d_1)-K\mathcal{N}(d_2)
  \end{equation}
with
  \begin{equation}
d_1 = \frac{\log\left(F_0^T/K\right)+ \frac 1 2  \sigma_{K,T}^2 T}{ \sigma_{K,T} \sqrt{T}}, \quad d_2 = d_1 - \sigma_{K,T} \sqrt{T},
  \end{equation}
where  $\mathcal{N}(x)= \int_{-\infty}^x e^{-z^2/2}dz/\sqrt{2\pi}$ is the cumulative density function of the standard Gaussian distribution and $F_0^T$ denotes the Futures price of the index: $F_0^T=\E\left[S_T\right]=S_0$ for the SPX in our setting \eqref{polynomial_model} and  $F_0^T=\E\left[\mbox{VIX}_T\right]$ for the VIX.

{To speed up the joint calibration, we applied functional quantization for fast pricing of VIX derivatives, detailed in Section \ref{S:VIXpricing}, and functional quantization with Neural Networks for fast pricing of SPX derivatives, detailed in Section \ref{S:SPXpricing}. The optimization problem in \eqref{calibration_error} is solved using the SciPy optimization library in Python.}

\subsection{Choice of the polynomial function \texorpdfstring{$p$}{Lg}}
We first comment on the choice of the polynomial function $p$ in \eqref{polynomial_model}. Based on our numerical experiments, we will take  $p$ a polynomial of order $5$ (i.e.~$M=5$ in \eqref{polynomial_model}) in the following form
\begin{align}
    p(x) =\alpha_0 + \alpha_1 x + \alpha_3 x^3 + \alpha_5 x^5, \quad x\in \R,
\end{align}
with $\alpha_0,\alpha_1,\alpha_3,\alpha_5 \geq 0$.
The high degree allows us to reproduce the upward slope of the VIX smile. Setting $\alpha_2=\alpha_4 =0$ allows us to reduce the number of parameters to calibrate.

\subsection{Choice of parameters}

We now comment on the choice of parameters of the four kernels of Table~\ref{tablekernels} used for joint calibration. First, we fix $\varepsilon=1/52$ for the shifted fractional kernel $K^{shift}$ and exponential kernel $K^{exp}$. Next, we also set $\theta=0.1$ for the log-modulated kernel $K^{log}$ as suggested by \cite{bayer2021log} to further reduce the dimension of parameters space. Thus, there are only 6 calibratable parameters plus the input curve $\xi_0(\cdot)$ for the kernels $K^{frac}, K^{shift}, K^{exp}$, namely
\begin{align}\label{eq:THeta}
\Theta := \{\alpha_0, \alpha_1, \alpha_3, \alpha_5, \rho, H \}
\end{align}
and an extra parameter $\beta$ for kernel $K^{log}$. Numerical experiments show no significant adverse impact on the joint calibration quality by narrowing the choice of parameters as we suggested.

For the treatment of the input curve $\xi_0(\cdot)$, we first strip the forward variance curve {of the market} using the celebrated formula by \citet{carr2001towards}
and then pass a cubic spline through the square root of the forward variance curves to enforce positivity across time{, recall \eqref{eq:fwdvarcalib}}. During calibration, the spline nodes of the input curve are adjusted as necessary to fit the level of implied volatilities for SPX \& VIX, as is usually done in forward variance type of models, see \cite{bayer2016pricing}.

{ We now summarize the steps involved when calibrating our models to market data:

\begin{enumerate}
    \item Extract the forward variance swap rate $x_i$ use the replication formula of Carr and Madan \cite{carr2001towards} for $(T_i, T_{i+1})$, with $T_i$ representing each maturity in the dataset,
    \item Approximate $\sqrt{\xi_0(t)}$, for  $t\in [T_i, T_{i+1}[$ by passing through a cubic spline at each node $(\frac{T_i+T_{i+1}}{2}, \sqrt{x_i})$,
    \item Given $\xi_0(t)$ and $\Theta$, compute the SPX implied volatility, VIX Futures and VIX implied volatility using functional Quantization and Neural Networks described in Sections \ref{S:VIXpricing} and \ref{S:SPXpricing} below for each maturity and compute the error function in \eqref{calibration_error},
    \item Using an optimization algorithm of choice, adjust $\Theta$ and $x_i$ to reduce the calibration error. We recommend keeping the calibrated $x_i$ within 50\% of the initial value so that the expected integrated variance at each $T_i$ does not deviate far from the market data,
    \item Repeat step 4 until the algorithm reaches a satisfactory local optimum.
\end{enumerate}
}

\subsection{Impact of kernel \texorpdfstring{$K$}{Lg} on joint calibration: an empirical comparison}

The choice of kernels plays a crucial role in the model's capability of jointly fitting the SPX \& VIX smiles. 
We will consider successively the four kernels of Table~\ref{tablekernels}. 

\paragraph{The fractional kernel} $K^{{frac}}(t)=t^{H-1/2}$, with $H\in (0,1/2]$,
 taken as the starting point, is extensively used in recent literature on rough volatility \cite{abi2022characteristic, bayer2016pricing, el2019characteristic}. 
Separate calibration of SPX/VIX appears to be satisfactory, however there are inconsistencies in the value of $H$ between the two indices. To produce the steep VIX ATM skew and lower level of VIX implied volatility, the calibrated $H$ is very close to zero  (similar to that of quadratic rough Heston model in  \cite{rosenbaum2021deep} where $H=0.01$). This is problematic for the SPX due to the `vanishing skew' phenomena as $H \rightarrow 0$, observed in \cite{forde2020rough} that also plagues models such as the rough Bergomi model. 

Despite pushing $\rho$ to the boundary value $-1$ in most days (which should increase the SPX ATM skew in stochastic volatility models) as shown in Figure~\ref{frac_param_history} of Appendix \ref{B:param_history}, the joint calibrated SPX ATM skew is too flat compared to the market data. The VIX implied volatility produced by the model is generally too high and does not have enough ATM skew, see Figure~\ref{frac_jc1} of Appendix \ref{B:jc_kernels_compare}. One can try improving the VIX fit by pushing $H$ closer to zero, but this will further flatten the SPX ATM skew.

\paragraph{The log-modulated kernel}
$K^{log}(t)={{t^{H-1/2}}}\text{max}(\theta \log(1/t),1)^{-\beta}, H \in [0,1/2], \theta>0, \beta>1$ is proposed by \cite{bayer2021log}, where $X$ is well-defined even when $H=0$. 
 In \cite[Figure 1.1]{bayer2021log}, it is shown that this kernel is capable of resisting the SPX ATM skew flattening suffered by the fractional kernel $K^{frac}$. The calibrated $H$ appears to be just about zero during normal market conditions with a slightly less saturated $\rho$, see Figure~\ref{log_m_param_history} of  Appendix \ref{B:param_history}. The joint fit seems to be slightly better than that of the fractional kernel $K^{frac}$, see Figure~\ref{log_m_jc1} of  Appendix \ref{B:jc_kernels_compare}.
 
However, the joint calibration results are still not satisfactory. It seems $H$ needs to go even further below zero (something impossible for the log modulated kernel $K^{log}$ and the fractional kernel $K^{frac}$) to produce the steep VIX ATM skew along with a lower level of VIX implied volatility. This motivates the use of the following shifted fractional kernel $K^{{shift}}$.

\paragraph{The shifted fractional kernel} $K^{shift}(t) = (t+\varepsilon)^{H-1/2}$, for some $\varepsilon > 0$, is now well-defined and 
 locally square-integrable for any $ H \in (-\infty,1/2]$. The kernel is even smooth enough to make  $X$  a semi-martingale but not Markovian, see Proposition~\ref{P:semi} below.  The empirical joint calibration results are {considerably} better compared to the fractional kernel $K^{frac}$ and the log-modulated kernel $K^{log}$, see Figure~\ref{shift_jc1} of Appendix \ref{B:jc_kernels_compare}. The history of calibrated $H$ spends most of its time below zero (except for a brief moment during the 2020 Covid pandemic), with $\rho$ no longer saturated and averaging around -0.7, see Figure~\ref{shift_param_history} of Appendix \ref{B:param_history}. One issue with this kernel is that $X$ is not Markovian and is path-dependent.

\paragraph{The exponential kernel $K^{exp}$.} The choice of the exponential kernel, along with its particular parametrization is motivated as a proxy of the shifted fractional kernel that would yield Markovian dynamics for $X$. To see this, we recall the following representation of the fractional kernel as a Laplace transform
\begin{align*}
    t^{H-1/2} &=  c_H \int_{\R_+} e^{-xt}  {x^{-H-1/2}} dx,
\end{align*}
with  $c_H =1/ \Gamma(1/2-H)$ and $\Gamma(z):=\int_{\R_+} x^{z-1} e^{-x}dx$ the Gamma function for $z>0$. We note that such representation as a Laplace transform has been used in the literature to disentangle the (infinite-dimensional) Markovian structure of fractional processes \cite{abi2019markovian, abi2021linear, carmona1998fractional, cuchiero2020generalized, harms2019affine} and develop efficient numerical Markovian approximations \cite{abi2019lifting,   alfonsi2021approximation,  bayer2021markovian, harms2019strong, zhu2021markovian}. For fixed $\varepsilon >0$, the shifted fractional kernel then reads
\begin{align*}
    (t+\varepsilon)^{H-1/2} &=   \int_{\R_+} e^{-xt} \mu_{H,\varepsilon}(dx) {\mbox with } \quad    \mu_{H,\varepsilon}(dx):= c_H  e^{-x\varepsilon}{x^{-H-1/2}} dx .
\end{align*}
Now we choose $c_{H,\epsilon}$ and $\gamma_{H,\varepsilon}$  such that 
 the measure $\nu_{H,\varepsilon}(dx) = c_{H,\varepsilon} \delta_{\gamma_{H,\varepsilon}}(dx)$ satisfies
\begin{align*}
    \nu_{H,\varepsilon}(\R_+) =  \mu_{H,\varepsilon}(\R_+)  \quad \mbox{and} \quad      \int_{\R_+} x\nu_{H,\varepsilon}(dx) =  \int_{\R_+} x\mu_{H,\varepsilon}(dx),
\end{align*}
 with $\delta$ the Dirac measure.
This yields that 
\begin{align*}
    c_{H,\varepsilon} = \varepsilon^{H-1/2}  \quad \mbox{and} \quad     \gamma_{H,\varepsilon} = (1/2-H) \varepsilon^{-1}
\end{align*}
leading to the following exponential kernel 
\begin{align*}
K^{exp}(t) := \int_{\R_+} e^{-xt}\nu_{H,\varepsilon}(dx) = \varepsilon^{H-1/2} e^{-(1/2-H) \varepsilon^{-1} t}.
\end{align*}
For small values of $\varepsilon$ and $H$, this parametrisation gives a (Markovian) Ornstein-Uhlenbeck process $X$ in \eqref{polynomial_model} with a fast mean reversion of order $(1/2-H)\varepsilon^{-1}$ and a large vol-of-vol of order $\varepsilon^{H-1/2}$ with the following dynamics
\begin{align*}
dX_t = -(1/2-H) \varepsilon^{-1} X_t dt    + \varepsilon^{H-1/2} dW_t.
\end{align*}

{It is interesting to note that special cases of such parametrizations of conventional stochastic volatility models have already appeared in the literature but for a fixed value of $H$: for $H=0$ one recovers the fast regimes extensively studied by \citet{fouque2003multiscale}, see also   \cite[Section 3.6]{fouque2000derivatives}; setting $H=-1/2$ yields the parametrization studied by  \citet{mechkov2015fast} and establishes the link with jump models, see also \cite{abi2023reconciling, mccrickerd2019foundations}.}

Based on empirical results, the exponential kernel $K^{exp}$ 
produces the best joint fit compared to the other kernels while being the simplest (semi-martingale and Markovian). For SPX maturities up to 3 months and VIX maturities up to 2 months, the exponential kernel $K^{exp}$ can achieve remarkable fits, as shown in Figure \ref{fig:Intro} of implied volatility surfaces dated 23 October 2017, with calibrated parameters $\rho = -0.6997, H = -0.06939, (\alpha_0, \alpha_1, \alpha_3, \alpha_5) = (0.82695,0.84388,0.55012,0.03271)$.

The historical time series of joint calibration rooted mean square error (RMSE) in Figure~\ref{rmse_evol} and 
the distribution of the RMSE in Figure~\ref{rmse_dist} of Appendix \ref{B:rmse_dist} shows that the exponential kernel $K^{exp}$ outperforms other kernels for all market conditions for both SPX \& VIX fit.

{
  \begin{figure}[H]
    \centering
    \includegraphics[scale=0.35]{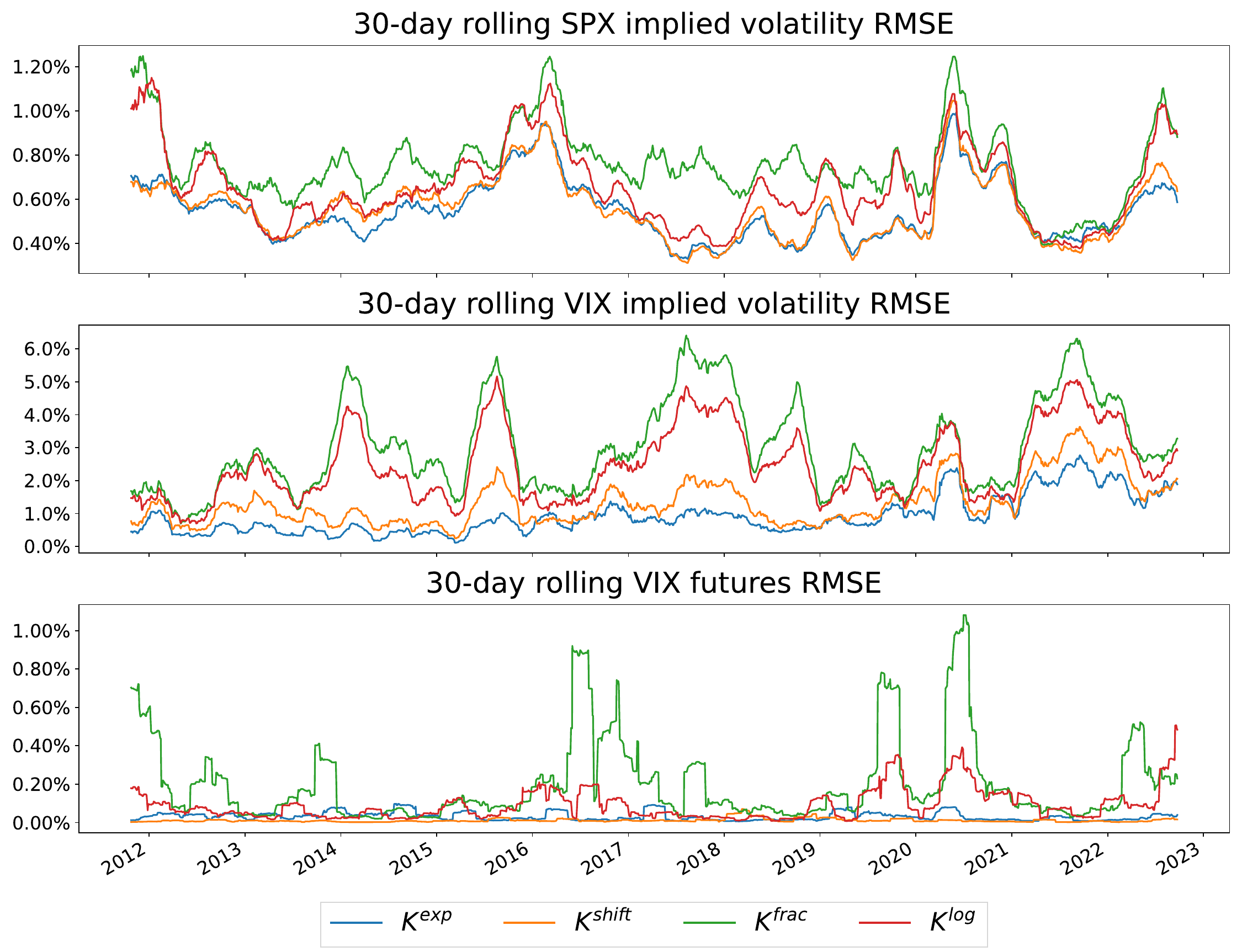}
    \caption{Time series of calibrated RMSE across different kernels between August 2011 and September 2022: the exponential kernel $K^{exp}$ outperforms other kernels in all market conditions.}
    \label{rmse_evol}
  \end{figure}
}

The evolution of jointly calibrated parameters $H$ and $\rho$ also appear to be stable over time in the case of the exponential kernel as shown in Figure \ref{exp_param_history}. {This further validates the robustness of the exponential kernel $K^{exp}$ to jointly fit SPX \& VIX implied volatilities}. Notice that $H$ and $\rho$ also appear to be negatively correlated to one another. We observe that $\rho$ is far from being saturated to $-1$ and $H$ is on average very small and dips below zero from time to time. The parameters $\rho$ and $H$ for the shifted fractional kernel $K^{shift}$ display a similar  trend, see Figure~\ref{shift_param_history} of Appendix \ref{B:param_history}.

{
  \begin{figure}[H]
    \centering
    \includegraphics[scale=0.45]{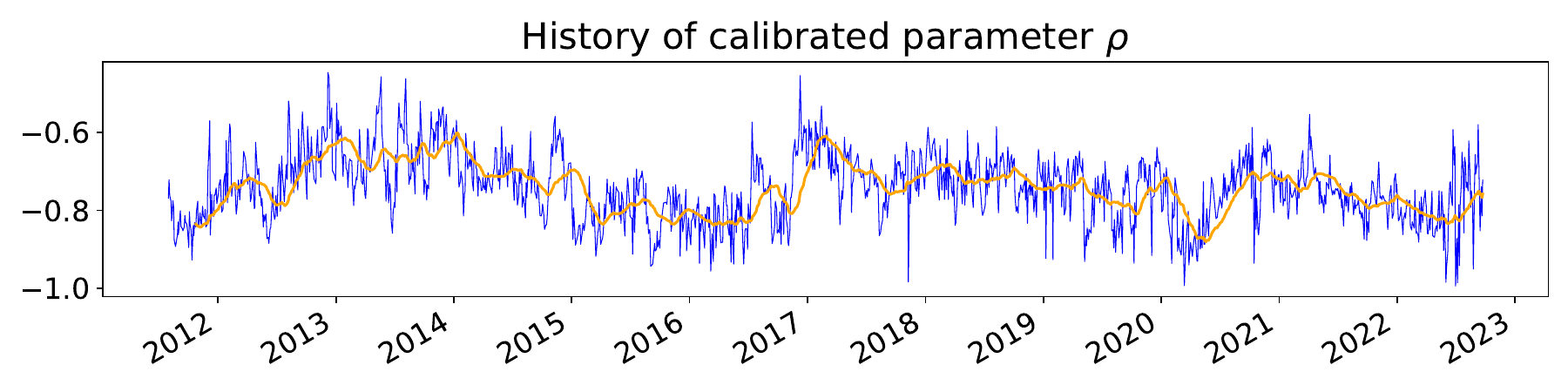}
    \includegraphics[scale=0.45]{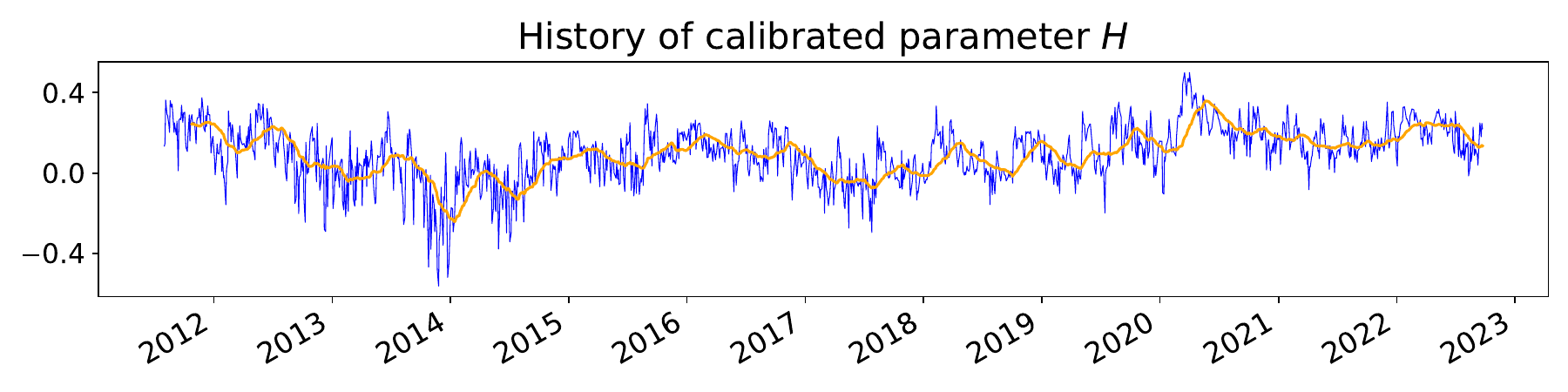}
    \caption{Evolution of the calibrated parameters $\rho$ and $H$ under the exponential kernel $K^{exp}$, the blue line is the actual value of the calibrated parameters in time, the orange line is the 30-day moving average.}
    \label{exp_param_history}
  \end{figure}
}

{For graphs on the joint calibration results by quantiles for the exponential kernel $K^{exp}$, the reader can refer to the Appendix \ref{exp_ker_quantile}}.

\section{Fast pricing of VIX derivatives via Quantization}\label{S:VIXpricing}

\subsection{\texorpdfstring{$L^2$}{Lg}-quantization of Gaussian random variables}\label{S:l2gaussian}

The idea {of quantization} is to approximate an $\mathbb{R}$-valued random variable $Y$ with a discrete random variable $\widehat Y$  in the computation of $\E [F(Y)]$ for some function $F$. Quantization provides a faster alternative to Monte Carlo when no closed-form expression of the expectation is available. We will concentrate on the case where $Y$ is Gaussian.

Formally, fix  $N\in \mathbb N$, for a given set of $N$-points  $\Gamma$ given by
\[
\Gamma = \{y_1,\ldots,y_{N}\} \subset \mathbb{R},  
\]
we consider the Borel partition $C = \{C_1,\ldots,C_{N}\}$  of $\mathbb{R}$ induced by $\Gamma$ using the nearest neighborhood projection
\[
C_{i} =  \left[y_{i-\frac{1}{2}}, y_{i+\frac{1}{2}}\right[, \quad y_{i\pm\frac{1}{2}}:=\frac{y_{i}+y_{i \pm 1}}{2}, \quad  y_0 = -\infty,\quad y_{N+1} = \infty, \quad i \in \{1,\ldots,N\}.
\]
The  $N$-quantizer $\widehat Y^{\Gamma}$ is defined by
  \begin{equation}
    \widehat Y^{\Gamma}(\omega) = \sum_{i=1}^{N} y_i 1_{C_i}(Y(\omega))
  \end{equation}
with the associated probability vector $p = \{p_1,\ldots, p_N\}$ such that
\begin{align}\label{eq:proba1d}
p_i = \mathbb{Q}(\widehat Y^{\Gamma} = y_i) = \mathbb{Q}(Y \in C_i) = \int_{ y_{i-\frac{1}{2}}}^{ y_{i+\frac{1}{2}}}d \mathbb{Q}_Y = \mathcal{N}(y_{i+\frac{1}{2}}) - \mathcal{N}(y_{i-\frac{1}{2}}),  \quad i \in \{1,\ldots, N\},  
\end{align}
with $\mathcal{N}(x)= \int_{-\infty}^x e^{-z^2/2}dz/\sqrt{2\pi}$ the cumulative density function of the standard Gaussian distribution. In particular, one can look for the $L^2$-optimal  $N$-quantizer $\widehat Y^{{\Gamma}*}$ with $\Gamma *$ given by
  \begin{equation}\label{L2_  quantizer_Y}
    {{\Gamma}*} = \arg\min_{\Gamma } \E\left[ \left(Y-\widehat Y ^{\Gamma}\right)^2 \right].
  \end{equation}

For standard Gaussian random variable $Y$, the existence and uniqueness of the $L^2$-optimal  $N$-quantizers have been established in \cite{pages2008quadratic}, along with a large family of distributions. The $L^2$-optimal $N$-quantizer is usually computed numerically, using either a zero search (Newton-Raphson gradient descent) or a fixed point procedure \cite{kieffer1982exponential, pages2003optimal}. Optimal quantizers for $N$ up to $5999$ have been computed offline in \cite{quantization_website} and are available online at \url{http://www.quantize.maths-fi.com/} in the format of pairs of $(y_i, p_i)$. 

From now on, we use $\widehat Y$ to denote the $L^2$-optimal $N$-quantizer $\widehat Y^{{\Gamma}*}$ to ease notations. It is now natural to consider the following approximations
  \begin{equation}
  \mathbb{E}[F(Y)] \approx \mathbb{E}[F(\widehat Y)]= \sum_{i=1}^{N} F(y_i)p_i.
  \end{equation}
Such ideas can be extended to approximate the Brownian motion  $W$  and the 
 VIX  in \eqref{eq:VIXclosed}. This is the object of the next subsections.

\subsection{\texorpdfstring{$L^2$}{Lg}-product functional quantization of Brownian motion}\label{qbm}
Let us now consider the  Brownian motion $(W_{t})_{0\leq{t}\leq{T}}$ defined under Section \ref{gpvm}. The idea of product functional quantization is to reduce the infinite dimension of the path space of $W$ into a finite $N$ number of paths. To do this, we start with the celebrated Karhunen–Loève decomposition of the Brownian motion $W$ in the form of
\begin{equation}\label{eq:KL}
W_{t}  = \sum_{k=1}^{\infty} \sqrt{\lambda_{k}}e_k{(t)}Y_k, {\quad 0 \leq t \leq T,}
\end{equation}
with $(Y_k)_{k\geq1}$ i.i.d.standard Gaussian and  
\[
e_{k}(t):=\sqrt{\frac{2}{T}} \sin \left(\pi(k-1 / 2) \frac{t}{T}\right), \quad \lambda_{k}:=\left(\frac{T}{\pi(k-1 / 2)}\right)^{2}, \quad k \geq 1.
\]
More precisely,  $(e_k, \lambda_k)$ are the $k^{th}$ pair of eigenfunctions in $L^2([0,T])$ and positive eigenvalues associated with the covariance kernel ${C_W}(t,s) = s\wedge t$
\[
\int_{0}^{T}C_{W}(t,s)e_k(s)ds = \lambda_k e_k(t), \quad k\geq 1, {\quad 0 \leq t  \leq T}.  \]
We note that  $(\lambda_k)_{k\geq 1}$ are decreasing and go to zero as $k$ goes to infinity.

 Using \eqref{eq:KL}, the product quantization is achieved by 1) truncating the infinite sum to a finite level $m$ and 2) quantizing the i.i.d standard Gaussian $Y_k$. For a given $N\in \mathbb N$, the product quantizer of the Brownian motion $W$ is defined as follows
  \begin{equation}\label{eq:hatW}
    \widehat{W}_{t}^{\left(N,m\right)}  = \sum_{k=1}^{m} \sqrt{\lambda_{k}}e_k{(t)}\widehat{Y}_{k}^{\left(N_{k}\right)},
  \end{equation}
where for each $k \in \{1,\ldots, m\} $, $\widehat{Y}_{k}^{\left(N_{k}\right)}$ is the $L^2$-optimal $N_k$-quantizer of the standard Gaussian $Y_k$ given in Section \ref{S:l2gaussian}  with $N_k$ number of quantized points $y^k = \{y^k_1,\ldots,y^k_{N_k}\}$. The product quantizer  $\widehat W^{(N,m)}$ has (at most) $N$ trajectories{, i.e.~$N_1 \times N_2 \times \ldots \times N_m\leq N$}.  More precisely, for a specific $\omega \in \Omega$, the trajectory of $\widehat W_t^{(N,m)}$ is defined as
  \begin{equation}\label{eq:bm_quant_traj}
    \widehat{W}_{t}^{\left(N,m\right)}(\omega)  = \sum_{k=1}^{m} \sqrt{\lambda_{k}}e_k{(t)}\sum_{\underline{i}\in \prod_{j=1}^{m}\{1,2,\ldots,N_j\}}y_{\underline{i}_k}^k 1_{\{(Y_1(\omega),\ldots,Y_m(\omega)) \in C_{\underline{i}} \}},
  \end{equation}
where $\underline{i}_k \in \{1,2,\ldots,N_k\} $ denotes the $k$-th element of the tuple $\underline{i}=(\underline{i}_1,\ldots, \underline{i}_m) \in  \prod_{j=1}^{m}\{1,2,\ldots,N_j\}$, with $C_{\underline{i}}$  the partition of $\mathbb{R}^m$ defined by
\begin{equation}\label{eq:cartesian_prod}
C_{\underline{i}} = \prod_{l=1}^{m} \left[y_{\underline{i}_l-\frac{1}{2}}^l, y_{\underline{i}_l+\frac{1}{2}}^l\right[, \hspace{0.4cm} y_{\underline{i}_l\pm\frac{1}{2}}^l:=\frac{y_{\underline{i}_l}^l+y_{\underline{i}_l \pm 1}^l}{2}, \hspace{0.2cm}  y_0^l = -\infty,\hspace{0.2cm} y_{N_l+1}^l = \infty.
\end{equation}
In other words, $C_{\underline{i}}$ is a Cartesian grid formed by $m$ number of optimal $L^2$-quantized standard Gaussians. Using the independence of the family $(\widehat{Y}_{k}^{(N_{k})})_{k=1,\ldots, m}$, the probability associated with each trajectory $w_{\underline{i}}$ of $\widehat{W}_{t}^{\left(N,m\right)}$, {defined as $w_{\underline{i}} = \sum_{k=1}^{m} \sqrt{\lambda_{k}}e_k{(t)} y_{\underline{i}_k}^k$}, is straightforward and given by
\[
p_{\underline{i}} : =  \mathbb{Q}\left(\widehat W^{(N,m)}=w_{\underline{i}}\right) =  \mathbb{Q}\left(\widehat Y^{(N_k)}_k=y_{\underline{i}_k}^k, k \in \{1,\ldots,m\}\right) = \prod_{k=1}^{m} \mathbb{Q}\left(\widehat Y^{(N_k)}_k=y_{\underline{i}_k}^k\right),
\]
where we recall that the quantities $ \mathbb{Q}(\widehat Y^{(N_k)}_k=y_{\underline{i}_k})$ appearing on the right hand side are explicitly given by \eqref{eq:proba1d}.

By fixing the number of trajectories $N$ and using the $L^2$-optimal $N_k$-quantizer $Y_k^{(N_k)}$, the optimal $L^2$-product quantizer on the space $\Omega\times [0,T]$  essentially comes down to the trade-off between the choice of $m$ and the sequence $(N_k)_{k \leq m}$ subject to $\prod_{k=1}^{m}N_k \leq N$ by solving the following optimisation problem
\begin{equation}
   \displaystyle  \min_{\substack{(N_k)_{k \leq m}, m \geq 1 \\ \prod_{k=1}^{m}N_k \leq N}} \mathbb{E}\left[\int_{0}^{T} \left(W_{s}-\widehat{W}_{s}^{\left(N,m\right)} \right)^{2}ds\right],
\end{equation}
with
\begin{equation}\label{L2_product_quantization_err}
\mathbb{E}\left[\int_{0}^{T} \left(W_{s}-\widehat{W}_{s}^{\left(N,m\right)} \right)^{2}ds\right] = \sum_{k=1}^{m}\lambda_{k}\mathbb{E}\left[\left(Y_k-\widehat{Y}_{k}^{\left(N_{k}\right)}\right)^{2}\right]+\sum_{k=m+1}^{\infty}\lambda_{k},
\end{equation}
{using the orthonormality of the eigenfunction $e_k$ and independence of $Y_k$.} Thus, the minimization of the error function  \eqref{L2_product_quantization_err} is a balancing act between the $L^2$-projection error of individual i.i.d.~random variables $Y_k$ ({intra-class} error) and the error of truncation of the KL decomposition ({inter-class} error). It can be proved, see   \cite[Theorem 5]{pages2008quadratic}, that a solution to \eqref{L2_product_quantization_err}  always exists and is unique. {The  quantizer $\widehat W^{(N,m)^*}$ with $\{(N_k)_{k \leq m^*}^*,m^*\}$ that solves \eqref{L2_product_quantization_err} is called the $L^2$ (rate) optimal product quantizer.}  To alleviate notations in the sequel, we will 1) use $\widehat W$ to denote the $L^2$-optimal product quantizer of $W$, and 2) identify the tuple $\underline{i}$ with a corresponding index $i \in \{1,\ldots,N\}$, so when we mention the $i$-th trajectory of $\widehat W$, we are referring to the trajectory $\underline{i}$ generated by the tuple $y_{\underline{i}} = (y_{\underline{i}_1}^1,\ldots,y_{\underline{i}_m}^m,)$, so that 
\[{w_i := w_{\underline{i}} = \sum_{k=1}^{m} \sqrt{\lambda_{k}}e_k{(t)} y_{\underline{i}_k}^k}, \quad N_1 \times N_2 \times \ldots \times N_m\leq N.
\]

In practice, the optimal decomposition sequence $N_{1}, \ldots, N_{m}$ are solved numerically using blind optimization, which consists of computing \eqref{L2_product_quantization_err} for every possible decomposition $N_1 \times N_2 \times \ldots \times N_m\leq N$, see \cite{pages2005functional}. For standard Brownian motions, the optimal decomposition sequence for a wide range of values of $N$  is available online at \url{http://www.quantize.maths-fi.com/}. For an illustration, we plot an $L^2$-optimal product quantizer $\widehat W$ in Figure~\ref{fig:bm}.

  \begin{figure}[H]
    \centering
    \includegraphics[scale=0.5]{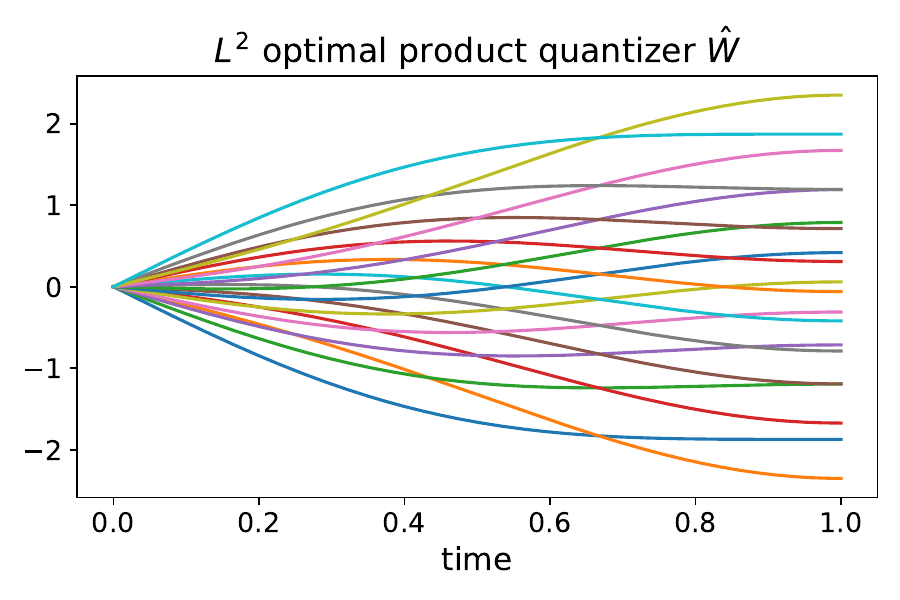}
    \caption{$L^2$-optimal product quantizer $\widehat W$ of the Brownian Motion, $N = N_1 \times N_2 = 10 \times 2 = 20$, T = 1.}\label{fig:bm}
  \end{figure}

\subsection{VIX {derivative} pricing via quantization}

We recall the explicit expression for the VIX in \eqref{eq:VIXclosed} involving the process $Z^u_T$ in \eqref{eq:ZuT}.

\subsubsection{Non-Markovian case}
\paragraph{Quantization of $Z_T^u$.} Fix $T>0$, {as done in \cite[Section 4]{bonesini2021functional}}, we will use the optimal product quantization $\widehat W$ in \eqref{eq:hatW} to build a functional quantizer $\widehat Z_T$ of the process $(Z_T^u)_{u \in [T, T+\Delta]}$ in \eqref{eq:ZuT}  as
\begin{align}
\widehat Z_T^u &= \int_0^T K(u-s)d \widehat W_s = \int_0^T K(u-s)\Dot{\widehat{W}}_sds  \\
&=  \sum_{k=1}^m \sqrt{\lambda_k}\widehat Y_k^{N_k}\int_0^T K(u-s)\Dot{e}_k(s) ds,\quad  \quad \quad   u \in [T,T+\Delta], \label{Z_quantized}
\end{align}
where $\dot f$ denotes the derivative of the function $t\mapsto f(t)$.

The quantization of $Z_T$ thus requires the computation of the integrals $\left(\int_0^T K(u-s)\Dot{e}_k(s) ds\right)_{u\geq T}$ which can be approximated numerically for general kernels.  For the fractional kernel {$K^{{frac}}$} and shifted fractional kernel  {$K^{shift}$}, recall  Table~\ref{tablekernels}, these quantities can be computed explicitly as specified in Appendix~\ref{A:formula}. For the log-modulated kernel {$K^{{log}}$}, we were unable to derive a closed-form solution and resorted to using numerical integration techniques to compute $\int_0^T K(u-s)\Dot{e}_k(s) ds$ directly {(e.g.~Gaussian quadrature, which seems to work well in practice)}.

The quantization of $Z_T$ is then straightforward. For  a tuple $\text{\underline{j}} \in \prod_{i=1}^{m}\{1,2,\ldots,N_i\}$,  we define $(z_{\underline{j}}^u)_{u \in [T,T+\Delta]}$ as the $\text{\underline{j}}^{\text{th}}$  trajectory of $(\widehat Z_T^u)_{u\in [T,T+\Delta]}$, formed through \eqref{Z_quantized} in the sense
\[\left(\widehat Y^{(N_k)}_k=y_{\underline{j}_k}^k, k \in \{1,\ldots,m\}, N_1 \times N_2 \times \ldots \times N_m\leq N \right)
\]
where $y_{\underline{j}_k}^k$ is the same as those defined in \eqref{eq:bm_quant_traj}, with the associated probability given by
\[p_{\underline{j}} =  \mathbb{Q}\left(\widehat Z_T^{u} = z_{\underline{j}}^{u}, T\leq u \leq T + \Delta \right) = \prod_{k=1}^{m} \mathbb{Q}\left(\widehat Y^{(N_k)}_k=y_{\underline{j}_k}^k\right)
\]
by the independence of $(\widehat Y^{(N_k)}_k)_{k=1,\ldots, m}$, where we recall that the quantities appearing on the right-hand side are explicitly given by \eqref{eq:proba1d}. 

Likewise to the quantized Brownian motion, we ease notations by identifying the tuple $\underline{j}$ with a corresponding index $j \in \{1,\ldots,N\}$, so that $(z_j^u)_{u \in [T,T+\Delta]}$ denotes the $j$-th trajectory of the quantizer $(\widehat Z^u_T)_{u \in [T,T+\Delta]}$ with its associated discrete probability $p_j{:=p_{\underline{j}}}$.

\paragraph{Quantization of VIX$_T$.} Once the quantizer $\widehat Z_T$ is computed,  we can plug it in the expression for the VIX in \eqref{eq:VIXclosed} to obtain the quantized version $\widehat{\mbox{VIX}}_T^2$
\begin{equation}
    \widehat{\mbox{VIX}}_T^2 = \frac{100^2}{\Delta} \sum_{k=0}^{2M}\sum_{i=0}^{k} (\alpha *\alpha)_k \binom{k}{i}\int_T^{T+\Delta} \frac{\xi_0(u)}{g(u)}\mathbb{E}\left[(G_T^u)^i \right] (\widehat Z_T^u)^{k-i}du.
\end{equation}
For $j=1,\ldots, N$, the $j$-th quantized  point $v_j$ for the $\mbox{VIX}_T^2$ is explicitly given by
\begin{equation}\label{vj}
v_j := \frac{100^2}{\Delta}\sum_{k=0}^{2M}\sum_{i=0}^{k} (\alpha *\alpha)_k \binom{k}{i}\int_T^{T+\Delta} \frac{\xi_0(u)}{g(u)}\mathbb{E}\left[(G_T^u)^i\right] (z_j^u)^{k-i} du.
\end{equation}
Recall that the moments $\E[(G^u_T)^i]$ are explicitly given by \eqref{eq:momentgaussian}, so that the integral in \eqref{vj}:
\[\int_T^{T+\Delta} \frac{\xi_0(u)}{g(u)}\mathbb{E}\left[(G_T^u)^i\right] (z_j^u)^{k-i}(u) du
\]
can be approximated efficiently (our numerical implementation shows that $50$ points between $T$ and $T+\Delta$ are largely sufficient using Gaussian quadrature).

With the quantized $\widehat{\mbox{VIX}}_T^2$, we can now price quickly a variety of  VIX derivatives with payoff function $\Phi$ by
\begin{align}\label{eq:VIXpayoffquantized}
\mathbb{E}[\Phi({\mbox{VIX}}_T)] \approx \mathbb{E}[\Phi(\widehat{\mbox{VIX}}_T)] = \sum_{j=1}^{N}\Phi(\sqrt{v_j})p_j.
\end{align}
In particular, we can price VIX Futures and VIX call options as follows:
  \begin{equation}\label{vix_future_quant}
    F_0^{\text{\scriptsize VIX}_{T}} = \mathbb{E}[\text{VIX}_T] \approx \sum_{j=1}^{N}\sqrt{v_j}p_j, \quad 
     C_0^{\text{\scriptsize VIX}_{T}} = \mathbb{E}\left[\left(\mbox{VIX}_T-K\right)^{+}\right] \approx \sum_{j=1}^{N}  (\sqrt{v_j}-K)^{+}p_j.
  \end{equation}

\paragraph{The moment-matching trick.} To further improve the accuracy of the  VIX pricing via quantization, we propose a moment-matching trick, that is
\begin{align}\label{eq:momentmatching}
\Tilde{z}_j^u := z_j^u \left(\frac{\mathbb{E}[(Z_T^u)^q]}{\sum_{i=1}^N (z_i^u)^q p_i}\right)^{\frac{1}{q}}, \quad  1\leq j \leq N,
\end{align}
for some {even} integer $q$ (since $Z_T$ has zero odd moments) and for each time step $u \in \left[T, T+\Delta \right]$ and then substitute $\Tilde{z}_j^u$ into \eqref{vj}. {This way, we can match the $q$-th moment of the quantizer $\widehat Z_T^u$ with that of $Z_T^u$.} We suggest using $q=4$ based on numerical experiments. {Note that 
\begin{align}
  \mathbb{E}\left[(Z_T^u)^q \right] =
      \left(\int_{u-T}^{u}K^2(s)ds\right)^{\frac{q}{2}}(q-1)!!
\end{align}
can be computed explicitly for all non-Markovian kernels in this paper for even integer $q$.} {Figure~\ref{fig:Z_T} shows the moment-matching trick results in quantized trajectories being more spread out, and this turns out to speed up convergence in VIX derivative pricing as we will see later.}

  \begin{figure}[H]
    \centering
    \includegraphics[scale=0.5]{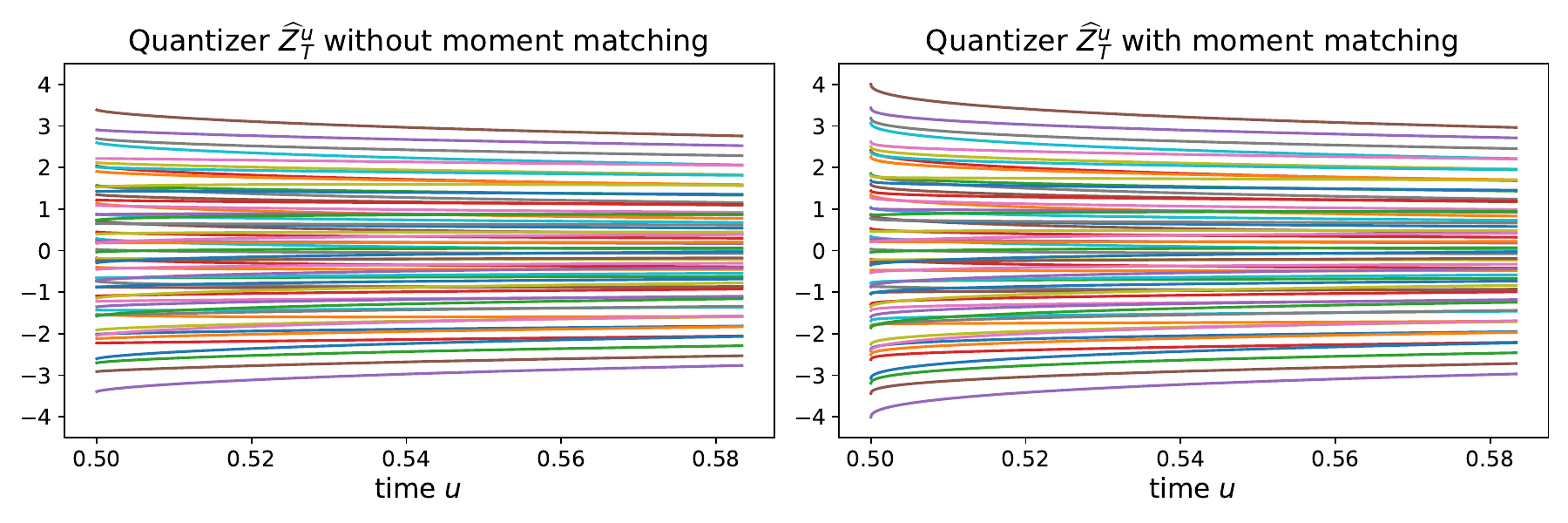}
    \caption{Comparison of quantizer $\widehat Z_T^u$ using the Volterra fractional kernel $K^{frac}$: LHS is without moment matching, RHS is with moment matching; $T = 0.5, H = 0.05$, $N = N_1 \times N_2 \times N_3 = 10 \times 3 \times 2 = 60$.}
    \label{fig:Z_T}
  \end{figure}

{We stress that such a moment matching trick improves the quality of the quantization considerably, both in terms of  VIX Futures prices and VIX implied volatility smile.} Indeed, as shown in Figure~\ref{Fig:momentmatching1}, quantization without the moment-matching trick is unusable in practice for the fractional kernel with small values of $H$ even with a lot of quantization trajectories $N=100,000$. See also   \cite[Figure 3]{bonesini2021functional}, where the number of quantized points was pushed as far as $N=1,000,000$ but the approximation is still well off the correct values due to the extremely slow convergence rate of the quantization of fractional process in the order of $1/(\log N)^{H}$, see  \cite{dereich2006high}.  After moment-matching, we achieve very accurate results with way fewer quantization points and faster convergence: Figure \ref{Fig:momentmatching2} shows the convergence of VIX pricing using quantization, where $N=200$ seems largely sufficient for pricing and calibrating VIX derivatives! (The Monte Carlo benchmark results are obtained using one million simulations with antithetic variables, with total time steps of $50$ between $T$ and $T+\Delta$.)

  \begin{figure}[H]
    \centering
    \includegraphics[scale=0.5]{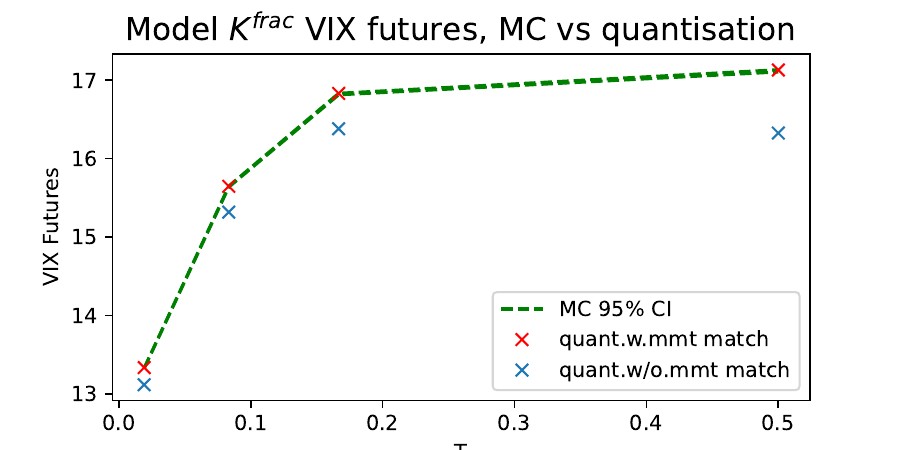}
    \includegraphics[scale=0.5]{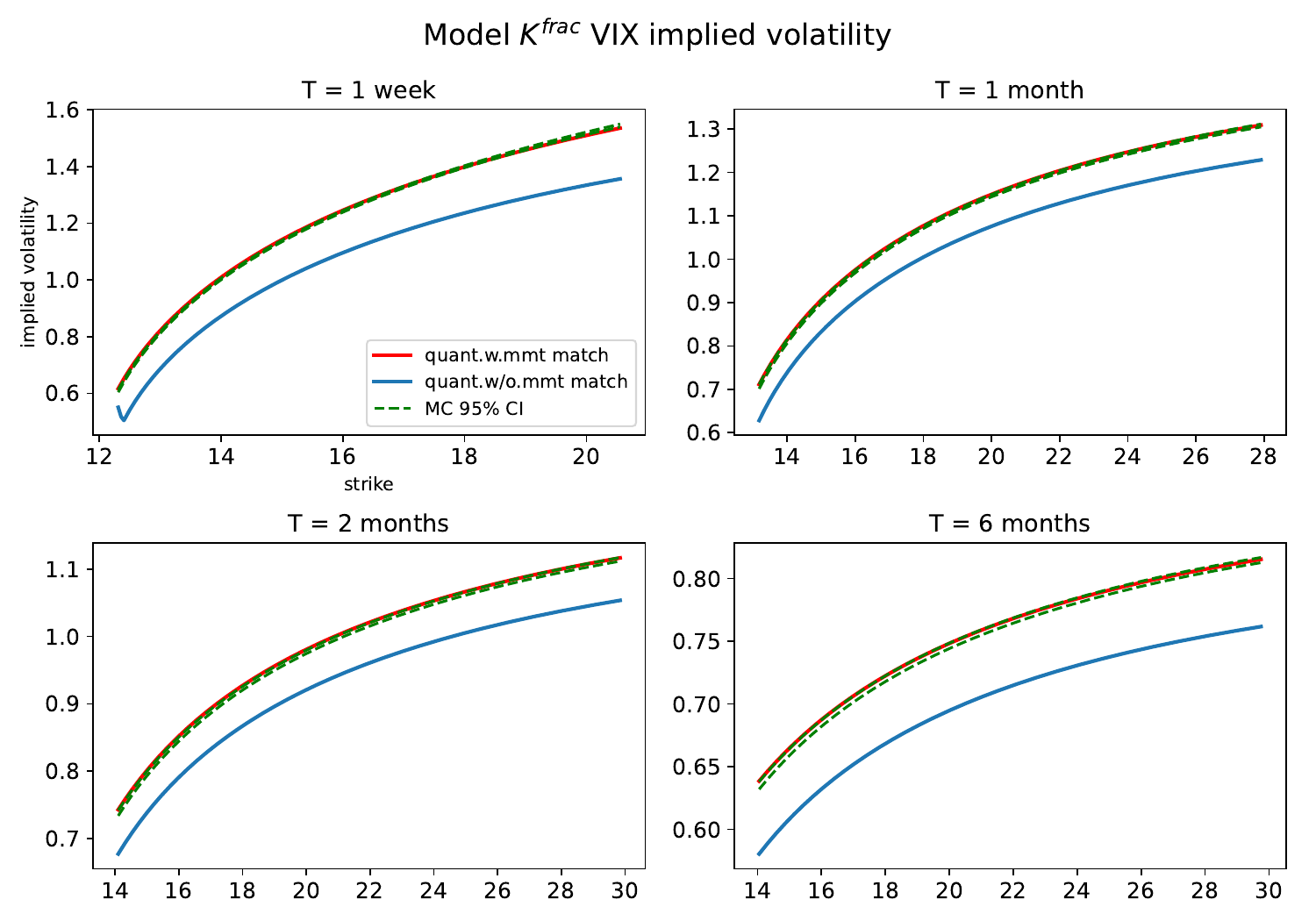}
    \caption{VIX Futures term structure and VIX call options implied volatility with  $N=100,000$ quantized trajectories of $\widehat Z_T$ with moment matching (in red), without moment matching (blue) versus Monte-Carlo prices (dashed green) with one million trajectories and 50 time steps, for the fractional kernel $K^{frac}$ with Hurst index  $H=0.05$ and parameters $\xi_0(t) = 0.005e^{-8t}+0.04(1-e^{-8t})$, $(\alpha_0,\alpha_1,\alpha_3,\alpha_5) = (0.01,1,0.214,0.227)$.}
    \label{Fig:momentmatching1}
  \end{figure}

  \begin{figure}[H]
    \centering
    \includegraphics[scale=0.6]{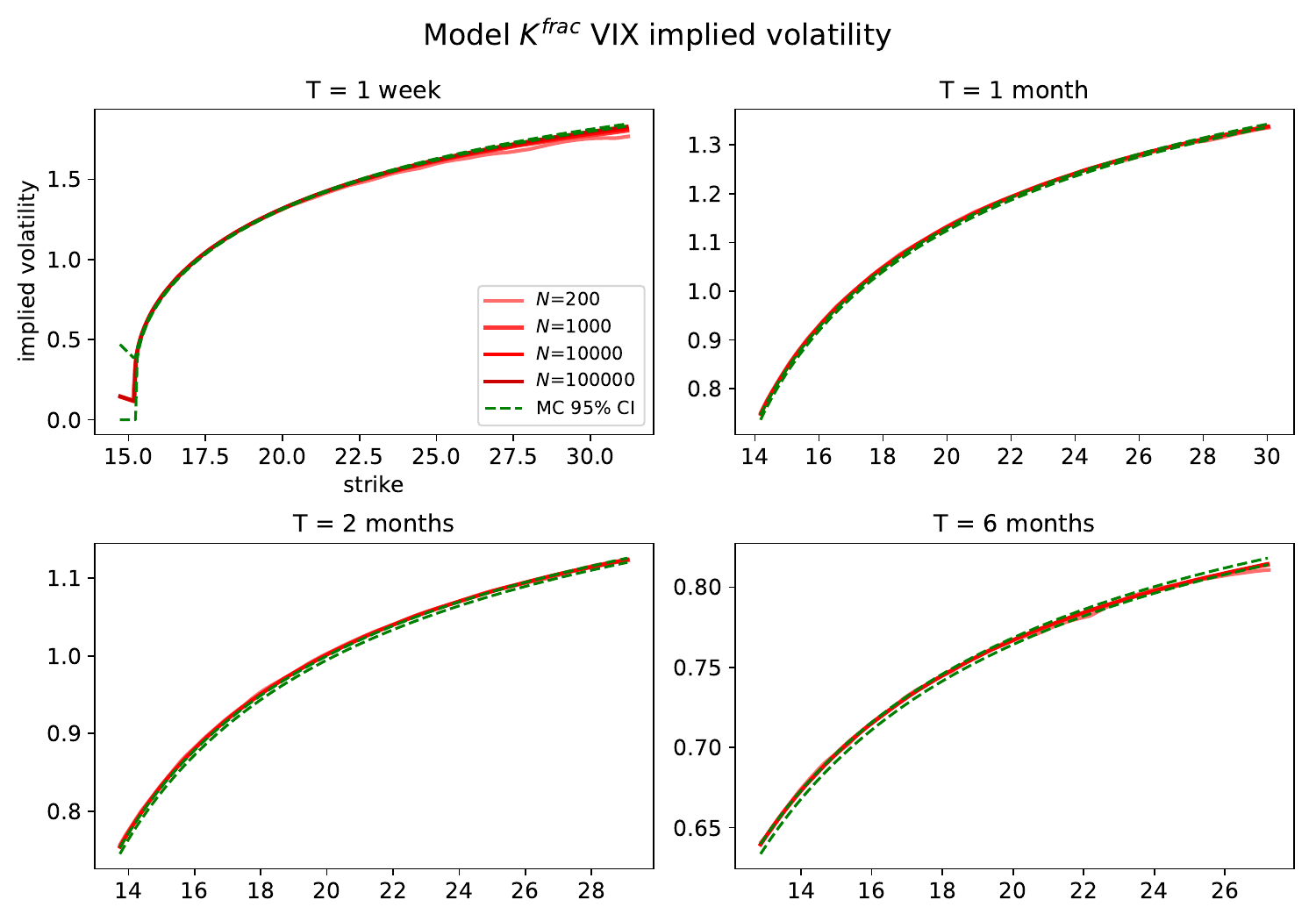}
    \caption{{Convergence of the implied volatility of VIX call options using quantization with varying $N$ number of trajectories of $\widehat Z_T$ under the fractional kernel $K^{frac}$ after moment matching. The Monte-Carlo prices (dashed green) are estimated with one million trajectories and 50 time steps. $H=0.05$, $\xi_0(t) = 0.03$, $(\alpha_0,\alpha_1,\alpha_3,\alpha_5) = (0.01,1,0.214,0.227)$}.}
  \label{Fig:momentmatching2}
  \end{figure}

\subsubsection{Markovian case}
In the case of the exponential kernel $K^{exp}(t)={\varepsilon^{H-1/2}}e^{-(1/2-H){\varepsilon^{-1}} t}$, the expression of $Z_T^u$  in \eqref{eq:ZuT}  can be simplified to  $Z_T^u = X_T e^{-{(1/2-H)\varepsilon^{-1}}(u-T)}$ with $X_T \sim  \mathcal{N} \left(0,{\varepsilon^{2H-1}} \frac{1-e^{-(1-2H){\varepsilon^{-1}} T}}{(1-2H){\varepsilon^{-1}}}\right)$. So that  VIX$_T$ is just a function of ${\varepsilon^{H-1/2}} \sqrt{\frac{1-e^{-(1-2H){\varepsilon^{-1}} T}}{(1-2H){\varepsilon^{-1}}}}Y$ with $Y$ a standard Gaussian
{\begin{align}
    \mbox{VIX}_T^2 &= \frac{100^2}{\Delta} \sum_{k=0}^{2M}\sum_{i=0}^{k} (\alpha *\alpha)_k \binom{k}{i}\left({\varepsilon^{H-1/2}} \sqrt{\frac{1-e^{-(1-2H){\varepsilon^{-1}} T}}{(1-2H){\varepsilon^{-1}}}}Y\right)^{k-i}\\
    &\quad \quad \quad \quad\quad \quad \quad \quad \times \int_T^{T+\Delta} \frac{\xi_0(u)}{g(u)}\mathbb{E}\left[(G_T^u)^i \right] (e^{-{(1/2-H)\varepsilon^{-1}}(u-T)(k-i)})du. \label{vix2_exp_kernel}
\end{align}}

Applying the $L^2$-optimal quantizer on $Y$ as in Section~\ref{S:l2gaussian},  for $j=1,\ldots, N$, the $j$-th quantized  point $v_j$ for the $\widehat {\mbox{VIX}}_T$ is explicitly given by
\begin{align}
v_j &= \frac{100^2}{\Delta}\sum_{k=0}^{2M}\sum_{i=0}^{k} y_j^{k-i} (\alpha *\alpha)_k \binom{k}{i}{\varepsilon^{(H-1/2)(k-i)}} \left(\frac{1-e^{-(1-2H){\varepsilon^{-1}} T}}{(1-2H){\varepsilon^{-1}}}\right)^{\frac{k-i}{2}}\\
& \quad \quad \quad \quad\quad \quad \quad \quad \times \int_T^{T+\Delta} \frac{\xi_0(u)}{g(u)}\mathbb{E}\left[(G_T^u)^i\right] e^{-(1/2-H){\varepsilon^{-1}}(u-T)(k-i)}du,
\end{align}
where $\{y_1,\dotsm, y_{N}\}$ are the quantized points of $\widehat Y$ with corresponding probabilities  $p_j  =  \mathbb{Q}(\widehat Y = y_j)$. The VIX derivatives can then be computed by plugging $(v_j,p_j)_{j=1,\ldots, N}$ in the formula  \eqref{eq:VIXpayoffquantized}. For the exponential kernel, the convergence is a lot faster {after applying a similar moment-matching trick directly on the quantized points of $\widehat Y$, with $q$ an even integer
\begin{align}\label{eq:momentmatching_markovian}
\Tilde{y}_j := y_j \left(\frac{\mathbb{E}[Y^q]}{\sum_{i=1}^N (y_i)^q p_i}\right)^{\frac{1}{q}}, \quad  1\leq j \leq N.
\end{align}
}

Th convergence of VIX pricing using quantization in the Markovian case is shown in Figure~\ref{convergence_exp_K}.

  \begin{figure}[H]
    \centering
    \includegraphics[scale=0.6]{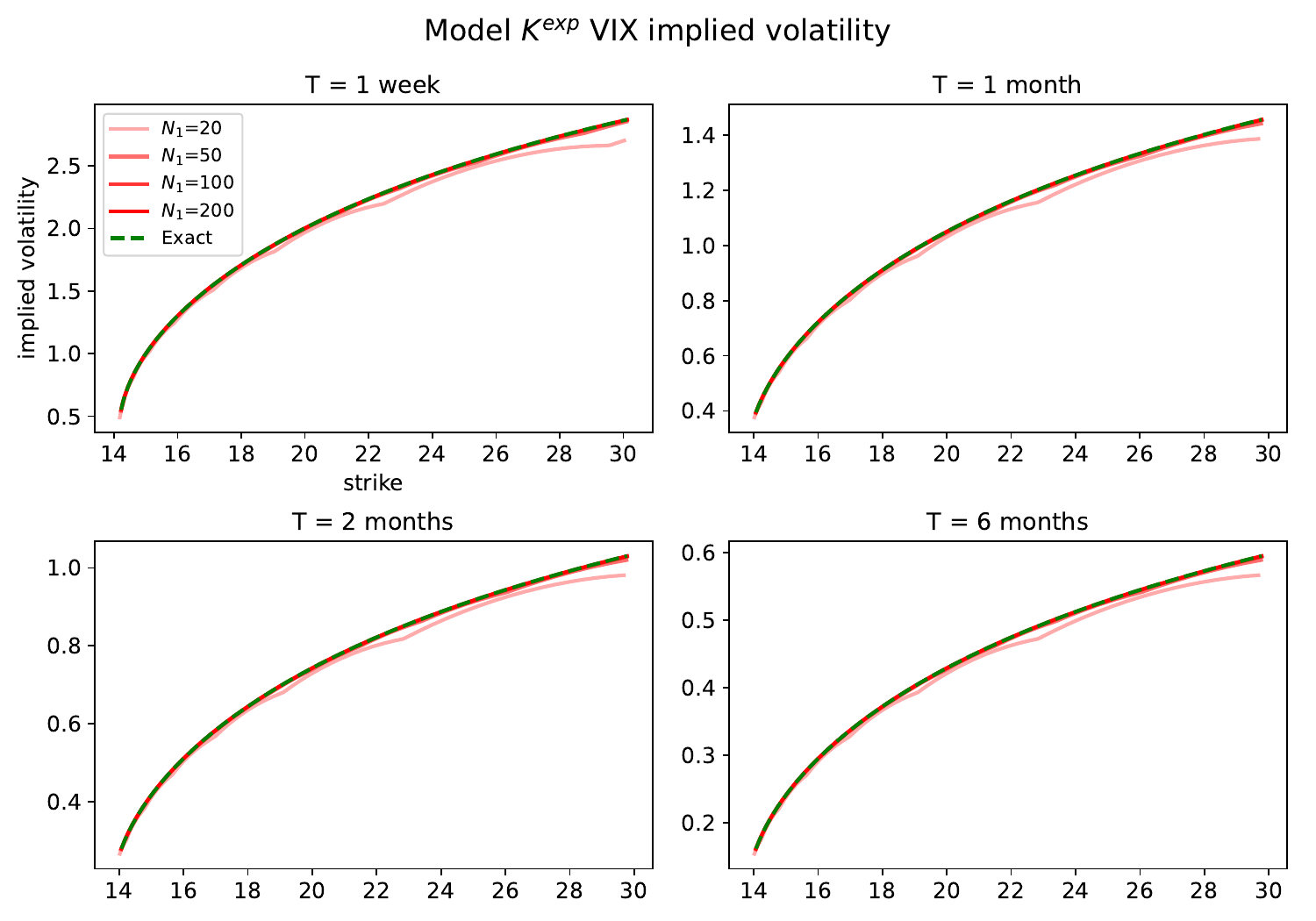}
    \caption{Convergence of the implied volatility of VIX call options using quantization with varying $N$ number of trajectories of $\widehat Z_T$ under the exponential kernel $K^{exp}$ after moment matching. The dotted green line represents VIX implied volatility calculated by integrating the payoff function of expression \eqref{vix2_exp_kernel} directly against the Gaussian density. $H=-0.2, \varepsilon = 1/52$, $\xi_0(\cdot) = 0.03$, $(\alpha_0,\alpha_1,\alpha_3,\alpha_5) = (0.01,1,0.214,0.227)$.}
    \label{convergence_exp_K}\end{figure}

{
Since $\text{VIX}$ is a function of Gaussian variable $Y$ in \eqref{vix2_exp_kernel}, one can also price VIX derivatives by integrating the payoff function against the Gaussian density numerically, see \cite{jaber2022quintic}. Nevertheless, we chose to price $\text{VIX}$ derivatives using functional quantization in this paper to ensure a fair comparison across all kernels and showcase the wide application of quantization in VIX derivative pricing. 
}

\section{Fast pricing of SPX options via Neural Networks with Quantization hints}\label{S:SPXpricing}

\subsection{Quantization techniques to SPX derivatives}\label{S:Quantization}
We now explain how to extend the quantization ideas to approximate derivative prices $P_0$ on $S$ of the form
\begin{align}
P_0:&= \mathbb{E}\left[F(\log(S_{T}))\right],
\end{align}
for some payoff function $F$. The process $\log(S)$ in \eqref{polynomial_model} reads
\[
\log(S_T) = \log(S_0)-\frac{1}{2}U_T+\rho V_T+\sqrt{1-\rho^{2}}V_T^{\perp},
\]
with
\[U_T=\int_{0}^{T}\sigma_{s}^{2}ds,\hspace{0.2cm} V_T=\int_{0}^{T}\sigma_{s}dW_{s},\hspace{0.2cm} V_T^{\perp}=\int_{0}^{T}\sigma_{s}dW_{s}^{\perp}.
\]
We denote by  $(\mathcal{F}^{W}_t)_{t\geq 0}$, the natural filtration generated by the Brownian motion $W$.

Using the celebrated conditioning argument on $\mathcal{F}_{{T}}^{W}$ in stochastic volatility models of \cite{renault1996option} combined with the fact that $\mathbb{E}\left[V_T^{\perp} \mid \mathcal{F}^{W}_T \right]$ is Gaussian with conditional mean zero and conditional variance $U_T$, we get  
\[
\begin{aligned}
P_0&= \mathbb{E}[F(\log(S_{T}))]\\
&=\mathbb{E}\left[\mathbb{E}\left[F\left(\log(S_0)-\frac{1}{2}U_T+\rho V_T+\sqrt{1-\rho^{2}}\ V_T^{\perp}\right)\Mid \mathcal{F}^{W}_T\right]\right]\\
&=\mathbb{E}\left[\Tilde F\left(U_T,V_T\right)\right]
\end{aligned}
\]
after applying properties of conditional expectations with the  deterministic function $\Tilde F$
\[
\Tilde F\left(u,v\right) = \int_{\mathbb{R}}F\left(\log(S_0)-\frac{1}{2}u+\rho v + \sqrt{(1-\rho^2)u}z\right)e^{-z^2/2}dz/\sqrt{2 \pi}.
\]

With the help of quantization, we then  approximate $P_0$ by
\[
\begin{aligned}\label{general_spx_pricing_formula}
P_0 \approx \sum_{j=1}^{N} p_j \Tilde F(u_j,v_j) \end{aligned}
\]
with  $(u_j,v_j,p_j)_{1 \leq j \leq N}$ the tuple of quantized points of $U_T$, $V_T$ and their associated discrete probability $p_j$ to be defined in the coming sections.

\begin{example}
For instance, for the case of a European Call option,  with payoff  $F(x) = (e^x-K)^{+}$, we have 
\[
\Tilde F(U_T,V_T) =\mathcal{BS}^{call}\left(S_0\exp\left(-\frac{1}{2}\rho^{2}U_T+\rho V_T\right),\sqrt{\frac{(1-\rho^{2})U_T}{T}},T,K\right)
\]
with
\[\mathcal{BS}^{call} (x,\sigma,T,K) = x \mathcal{N}(d_1)-K \mathcal{N}(d_2), 
\]
\[
d_1 = \frac{\log\left(\frac{x}{K}\right)+\frac{1}{2} \sigma^2T}{\sigma \sqrt{T}}, \quad d_2 = d_1 - \sigma \sqrt{T},
\]
with  $\mathcal{N}(x)= \int_{-\infty}^x e^{-z^2/2}dz/\sqrt{2\pi}$ the cumulative density function of standard Gaussian. 

With the help of quantization, we  price the call by
\begin{equation}\label{price_spx_quantization}
C_0 \approx \sum_{j=1}^{N} p_j \mathcal{BS}^{call}\left(S_0\exp\left(-\frac{1}{2}\rho^{2}u_{j}+\rho v_{j}\right),\sqrt{\frac{(1-\rho^{2})u_{j}}{T}},T,K\right).
\end{equation}
\end{example}

 The task now is to quantize the two random variables ($U_T, V_T)$. 
 
\subsubsection{Quantization of \texorpdfstring{$U_T$}{Lg}}

Since the path of $U$ is determined by the volatility process $\sigma$, which is itself determined by the path of $X$, the first task is to quantize $X$. Let's define the quantizer $\widehat X$ as
\begin{equation}\label{X_quantized}
\widehat X_t = \int_0^t K(t-s)d \widehat W_s = \int_0^t K(t-s)\Dot{\widehat{W_s}}ds  = \sum_{k=1}^m \sqrt{\lambda_k}\widehat Y_k^{N_k}\int_0^t K(t-s)\Dot{e}_k(s) ds.
\end{equation}

For general kernel $K$, the integral $\int_0^t K(t-s)\Dot{e}_k(s) ds$ can be computed numerically  (e.g. Gaussian quadrature). For the cases of fractional kernel $K^{frac}$, shifted fractional $K^{shift}$ and exponential kernel $K^{exp}$ this integral can be computed explicitly in Appendix \ref{A:formula_Xt}.

{Figures \ref{fig:X_t} and  \ref{fig:X_t_fmr} show the quantized trajectories of $\widehat X$ under the fractional kernel and exponential kernel for different values of $H$. Recall that the exponential kernel $K^{exp}$ is well-defined for negative values of $H$.} Compared to the trajectories of $\widehat Z_T$ in Figure \ref{fig:Z_T}, the trajectories of $\widehat X$ seem more ``noisy" and tend to cross each other more often.

  \begin{figure}[H]
    \centering
    \includegraphics[scale=0.5]{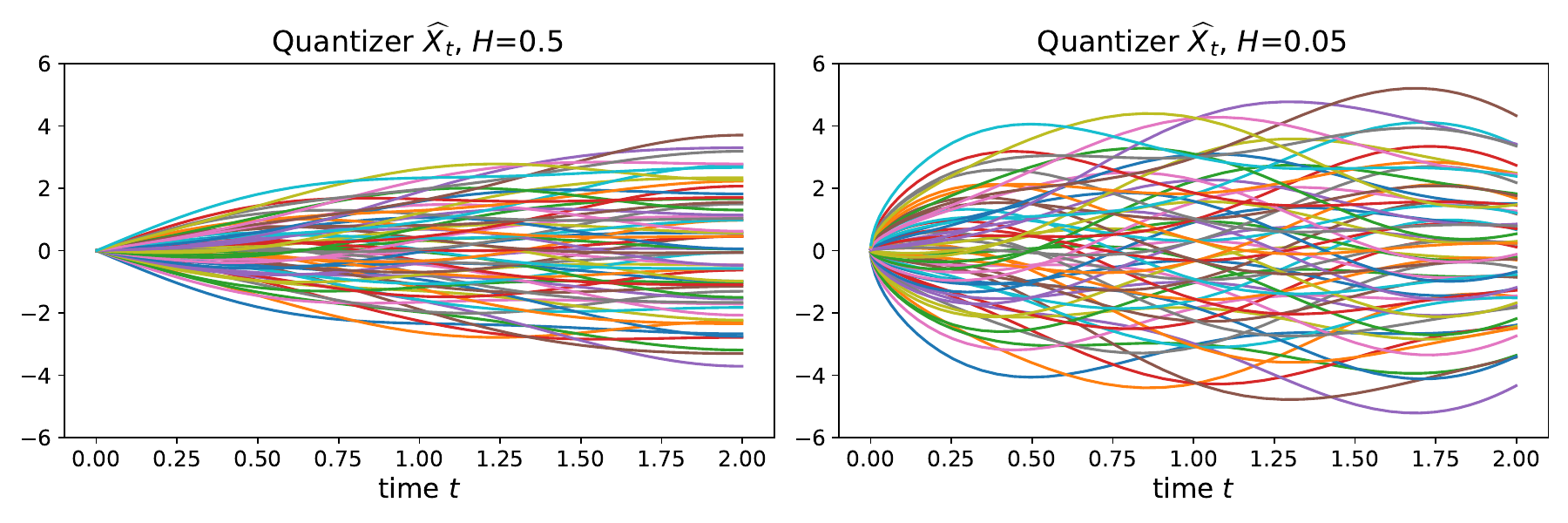}
    \caption{Quantizer $\widehat X_t$ using the fractional kernel $K^{frac}$, LHS is $H = 1/2$, RHS is $H=0.05$; $T = 2, N = N_1 \times N_2 \times N_3 = 10 \times 3 \times 2 = 60$.}
    \label{fig:X_t}
    
  \end{figure}
  
    \begin{figure}[H]
    \centering
    \includegraphics[scale=0.5]{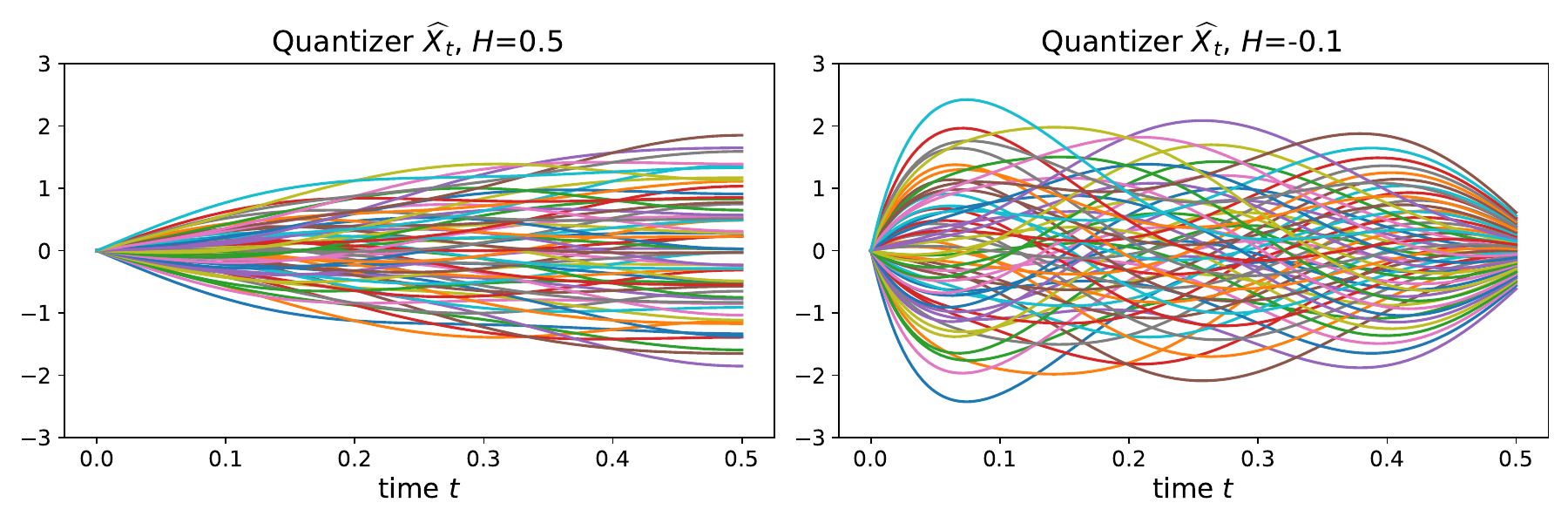}
    \caption{Quantizer $\widehat X_t$ using the exponential kernel $K^{exp}$, LHS is $H = 0.5$, RHS is $H=-0.1$; $\varepsilon = 1/52, T = 0.5, N = N_1 \times N_2 \times N_3 = 10 \times 3 \times 2 = 60$.}
    \label{fig:X_t_fmr}
  \end{figure}

We now define $(x^{\underline{j}}_t)_{t\geq 0}$ as the $\underline{j}$-th trajectory at time $t$ of the quantizer $(\widehat X_t)_{t\geq 0}$ in \eqref{X_quantized} with associated probability $p_{\underline{j}}$ as 
\[p_{\underline{j}} = \mathbb{Q}\left(\widehat X_t = x_t^{\underline{j}}, 0\leq t \leq T \right) = \prod_{k=1}^{m} \mathbb{Q}\left(\widehat Y^{(N_k)}_k=y_{\underline{j}_k}^k\right).
\] 
Again, to ease notations we identify the tuple $\underline{j}$ with a corresponding index $j \in \{1,\ldots,N\}$, so that $(x_t^j)_{t \in [0,T]}$ denotes the $j$-th trajectory of the quantizer $(\widehat X_t)_{t \in [0,T]}$ with its associated discrete probability $p_j:=p_{\underline{j}}$.

Define the quantizer $\widehat \sigma$  in \eqref{polynomial_model} as
\[
\widehat \sigma_t = \sqrt{\xi_0(t)}\frac{p(\widehat X_t)}{\sqrt{\E\left[p( X_t)^2\right]}}, \quad t \geq 0.\\
\]
We are now ready to define the quantizer $\widehat U$ as
\[
\widehat U_T=\int_{0}^{T} \widehat \sigma_{s}^{2}ds.
\]

This is a Riemann integral, which can be easily computed trajectory by trajectory using numerical integration {in the form of
\[
\int_{0}^{T} \widehat \sigma_{s}^{2}ds \approx \sum_{p=1}^{n_T}\widehat \sigma^2_{t_p}w_p,
\]
for some numerical quadrature methods with weights $w$. For our numerical experiments, we used Gaussian quadrature with $n_T = 50$, but other quadratures are possible.
}
We denote $u_T^j$ as the $j$-th quantized point of $\widehat U_T$, with the associated trajectory probability $p_j$ carried over directly from the quantizer $\widehat X$.

\subsubsection{Quantization of \texorpdfstring{$V_T$}{Lg}}
The quantization of $V_T$ is more intricate. To understand the intricacy, let us first look at the case where $\sigma$ is a semi-martingale.

\paragraph{Semi-martingale case.} This is the case for smooth enough kernels as shown in the following proposition.

\begin{proposition}\label{P:semi}
Assume that $K$ is absolutely continuous on  $[0,T]$ with a locally square-integrable derivative $K'$. Then, the process 
$$ X_t = \int_0^t K(t-s)dW_s  $$
is a semi-martingale with dynamics:
\begin{align*}
    dX_t = K(0) dW_t + \left(  \int_0^t K'(t-s) dW_s \right) dt.
\end{align*}
\end{proposition}
\begin{proof}
Using that $K= K(0) + \int_0^{\cdot}K'(s)ds$, we obtain  
\begin{align}
    X_t = K(0)W_t + \int_0^t \int_0^{t-s} K'(t-s-r) dr dW_{s}  = K(0)W_t + \int_0^t \int_0^{r} K'(r-s)  dW_s dr,
\end{align}

where the second equality follows from the stochastic Fubini theorem which applies since $K'$ is square-integrable. This ends the proof. 
\end{proof}

\begin{example}
The shifted fractional kernel $K^{shift}$ and the exponential kernel $K^{\exp}$ clearly satisfy the assumptions of Proposition~\ref{P:semi}. 
\end{example}

Based on the works of   \citet{wong1965convergence} and \citet{sellami2011convergence}, we have the convergence of the approximation $\int_0^T \widehat \sigma_s d\widehat W_s$ towards the Stratonovich integral $\int_0^T \sigma_s \circ dW_s$
\begin{align}\label{stratonovich_integral}
\int_0^T \widehat \sigma_s d\widehat W_s \xrightarrow[N\xrightarrow{}\infty]{} \int_0^T \sigma_s \circ dW_s = \int_0^T \sigma_s  d W_s + \frac 1 2 \left<\sigma,W \right>_T,
\end{align}
where $\langle \sigma, W\rangle$ is the quadratic covariation between the two semimartingales $\sigma$ and $W$ given by
\[
\langle \sigma, W \rangle_T = K(0)\int_0^T \sqrt{\xi_0(s)} \frac{p'(X_s)}{\sqrt{\mathbb{E}\left[p(X_s)^2\right]}}ds, \quad p'(x) = \sum_{k= 0}^{M} k \alpha_k x^{k-1}.
\]
We therefore construct 
 $\widehat {\left<\sigma, W \right>}_T$ as the quantized version of the quadratic covariation as follows
\begin{align}\label{eq:quadraticcov}
\widehat {\left<\sigma,W \right>}_T = K(0)\int_0^T \sqrt{\xi_0(s)} \frac{p'(\widehat X_s)}{\sqrt{\mathbb{E}\left[p(X_s)^2\right]}}ds.
\end{align}
 Using the identity \eqref{stratonovich_integral}, we  define the $\widehat V_T$ the quantizer of $V_T$ as  
\[
\widehat V_T = \int_0^T \widehat \sigma_s d\widehat W_s-\frac{1}{2}\widehat {\left<\sigma,W \right>}_T. 
\]
In practice, the quantizer $\widehat V_T$ is {also} computed numerically trajectory by trajectory. Likewise, we denote $v_T^j$ as the $j$-th trajectory of $\widehat V$ at time $T$, with the associated trajectory probability $p_j$ carried over directly from quantizer $\widehat X$.

\paragraph{Non semi-martingale case.}{If the process $(\sigma_t)_{t\geq 0}$ is not a semimartingale and has infinite quadratic variation, which is the case for the fractional $K^{frac}$ and the log-modulated kernels $K^{log}$  with $H<1/2$: the quadratic covariation $\langle \sigma, W \rangle$ explodes.  This can be seen informally on \eqref{eq:quadraticcov}  since the kernels are singular at $0$, i.e. ~$K^{frac}(0)=K^{log}(0)=+\infty$.  We suggest the following workaround.

{To avoid the explosion in \eqref{eq:quadraticcov}, we replace $K(0)$ with a positive constant $C$ so that
\[
\widehat {\left<\sigma,W \right>}_T^C = C\int_0^T \sqrt{\xi_0(s)} \frac{p'(\widehat X_s)}{\sqrt{\mathbb{E}\left[p(X_s)^2\right]}}ds
\]
and thus define the quantizer $\widehat V_T$ in the non semi-martingale setting as
\[
\widehat V_T = \int_0^T \widehat \sigma_s d\widehat W_s-\frac{1}{2}\widehat {\left<\sigma,W \right>}_T^C.
\]
A natural choice of $C$ is a value that enforces the centered martingale property of $\widehat V_T$, i.e.~$\E\left[\widehat V_T\right] = 0$. That is, we set
\begin{align}\label{eq:renormalization}
C = 2\frac{\E\left[\int_0^T \widehat \sigma_s d\widehat W_s\right]}{\int_0^T \sqrt{\xi_0(s)} \frac{\E\left[p'(\widehat X_s)\right]}{\sqrt{\mathbb{E}\left[p(X_s)^2\right]}}ds}
\end{align}
which can be computed numerically.

It is worth highlighting that $C$ diverges at the limit but acts as a renormalization constant. This bears resemblance to the renormalization theory \cite{hairer2014theory} and the approach in  \cite[Theorem 1.3]{bayer2020regularity}.
}}

\subsubsection{Numerical illustration}\label{ss:spx_quant_nums}
{For the numerical implementation, we propose two moment-matching tricks to improve convergence\footnote{{These moment-matching tricks are different from the one proposed for VIX quantization, recall \eqref{eq:momentmatching}, and seem to be more suited for pricing SPX derivatives numerically.}}. First, we consider the following modified quantizer of $\widehat W$
\begin{equation}\label{W_quant_match}
\widetilde{\widehat{W}}_t = \sum_{k=1}^{m} \sqrt{\lambda_{k}+\varepsilon_{\lambda}}e_k{(t)}\widehat{Y}_{k}^{\left(N_{k}\right)},
\end{equation}
where $\varepsilon_{\lambda}$ is set to match the trace of the covariance kernel of $X$, that is
\[
\mathbb{E} \left[\int_0^T \widehat{X}_s^2ds
\right] = \mathbb{E} \left[\int_0^T X_s^2ds
\right] =  \int_0^T \int_0^s K(u)^2duds,
\]
with $\widetilde{\widehat{W}}_t$ now entering into the definition of the quantizer $\widehat X$ in \eqref{X_quantized}.
The RHS of the integral above can be computed for all kernels in this paper. 

The second moment matching trick is to introduce a constant $b$ in front of $\widehat V_T$ such that
\[
\mathbb{E}\left[b \widehat V_T^2 \right] = \mathbb{E} \left[\int_0^T \sigma_s^2ds = 
\right] = \int_0^T \xi_0(s)ds,
\]
so that the new object $b \widehat V_T$ now matches both the first and second moments of $V_T$.
}

Figure \ref{fig:steinstein} shows that moment matching improves the quantization results. For $H=1/2$, the moment matching technique produces similar accuracy to that of Romberg interpolation proposed by \cite{pages2008quadratic}. However, for $H<1/2$, the moment matching tricks outperform Romberg interpolation.

  \begin{figure}[H]
    \centering
    \includegraphics[width=0.7\textwidth]{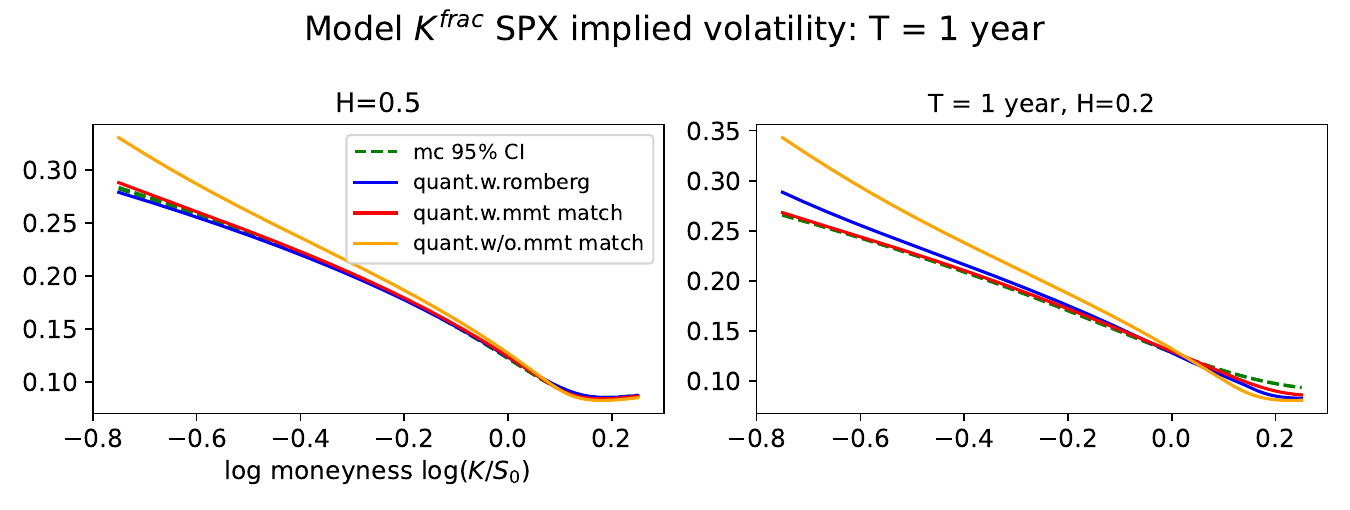}
    \caption{SPX implied volatility smile under the fractional kernel, $\rho = -0.8$, $\alpha = [0,1,0,0,0,0]$, $\xi_0(t) = 0.02$. The dotted green line is Monte Carlo; the blue line is quantization with Romberg interpolation between $N=1,000$ and $N=10,000$ as per \cite{pages2008quadratic}; the red line is quantization with 10,000 points applying moment matching; the orange line is quantization with 10,000 points without moment matching.\ }
    \label{fig:steinstein}  
  \end{figure}

Unfortunately, quantization approximation degrades as $H$ goes to zero, as shown in Figure \ref{fig:spx_quant_simple}. In the next section, we will discuss how to use Neural Networks to further improve quantization estimates, especially for lower $H$.

  \begin{figure}[H]
    \centering
    \includegraphics[width=0.8\textwidth]{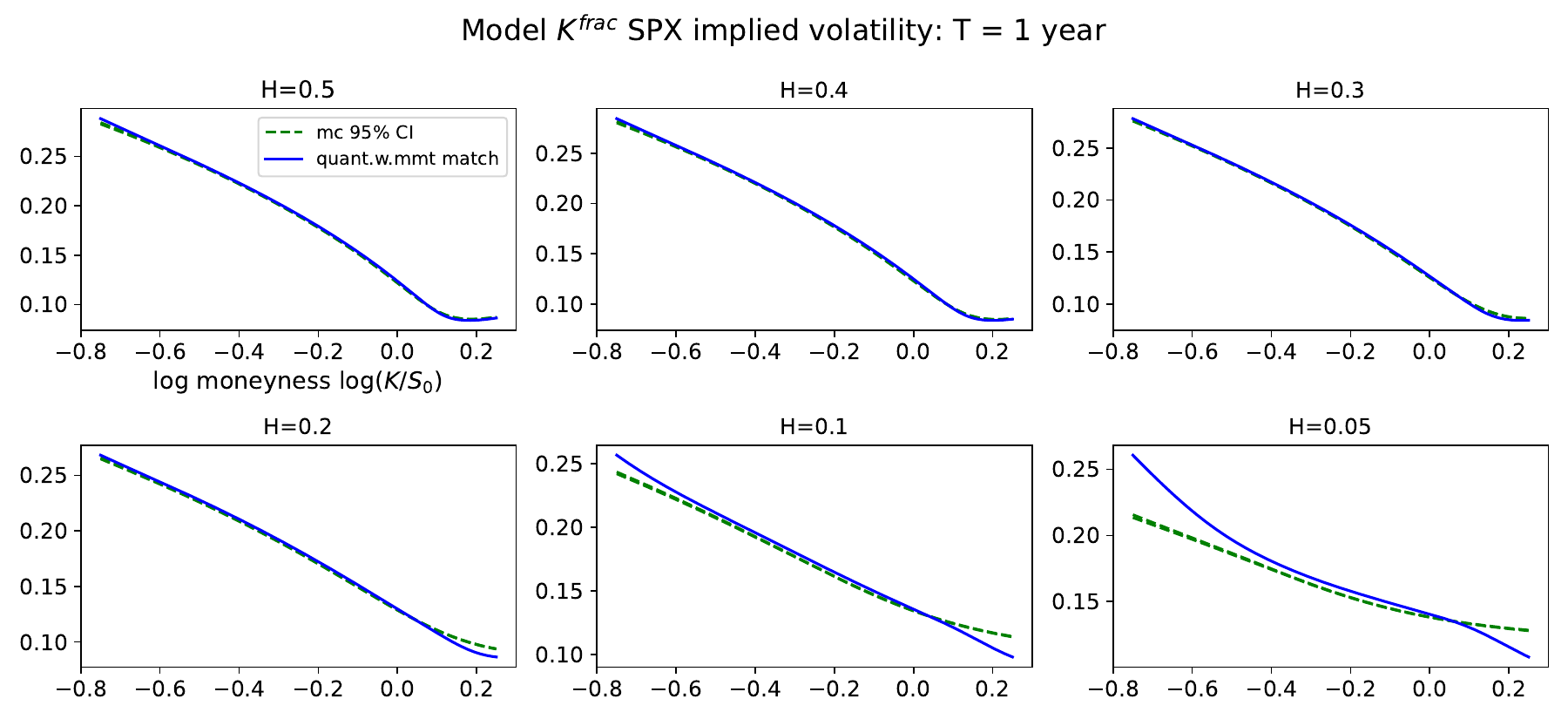}
    \caption{SPX smile under the fractional kernel $K^{frac}$ with different values of $H$, $\rho = -0.8$, $\alpha = [0,1,0,0,0,0]$, $\xi_0(t) = 0.02$. The Green dotted line is Monte Carlo; the blue line is quantization with 10,000 points after moment matching.} \label{fig:spx_quant_simple}
  \end{figure}

\subsection{Quantization Neural Network}\label{S:NN}

As we see in Figure \ref{fig:spx_quant_simple}, the results obtained through quantization degrade as 
$H$ gets closer to zero for the  singular kernels $K^{frac}$ and $K^{log}$. This is also true to some degree for the non-singular kernels $K^{shift}$ and $K^{exp}$ for lower values of $H$. To further improve the estimates on {SPX} options, we will use Neural Networks.

First, define $n_{\Theta}$ the dimension of the input of the Neural Networks, with $\Theta$ referring to the parameters of the Gaussian polynomial volatility model defined in \eqref{eq:THeta} (including the extra parameter $\beta$ for the log-modulated kernel $K^{log}$) together with the maturity parameter $T${, so that $n_{\Theta} = dim(\Theta)+1$.} Recall that we fixed $\varepsilon = 1/52$ (1 week) for the exponential kernel $K^{exp}$ and time shifted fractional kernel $K^{shift}$, and $\theta=0.1$ for the log-modulated fractional kernel $K^{log}$.

Next, we use Neural Networks (or 3 Neural Networks $\mathcal {NN}^1, \mathcal {NN}^2, \mathcal {NN}^3$) to modify existing quantized trajectories of $\widehat X$ in {\eqref{X_quantized}} and $d\widetilde{\widehat W}$ {the derivative of $\widetilde{\widehat{W}}$} in \eqref{W_quant_match}, as well as to tweak the probability vector $p = (p_1,\cdots,p_N)$ by
\begin{equation}\label{nn_equations}
\begin{aligned}
X^{\mathcal{NN}} & =  \widehat X + \mathcal{NN}^1,\\
dW^{\mathcal{NN}} & =  d{\widetilde{\widehat W}} + \mathcal{NN}^2, \\
p^{\mathcal{NN}} & = \text{softmax}(p+\mathcal{NN}^3 ),
\end{aligned}
\end{equation}
with the function $\text{softmax}: \mathbb{R}^N 	\rightarrow{(0,1)^N}, \text{softmax}(r)_i = \frac{e^{r_i}}{\sum_{j=1}^{N}e^{r_j}}$ for $r = (r_1, \dots, r_N) \in \mathbb{R}^N$, ensuring the output probabilities are positive and sum to $1$. Here, $\widehat X$ and $d\widetilde{\widehat W}$ are matrices of dimension $N\times n_T$ representing the $N$-quantizer with a discrete number of time steps $n_T$. $\mathcal{NN}$ are Neural Networks such that
\[
\begin{aligned}
& \mathcal{NN}^1, \mathcal{NN}^2: \mathbb{R}^{n_\Theta} \rightarrow{ \mathbb{R}^{N\times n_T}},\\
& \mathcal{NN}^3:\mathbb{R}^{n_\Theta}\rightarrow{\mathbb{R}^N}.
\end{aligned}
\]

At this point, we highlight that the forward variance curves $(\xi_0(t))_{t\leq T}$ is not part of the input parameters of the Neural Networks, since we want to tweak the trajectories of $\widehat X$,  d${\widetilde{\widehat W}}$ and the trajectory probabilities so that they remain $\mathbf{independent}$ of the shape of $(\xi_0(t))_{t\leq T}$. {The treatment of forward variance curves has always been challenging in deep pricing, for example in \cite{horvath2021deep,romer2022empirical}, where a large number of input parameters for the Neural Networks were required to incorporate piece-wise constant forward variance curves}. However, our Neural Networks approach {solves this problem by} 1) using lower input dimension, and 2) generalizing over a larger variety of forward variance curves {(e.g. not only piece-wise constant)} during training.

\subsubsection{Neural Network setup}
The input parameters \eqref{eq:THeta} (and an extra parameter $\beta$ for the log-modulated kernel $K^{log}$) together with the maturity parameter $T$ is first normalized into the interval $[-1,1]^{n_{\Theta}}$ before feeding into the Neural Networks $\mathcal{NN}$. In terms of Neural Networks' architecture, we chose 3 hidden layers of 30 neurons each, connected using $tanh(x)  = \frac{e^x-e^{-x}}{e^x+e^{-x}}$ activation function, except for the output layer where identity function is used. The Neural Networks are built using the Tensorflow package in Python.

We then recompute $\widehat U_T$, $\widehat V_T$ using $X^{\mathcal{NN}}$ and $dW^{\mathcal{NN}}$, { while also applying the second moment-matching trick described in Section \ref{ss:spx_quant_nums}}. With new quantized points and probabilities $(u_j^{\mathcal{NN}},v_j^{\mathcal{NN}},p_j^{\mathcal{NN}})_{j \leq N}$, we compute the call option price
\begin{equation}\label{price_spx_NN}
C(\Theta,T,K)^{\mathcal{NN}} \approx \sum_{j=1}^{N} p_j^{\mathcal{NN}} \mathcal{BS}^{call}\left(S_0\exp(-\frac{1}{2}\rho^{2}u_{j}^{\mathcal{NN}}+\rho v_{j}^{\mathcal{NN}}),\sqrt{\frac{(1-\rho^{2})u_{j}^{\mathcal{NN}}}{T}},T,K\right).
\end{equation}

Introducing the following root mean square error
\[
\textbf{RMSE}_1 = \sqrt{\frac{1}{N_d \times M_k}\sum_{i=1}^{N_d}\sum_{j=1}^{M_k}\left( C(\Theta_i,T_i, K_j)^{\mathcal{NN}}-C(\Theta_i,T_i,K_j)^{\mathcal{MC}} \right)^2},
\]
where $C(\Theta_i,T_i,K_j)^{\mathcal{NN}}$ is the call option priced using $(X^{\mathcal{NN}},dW^{\mathcal{NN}}, p^{\mathcal{NN}})$ and $C(\Theta_i,T_i,K_j)^{\mathcal{MC}}$ is the call option priced through Monte Carlo. $N_d$ is the number of different ${\Theta}$ in the {training dataset} and $M_k$ is the number of strikes.

It is known that implied volatilities for SPX options are more sensitive to option price movements for deep out-of-the-money/in-the-money strikes at shorter maturities. In response, we propose a second loss function that penalizes more the regions far away from near-the-money at shorter maturities
\[
\textbf{RMSE}_2 = \sqrt{\frac{1}{N_d \times M_k}\sum_{i=1}^{N_d}\sum_{j=1}^{M_k}\left(\frac{1}{T_i}\log \left( \frac{C(\Theta_i,T_i,K_j)^{\mathcal{NN}}-lb^{call}(K_j)}{C(\Theta_i,T_i,K_j)^{\mathcal{MC}}-lb^{call}(K_j)} \right) \right)^2},
\]
where $lb^{call}(K) := (S_0-K)^{+}$ is the theoretical lower bound of a call option in our setting \eqref{polynomial_model}. 

We set the loss function for training the Neural Networks as
\[
\mathcal{L} = \textbf{RMSE}_1 + \frac{1}{2}\textbf{RMSE}_2.
\]

Despite the extra computation steps involved between the output of the Neural Networks in \eqref{nn_equations} and the call option price in \eqref{price_spx_NN}, we can capture these computations inside Tensorflow's computational graph to perform backward propagation to update Neural Networks' weights. We train the Neural Networks by splitting generated data 85/15 between the training and validation sets. For each epoch, the training data is divided into mini-batches of 200 samples. During the training, Adam Optimizer is used with a learning rate of 0.001 for the first 1,000 epochs, and then a learning rate of 0.0001 for the next 1,000 epochs. More training epochs with smaller learning rates do not improve the results notably, and there appears to be no over-fitting by checking the validation error vs.~training error, as seen for instance in the case of the exponential kernel $K^{exp}$ in Figure \ref{nn_heatmap_valid} and \ref{nn_heatmap_train}. The distribution of relative errors from the validation dataset seems to be similar to that from the training dataset. The relative error is calculated as $|C(\Theta,T,K)^{\mathcal{NN}}-C(\Theta,T,K)^{\mathcal{MC}}|/(C(\Theta,T,K)^{\mathcal{MC}}+\epsilon)$ with $\epsilon=0.1$ to ensure that the errors do not blow up for very small prices in the out of the money region.  

  \begin{figure}[H]
    \centering
    \includegraphics[width=1\textwidth]{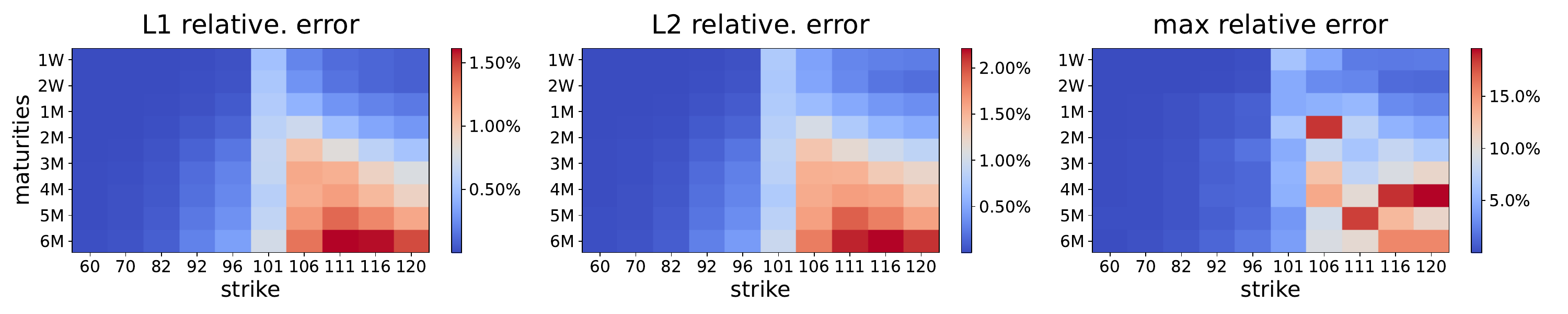}
    \caption{Heatmap of relative errors between estimated prices by Neural Networks vs. Monte Carlo from the validation dataset in the case of exponential kernel $K^{exp}$.}
    \label{nn_heatmap_valid}
  \end{figure}
  
  \begin{figure}[H]
    \centering
    \includegraphics[width=1\textwidth]{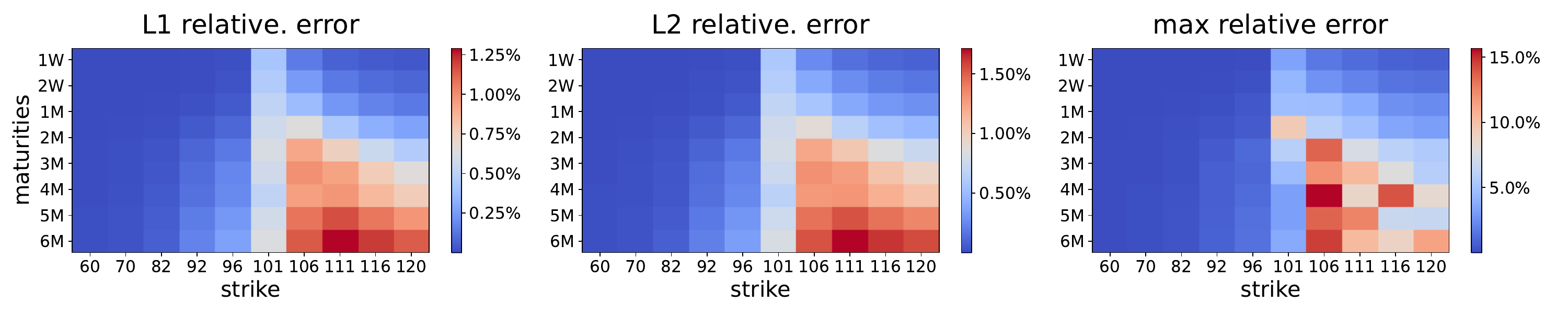}
    \caption{Heatmap of relative errors between estimated prices by Neural Networks vs. Monte Carlo from the training dataset in the case of exponential kernel $K^{exp}$.}
    \label{nn_heatmap_train}
  \end{figure}

 Figure \ref{nn_example} shows the improvement of option price estimation via Neural Networks, compared to that of quantization only.
  \begin{figure}[H]
    \centering
    \includegraphics[scale=0.4]{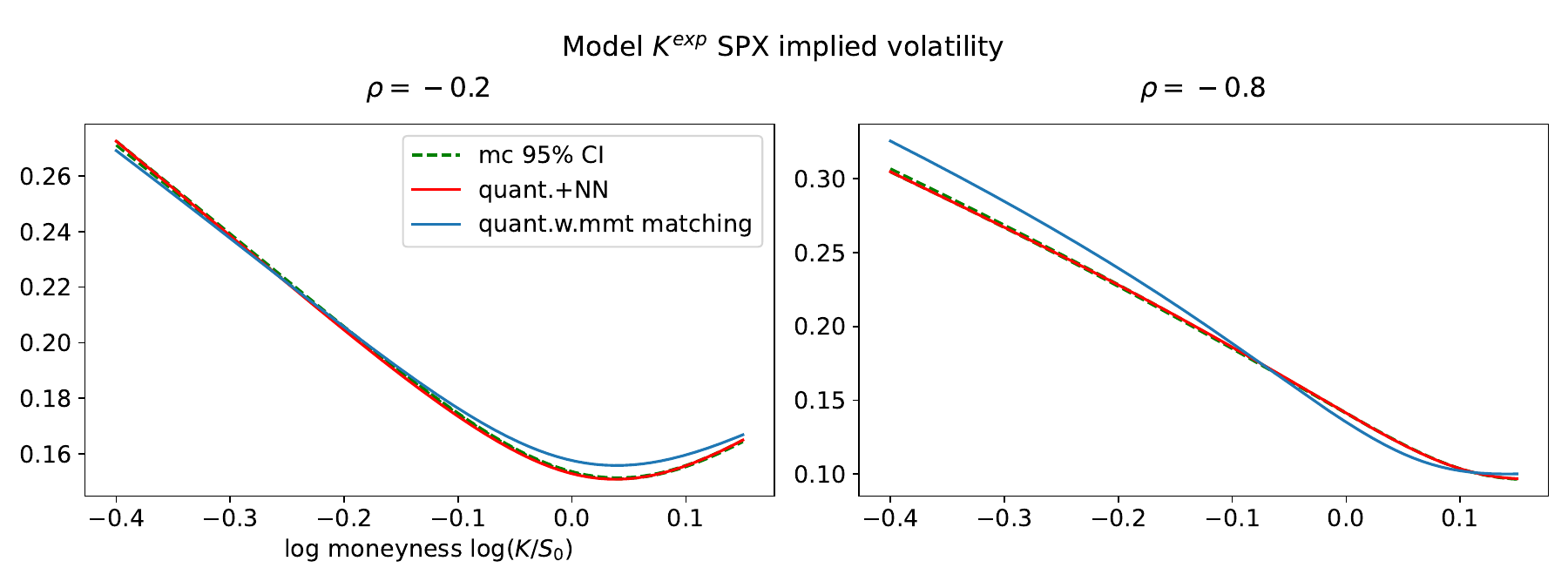}
    \caption{SPX implied volatility smile comparison between Monte Carlo (dotted green lines), quantization only (in blue) and with Neural Networks (in red) with the exponential kernel $K^{exp}$ for different values of $\rho$; $H = -0.1, (\alpha_0, \alpha_1, \alpha_3, \alpha_5) = (0.001,1,0.1,0.01), T = \text{6 months}$, $\xi_0(t) = 0.005e^{-8t}+0.04(1-e^{-8t})$. The quantization estimate (in blue) is computed using $10,000$ quantization points, where the quantization + Neural Networks (in red) is based on only $60$ quantization points.}
    \label{nn_example}
  \end{figure}

\subsubsection{A closed form density function for \texorpdfstring{$\log(S)$}{Lg}}
{Our approach emits a closed form density function of $\log(S_T/S_0)$, allowing us to price any derivatives depending on the final payoff of spot $S$. To see this, recall that 
\[
\log(S_T) = \log(S_0)-\frac{1}{2}U_T+\rho V_T+\sqrt{1-\rho^{2}}V_T^{\perp}.
\]
By modelling $U_T$ and $V_T$ using quantization, the law of $\log(S_T/S_0)$ is a Gaussian mixture with density function $f(x)$
\[
f(x) = \sum_{j=1}^{N} p_j^{\mathcal{NN}}f_j(x),
\]
where
\[f_j(x) = 
\frac{1}{\sqrt{2\pi(1-\rho^2)u_{j}^{\mathcal{NN}}}}\exp\left(\frac{-(x+\frac{1}{2}u_{j}^{\mathcal{NN}}-\rho v_{j}^{\mathcal{NN}})^2}{2(1-\rho^2)u_{j}^{\mathcal{NN}}}\right).
\]

Figure \ref{nn_density_example} is an example of density of $\log(S_T/S_0)$ produced by Neural Networks vs. Monte Carlo: the two are very close to each other:

  \begin{figure}[H]
    \centering
    \includegraphics[scale=0.4]{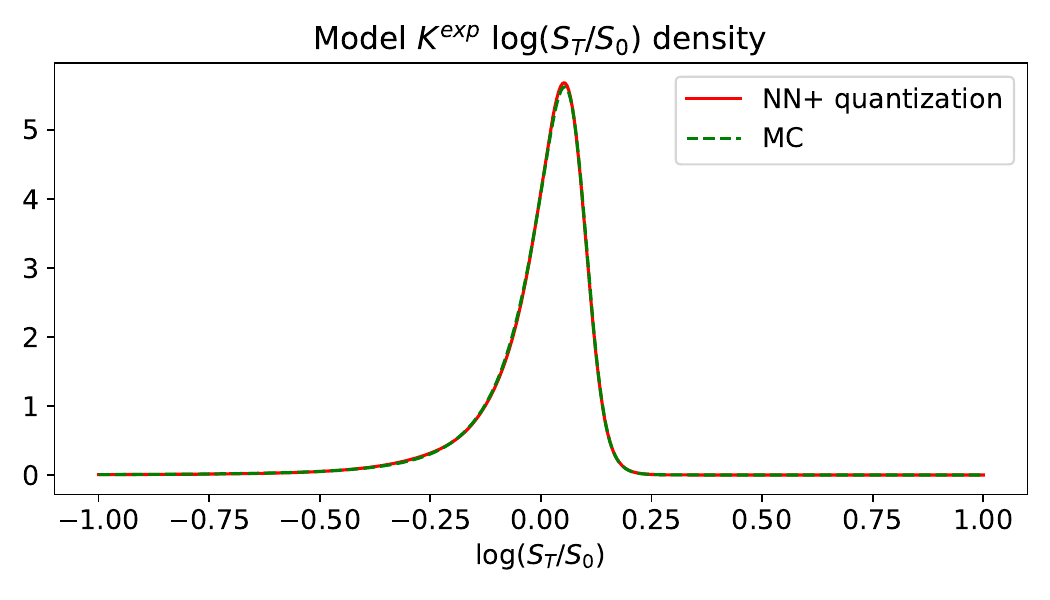}
    \caption{$\log(S_T/S_0)$ density comparison between Monte Carlo (dotted green lines) and Neural Networks (in red) with the exponential kernel $K^{exp}$; $H = -0.1, \rho=-0.8, (\alpha_0, \alpha_1, \alpha_3, \alpha_5) = (0.001,1,0.1,0.01), T = \text{6 months}$, $\xi_0(t) = 0.005e^{-8t}+0.04(1-e^{-8t})$.}
    \label{nn_density_example}
  \end{figure}
}

\subsubsection{Data Generation}\label{ss:data_gen}

First, we sample 100,000 different combinations of input parameters for each kernel in this paper as per the following
\[
\rho \sim \mathcal{U}\left[-1, -0.2\right], \left(\alpha_0, \alpha_1, \alpha_3, \alpha_5 \right)  \sim \mathcal{U} \left[0, 1\right]^4,T \sim \mathcal{U} \left[0.01, 0.5\right]; \beta \sim \mathcal{U} \left[1, 4\right],
\]
with $\mathcal{U}$ the uniform distribution. For the parameter $H$, we sampled $H \sim \mathcal{U} \left[-1, 0.5\right]$ for the exponential kernel $K^{exp}$ and the shifted fractional kernel $K^{shift}$. For the fractional kernel $K^{frac}$ and the log modulated kernel $K^{log}$, we sampled $H \sim \mathcal{U} \left[0.005, 0.5\right]$

To sample realistic forward variance curves $(\xi_0(t))_{t \leq T}$, we extracted the forward variance curves of SPX between 2017 and 2021 using the celebrated formula by Carr \& Madan \cite{carr2001towards}. To help generalizing the Neural Network to different $(\xi_0(t))_{t \leq T}$, we added some noise to the extracted $(\xi_0(t))_{t \leq T}$ at each discretised time step $(t_p)_{1 \leq p \leq n_T}$ by multiplying $e^{0.2Y_p}$ with $(Y_p)_{1 \leq p \leq n_T}$ i.i.d.~standard Gaussian. This idea is similar to that of \cite{romer2022empirical}.

After randomly pairing up the sampled parameters with sampled the forward variance curves, we compute the option call price via Monte Carlo simulations. The process $X$ is simulated exactly using Cholesky decomposition for all kernels.  Apart from the covariance matrix of $X$ under the log modulated kernel $K^{log}$ which is computed numerically, all covariance matrices are computed explicitly. 

To further reduce MC variance, We make use of the antithetic variable for $X$ and the control variable proposed by \cite{bergomi2015stochastic, mccrickerd2018turbocharging}. The MC prices are computed on a fixed vector of strikes $K:= (K_1, \ldots, K_{M_k})$ ranging between $60\%$ and $120\%$ of spot price $S_0=100$. The number of time steps $n_T$ is set to $800$ between 0 and $T$ for each simulation, with the number of simulations set to $80,000$ including antithetic variables.

\subsubsection{Critiques of Neural Networks}
Compared to the works of  \cite{horvath2021deep,romer2022empirical,rosenbaum2021deep}, our Neural Networks:
\begin{enumerate}
  \item learn the joint probability distribution of $(U, V, \log(S))$ instead of  learning the map between model parameters $\Theta$ and its implied volatility;
  \item do not require a fixed mesh of strikes and maturities and interpolation between various strikes/maturities during joint calibration.
\end{enumerate}

Our approach thus brings several benefits, for instance:
\begin{enumerate}

  \item \textbf{Greater flexibility}: since the output is a joint density, we can price vanilla options for any strikes and maturities; pricing of  other types of derivatives based on the stock price is also possible (transfer learning if needed);
  \item \textbf{Butterfly arbitrage-free}: positive density integrating to 1 for $\log(S)$;
  \item \textbf{Improved interpretability}: Neural Networks are used as a corrector of a first-order proxy coming from quantization;
  \item \textbf{Smaller input dimension}: the forward variance curve and strikes are not part of the input parameters.
\end{enumerate}

Of course, the price to pay for having a more flexible Neural Networks model is the large number of Neural Networks weights involved and longer training time. Recall the output of $\mathcal{NN}^1$ and $\mathcal{NN}^2$ have dimension of $N \times n_T$. As we have set $N=60$ and $n_T=50$, the number of weights connecting the last hidden layer and the output layer alone is $(30+1) \times (60 \times 50) = 93,000$ for $\mathcal{NN}^1$ and $\mathcal{NN}^2$. Larger $N$ and $n_T$ could further improve option prices estimated through Neural Networks, but this will also take longer time to train.

{In addition, it has been reported in \cite{ferguson2018deeply} that Neural Networks can learn to remove/reduce Monte Carlo noise in the training data. We also tested this claim on our Neural Network structure using the Volterra Stein-Stein model from \cite{abi2022characteristic} that is a special case of the Gaussian polynomial volatility model, see Section \ref{gpvm}. Under the Volterra Stein-Stein model, there exists a closed-form expression for SPX pricing via Fourier which allows direct comparison against Monte Carlo prices and Neural Networks trained using Monte Carlo prices. We observed considerable reductions in Monte Carlo noise using our Neural Network structure with detailed results in Appendix \ref{C:nn_rss}.}

\appendix

\section{Formulae}\label{A:formulae}
\subsection{A formula for \texorpdfstring{$\int_0^T (u-s)^{H-1/2}\Dot{e}_k(s) ds$}{Lg}}\label{A:formula}
\paragraph{Fractional kernel $K^{frac}$:} Thanks to \cite{bonesini2021functional}, there is a semi-closed form solution for the integral involving the fractional kernel $K^{frac}$
\[
\begin{aligned}
    &\int_0^T (u-s)^{H-1/2}\Dot{e}_k(s) ds = \Bigg\{ \cos \left(\frac{\left(k-\frac{1}{2}\right)}{T+\Delta} u \pi\right)\left(\zeta_{\frac{1}{2}}\left(\frac{\left(k-\frac{1}{2}\right)}{T+\Delta} u, h_1\right)-\zeta_{\frac{1}{2}}\left(\frac{\left(k-\frac{1}{2}\right)}{T+\Delta}(u-T), h_1\right)\right)\\
    &+\pi \sin \left(\frac{\left(k-\frac{1}{2}\right)}{T+\Delta} u \pi\right)\left(\zeta_{\frac{3}{2}}\left(\frac{\left(k-\frac{1}{2}\right)}{T+\Delta} u, h_2\right)-\zeta_{\frac{3}{2}}\left(\frac{\left(k-\frac{1}{2}\right)}{T+\Delta}(u-T), h_2\right)\right)\Bigg\}\frac{\sqrt{2}(T+\Delta)^H}{\left(k-\frac{1}{2}\right)^{H+\frac{1}{2}}\sqrt{\lambda_k}}
\end{aligned}
\]
with $h_1 = \frac{1}{2}(H+\frac{1}{2})$, $h_2 = h_1+\frac{1}{2}$ and
\[
\zeta_q(z,h) = \frac{z^{2h}}{2h} \, {}_1{F}_{2} (h;q,1+h;\text{-}\frac{1}{4}\pi^2z^2)
\]
where ${}_1 {F}_{2}$ is the hypergeometric function.

\paragraph{Shifted fractional kernel $K^{shift}$:} The formula is similar to the one above for the fractional kernel $K^{frac}$ by simply replacing $u$ by $u+\varepsilon$.

\subsection{A formula for \texorpdfstring{$\int_0^t (t-s)^{H-1/2}\Dot{e}_k(s) ds$}{Lg}}\label{A:formula_Xt}
\paragraph{Fractional kernel $K^{frac}$:}Thanks to \cite{bonesini2021functional}, there is a semi-closed form solution for the integral involving the fractional kernel $K^{frac}$
\[
\begin{aligned}
    &\int_0^t(t-s)^{H-1/2}\Dot{e}_k(s) ds = \frac{2}{1+2H}\sqrt{\frac{2}{T \lambda_k}}t^{H+1/2}{}_2{F}_{1}(1; \frac{3}{4}+\frac{H}{2}, \frac{5}{4}+\frac{H}{4}; -\frac{t^2}{4\lambda_k}).
\end{aligned}
\]

\paragraph{Shifted fractional kernel $K^{shift}$:}

\[
\begin{aligned}
    &\int_0^t (t+\varepsilon-s)^{H-1/2}\Dot{e}_k(s) ds = \Bigg\{ \cos \left(\frac{\left(k-\frac{1}{2}\right)}{T} (t+\varepsilon) \pi\right)\left(\zeta_{\frac{1}{2}}\left(\frac{\left(k-\frac{1}{2}\right)}{T} (t+\varepsilon), h_1\right)-\zeta_{\frac{1}{2}}\left(\frac{\left(k-\frac{1}{2}\right)}{T}\varepsilon, h_1\right)\right)\\
    &+\pi \sin \left(\frac{\left(k-\frac{1}{2}\right)}{T} (t+\varepsilon) \pi\right)\left(\zeta_{\frac{3}{2}}\left(\frac{\left(k-\frac{1}{2}\right)}{T} (t+\varepsilon), h_2\right)-\zeta_{\frac{3}{2}}\left(\frac{\left(k-\frac{1}{2}\right)}{T}\varepsilon, h_2\right)\right)\Bigg\}\frac{\sqrt{2}t^H}{\left(k-\frac{1}{2}\right)^{H+\frac{1}{2}}\sqrt{\lambda_k}}
\end{aligned}
\]
with $h_1 = \frac{1}{2}(H+\frac{1}{2})$, $h_2 = h_1+\frac{1}{2}$ and
\[
\zeta_q(z,h) = \frac{z^{2h}}{2h} \, {}_1{F}_{2} (h;q,1+h;\text{-}\frac{1}{4}\pi^2z^2)
\]
where ${}_1 {F}_{2}$ is the hypergeometric function.

\paragraph{Exponential kernel $K^{exp}$:}the form solution for the integral involving the exponential kernel
\[
\begin{aligned}
    \int_0^t K(t-s)\Dot{e}_k(s) &= \int_0^t \varepsilon^{H-1/2}e^{-(1/2-H)\varepsilon^{-1}(t-s)}\Dot{e}_k(s) ds\\
    &= \varepsilon^{H-1/2} \sqrt{\frac{2}{T}}\frac{-(1/2-H)\varepsilon^{-1} e^{-(1/2-H)\varepsilon^{-1} t}+(1/2-H)\varepsilon^{-1} \cos(\frac{t}{\sqrt{\lambda_k}})+\frac{1}{\sqrt{\lambda_k}}\sin{\frac{t}{\lambda_k}}}{((1/2-H)^2\varepsilon^{-2}+\frac{1}{\lambda_k})\sqrt{\lambda_k}}.
\end{aligned}
\]

{\section{More joint calibration results}\label{B:jc}
\subsection{RMSE across different kernels after calibration}\label{B:rmse_dist}
{
  \begin{figure}[H]
    \centering
    \includegraphics[scale=0.42]{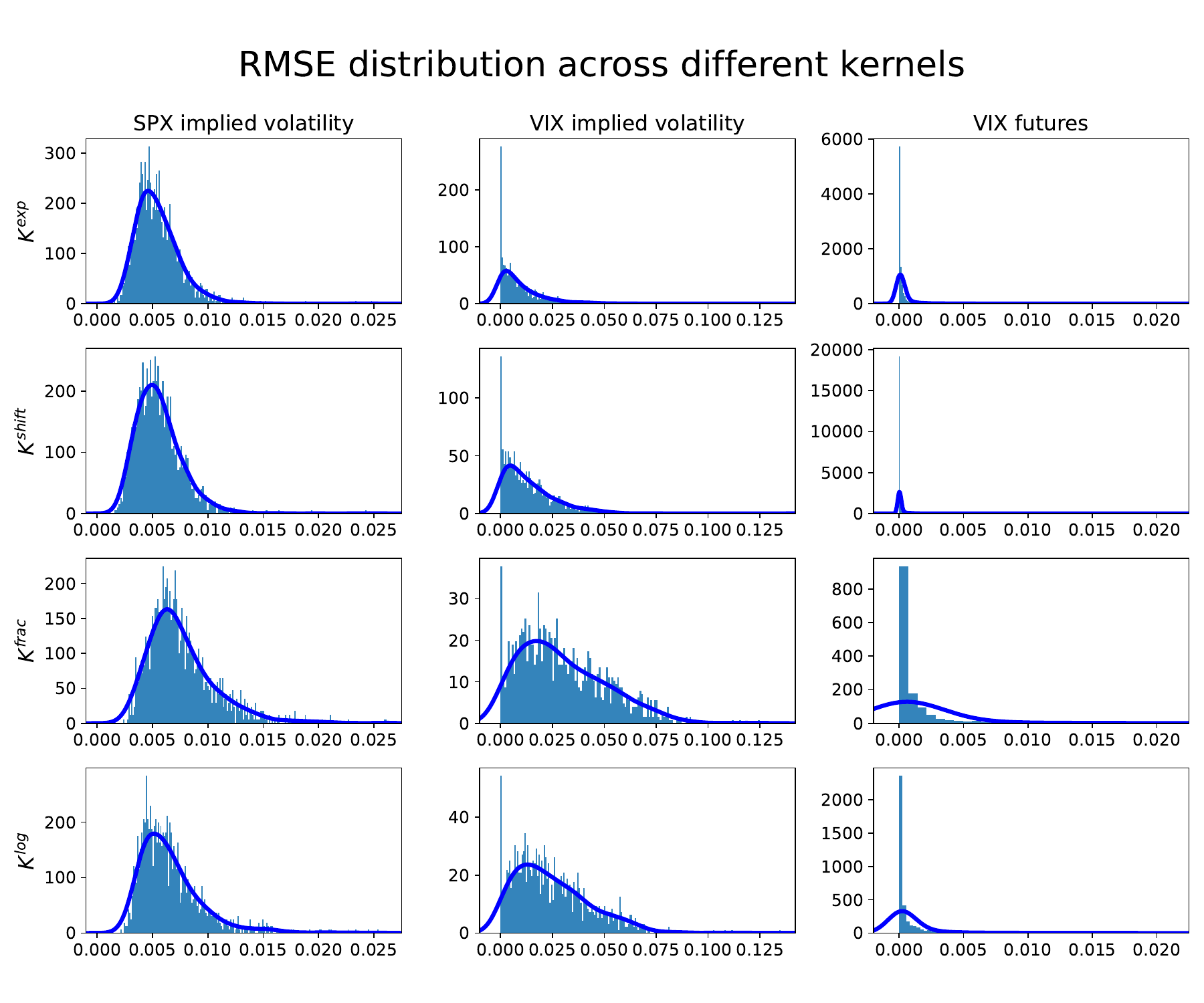}
    \caption{Histogram of the distribution of calibration RMSE for different kernels from August 2011 to September 2022.}
    \label{rmse_dist}
  \end{figure}
}

\subsection{Joint calibration among other kernels for the date 23 October 2017}\label{B:jc_kernels_compare}

{
  \begin{figure}[H]
    \centering
    \includegraphics[width=0.5\textwidth]{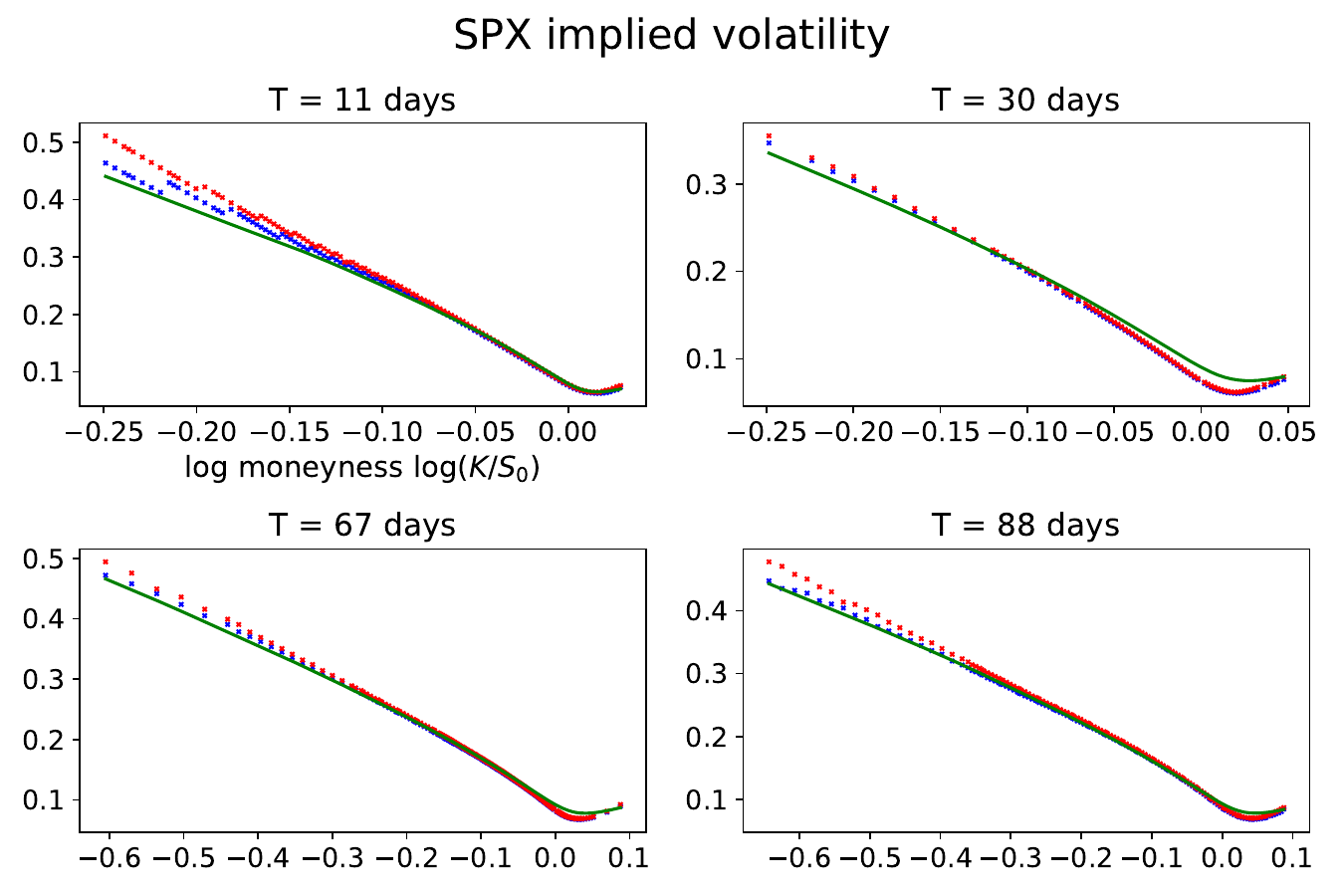}%
    \includegraphics[width=0.5\textwidth]{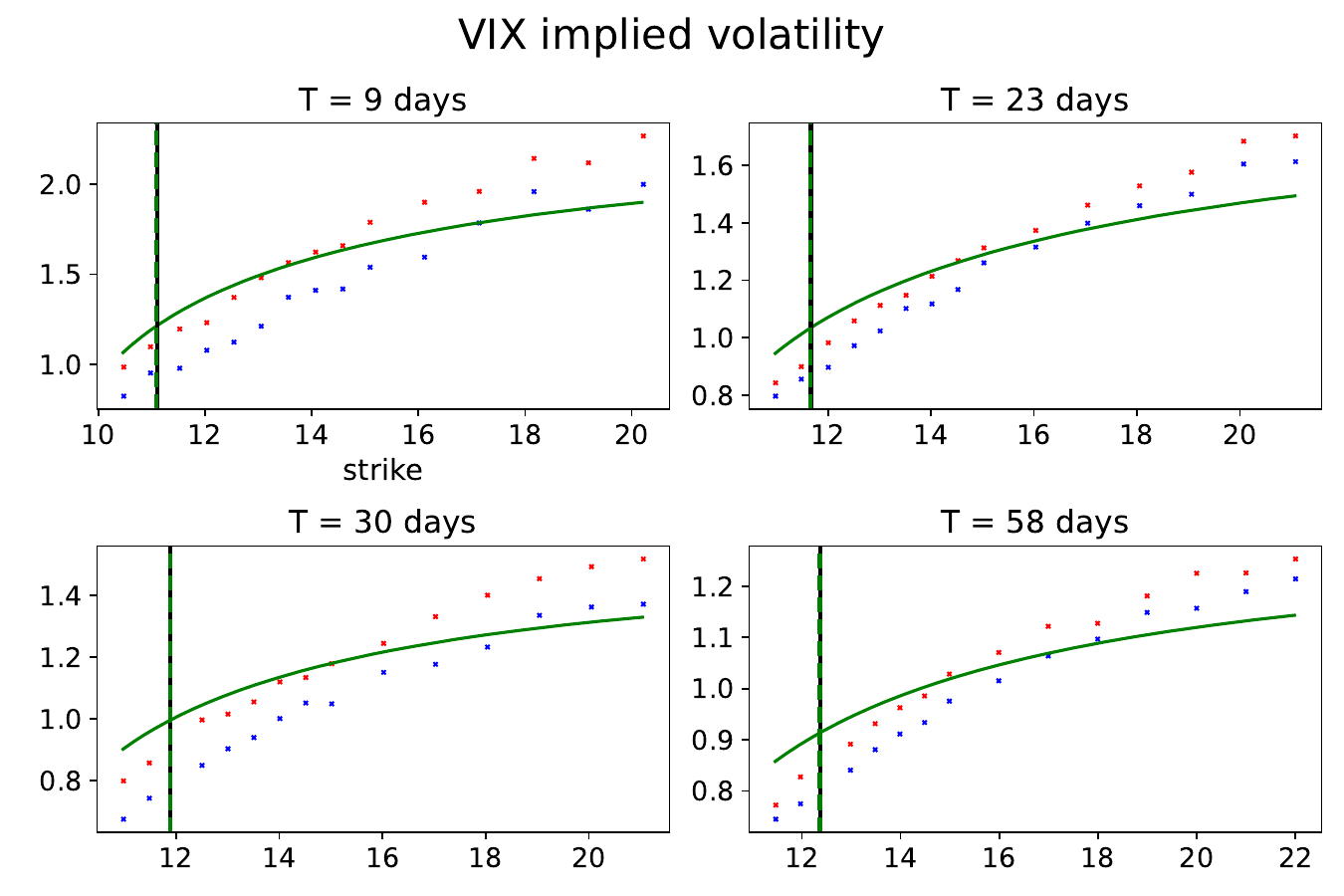}
    \caption{Fractional kernel $K^{frac}$: Joint calibration of SPX implied volatility, VIX implied volatility, and VIX Futures on 23 October 2017 using the fractional kernel $K^{frac}$. The blue and red dots are bid/ask implied volatilities, with the green lines are model fit. The vertical bars represent VIX Futures price. Calibrated parameters are: $\rho = -1, H = 0.09698, (\alpha_0, \alpha_1, \alpha_3, \alpha_5) = (0.61799,0.90211,1,0.0097)$.}
    \label{frac_jc1}
  \end{figure}
}

{
  \begin{figure}[H]
    \centering
    
    \includegraphics[width=0.5\textwidth]{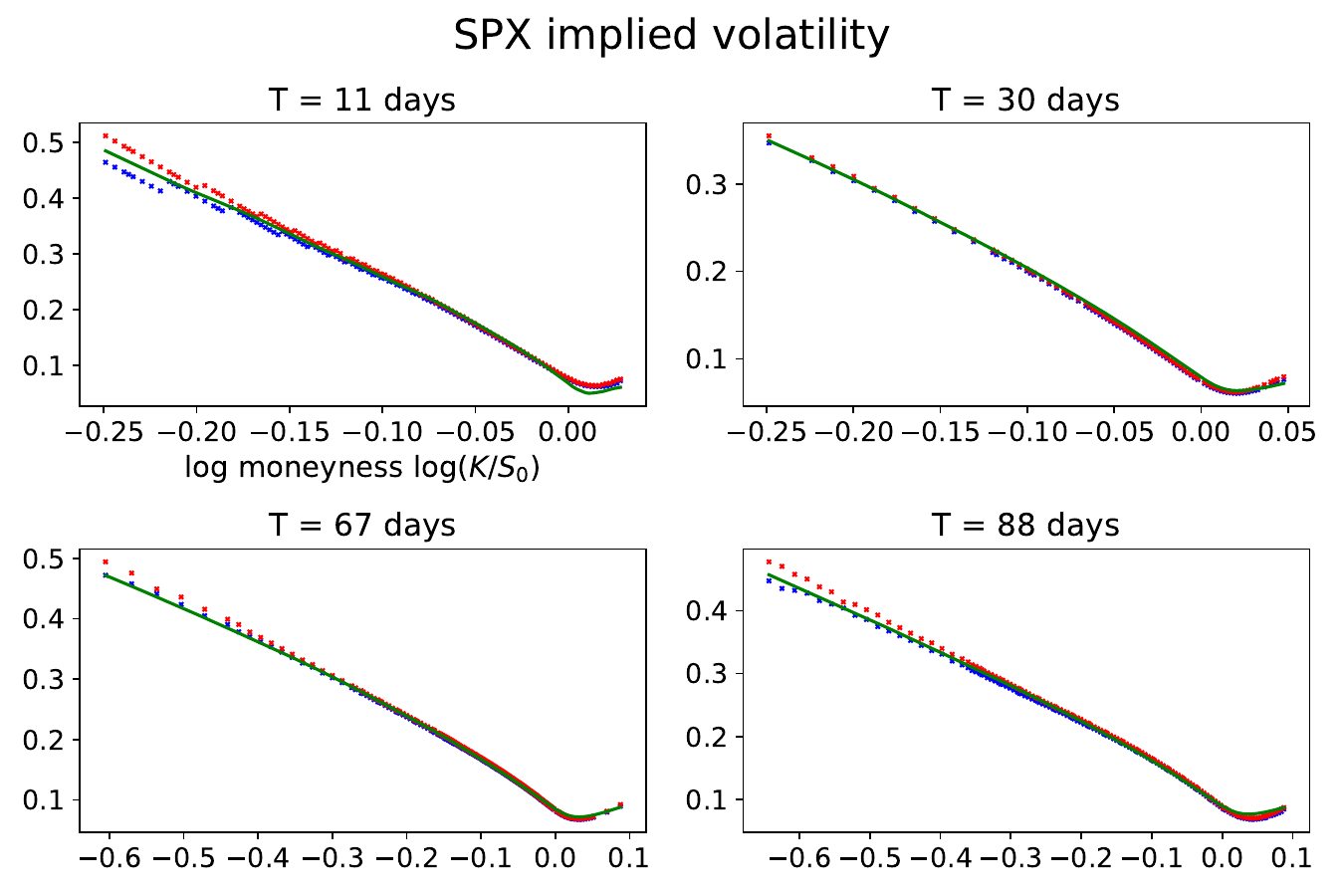}%
    \includegraphics[width=0.5\textwidth]{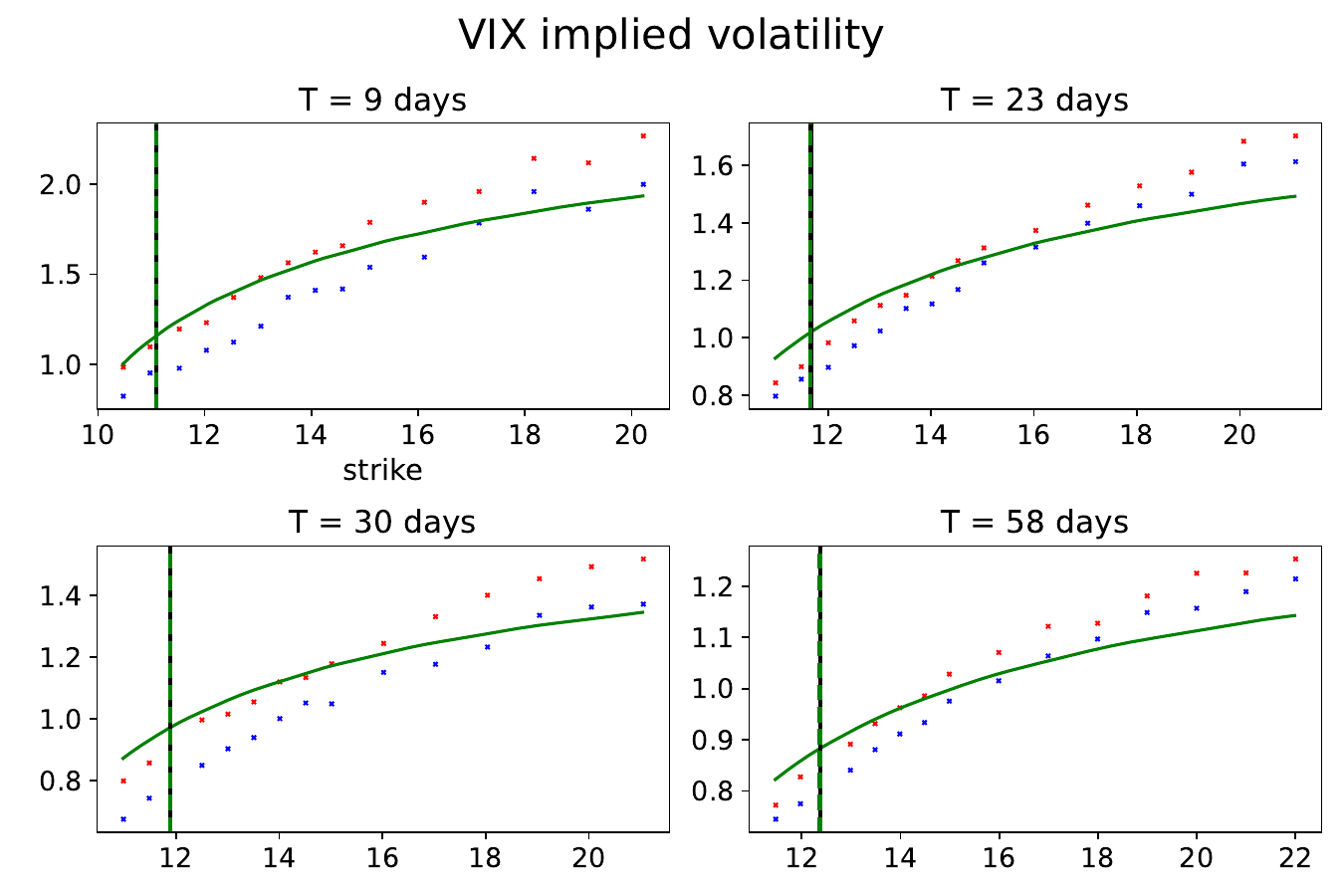}
    \caption{Log modulated kernel $K^{log}$: Joint calibration of SPX implied volatility, VIX implied volatility, and VIX Futures on 23 October 2017 using the log modulated kernel $K^{log}$. The blue and red dots are bid/ask implied volatilities, with the green lines are model fit. The vertical bars represent the VIX Futures price. Calibrated parameters are: $\rho = -0.993822, H = 0.01608,\beta=4, (\alpha_0, \alpha_1, \alpha_3, \alpha_5) = (0.32486,0.0,0.49826,0.017482)$.}
    \label{log_m_jc1}
  \end{figure}
}

{
  \begin{figure}[H]
    \centering
    
    \includegraphics[width=0.5\textwidth]{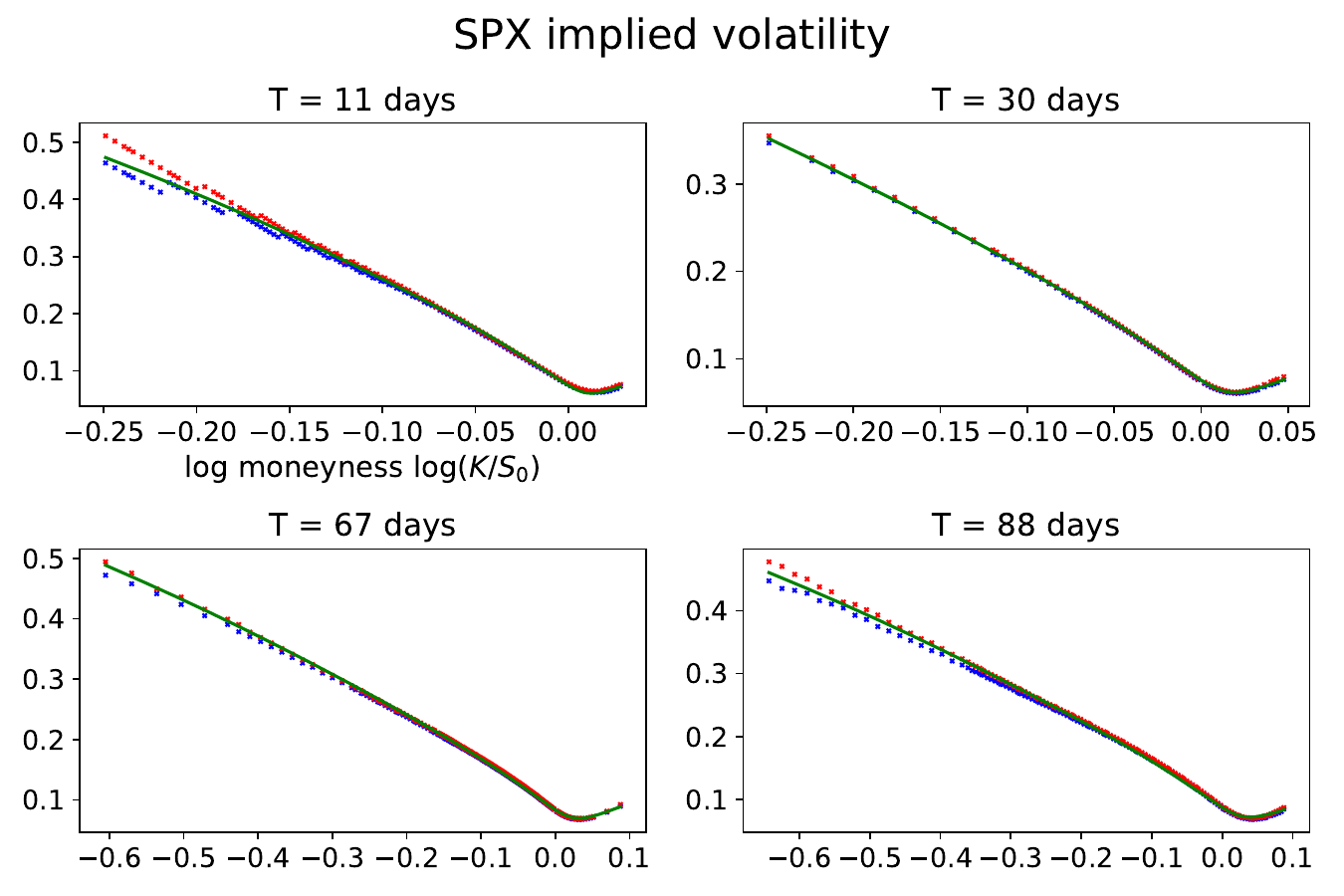}%
    \includegraphics[width=0.5\textwidth]{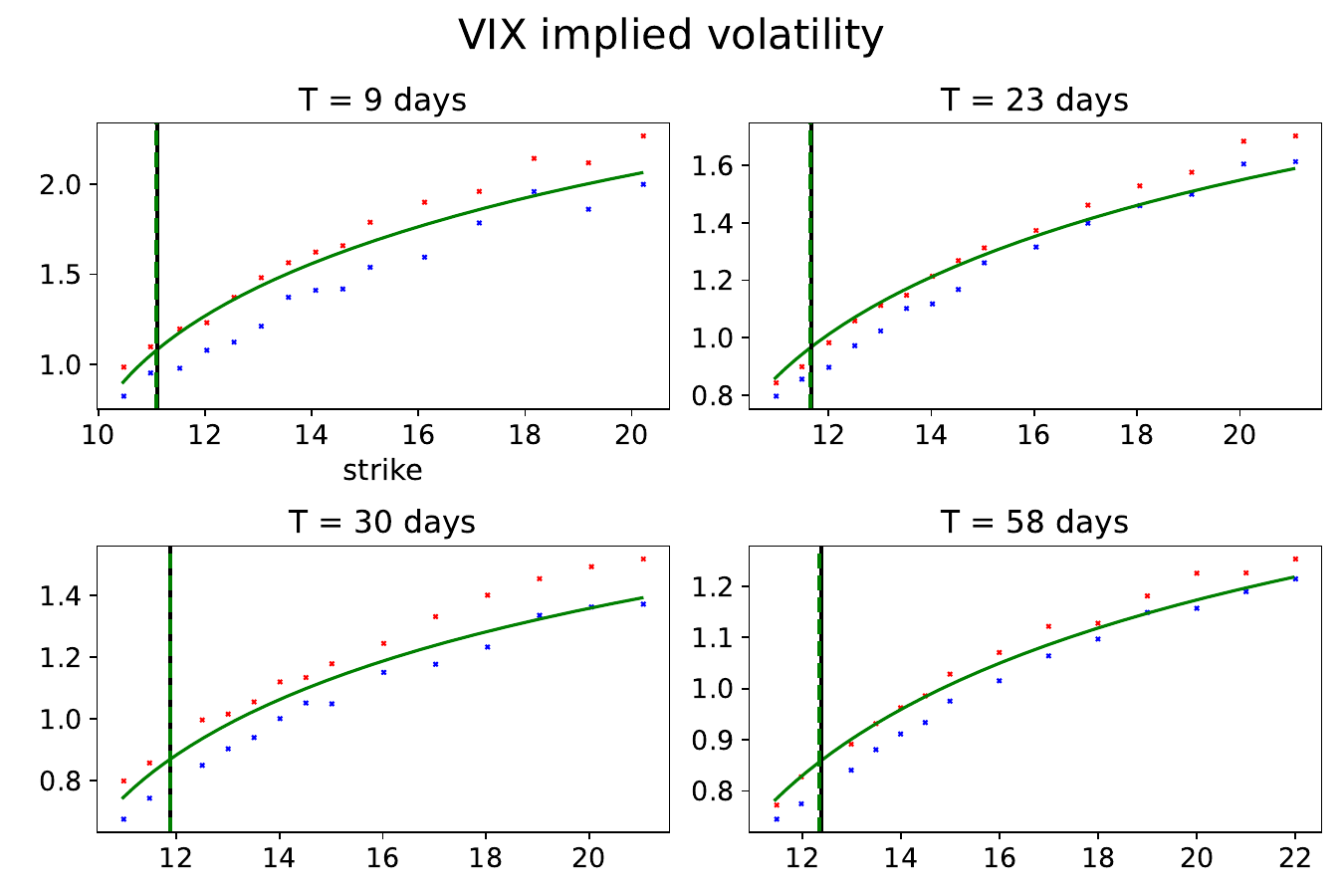}
    \caption{Shifted fractional kernel $K^{shift}$: Joint calibration of SPX implied volatility, VIX implied volatility, and VIX Futures on 23 October 2017 using the shifted fractional kernel $K^{shift}$. The blue and red dots are bid/ask implied volatilities, with the green lines are model fit. The vertical bars represent the VIX Futures price. Calibrated parameters are: $\rho = -0.6891, H = -0.845297, (\alpha_0, \alpha_1, \alpha_3, \alpha_5) = (0.1035051,0.163729,0.0001054,5.1882e-8)$.}
    \label{shift_jc1}
  \end{figure}
}

\subsection{Evolution of calibrated parameters under different kernels}\label{B:param_history}

{
  \begin{figure}[H]
    \centering
    \includegraphics[scale=0.45]{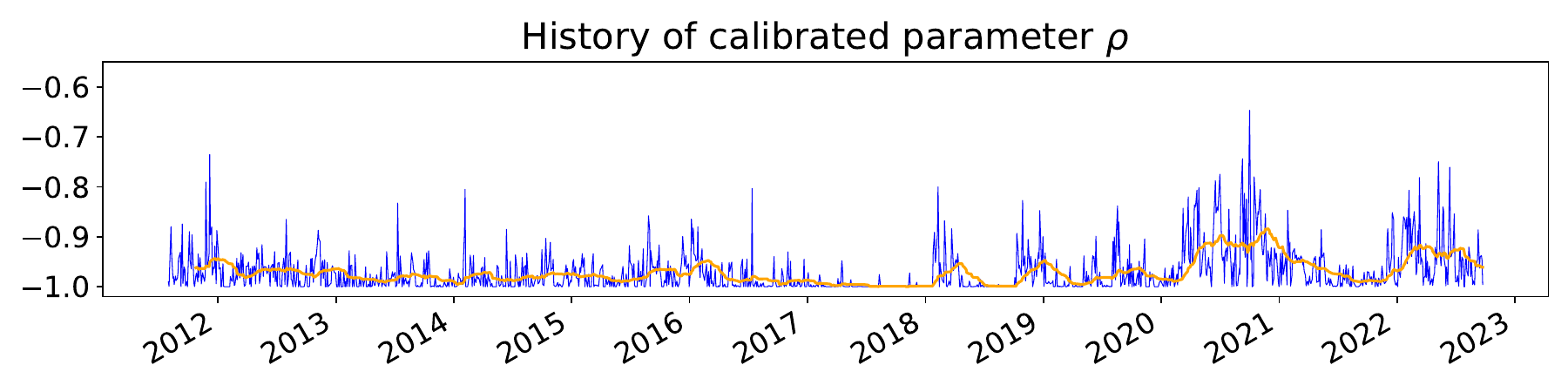}
    \includegraphics[scale=0.45]{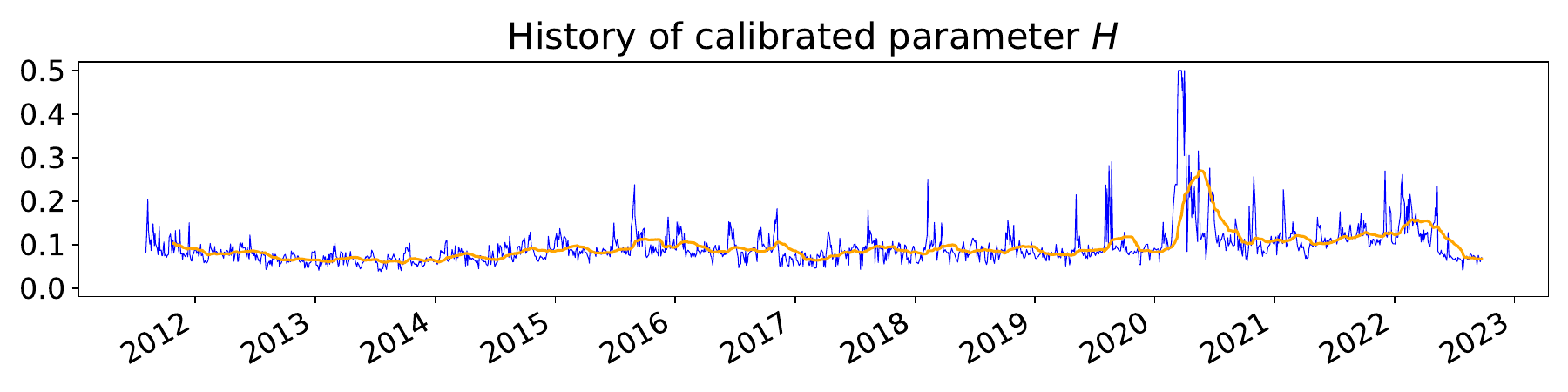}
    \caption{Fractional kernel $K^{frac}$: Evolution of the calibrated parameters $\rho$ and $H$ under the fractional kernel $K^{frac}$, the blue line is the actual value of the calibrated parameters in time, the orange line is the 30-day moving average. Note how $\rho$ is saturated at -1 on most days, but still not enough to capture the SPX ATM skew. $H$ is unable to descend to near zero due to the ``vanishing" skew phenomena.  }
    \label{frac_param_history}
  \end{figure}
}

{
  \begin{figure}[H]
    \centering
    \includegraphics[scale=0.45]{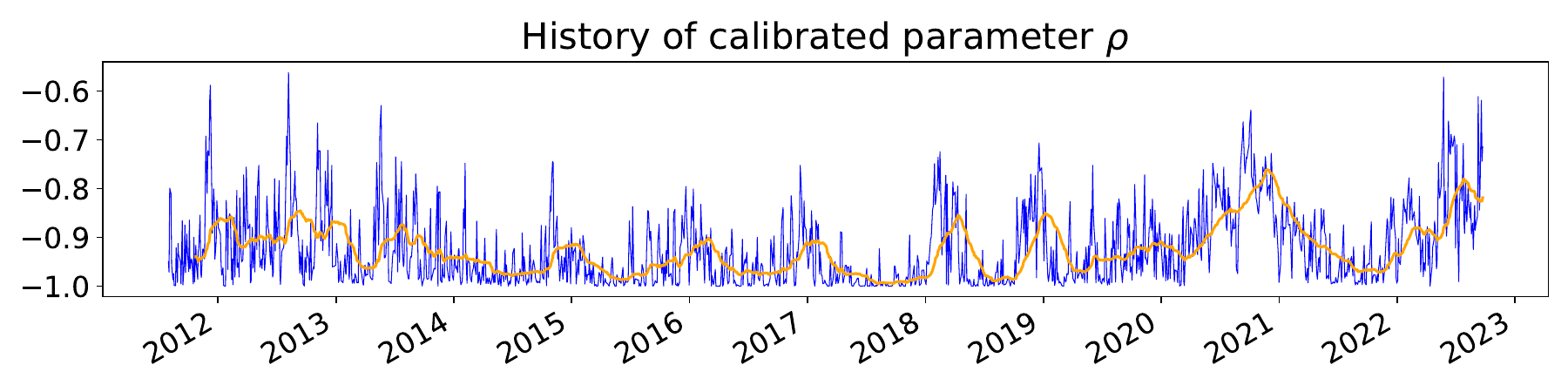}
    \includegraphics[scale=0.45]{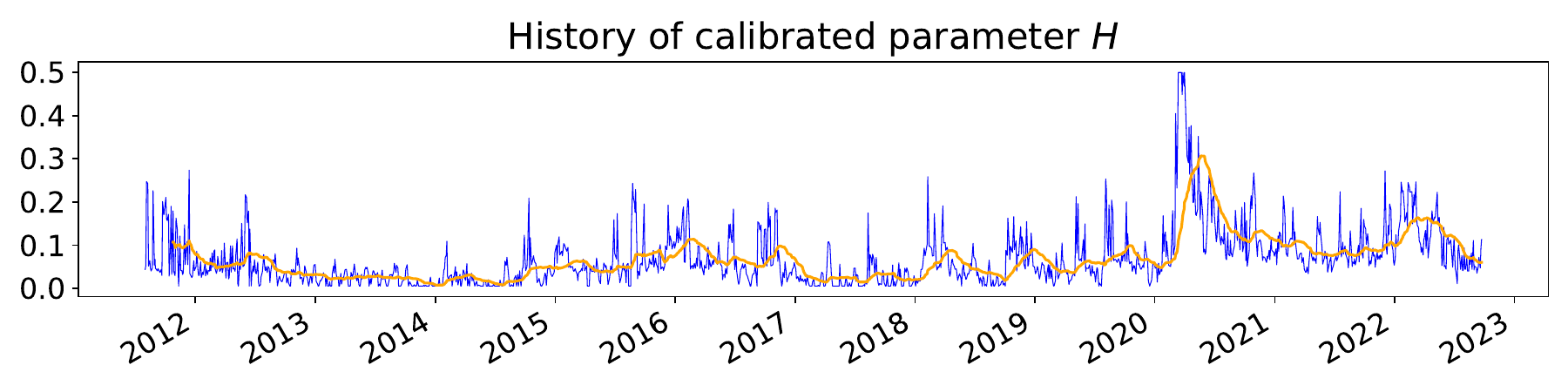}
    \caption{Log modulated kernel $K^{log}$: Evolution of the calibrated parameters $\rho$ and $H$ under the log modulated kernel $K^{log}$, the blue line is the actual value of the calibrated parameters in time, the orange line is the 30-day moving average. Note how $\rho$ is still very close to  -1 on most days, while $H$ is much closer to zero than that of fractional kernel $K^{frac}$. In any case, the log modulated kernel $K^{log}$ still struggles to jointly calibrate SPX \& VIX.}
    \label{log_m_param_history}
  \end{figure}
}

{
  \begin{figure}[H]
    \centering
    \includegraphics[scale=0.45]{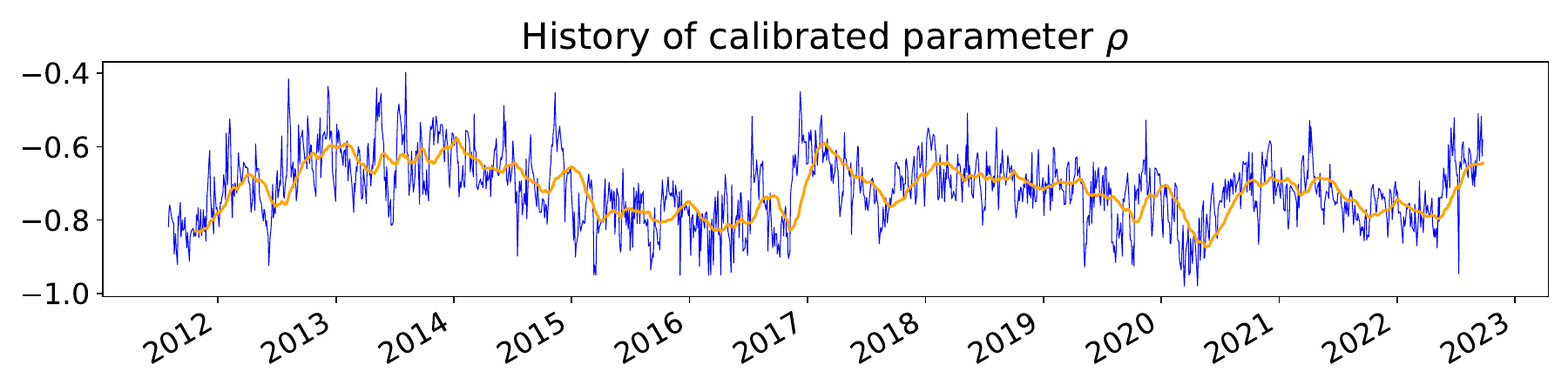}
    \includegraphics[scale=0.45]{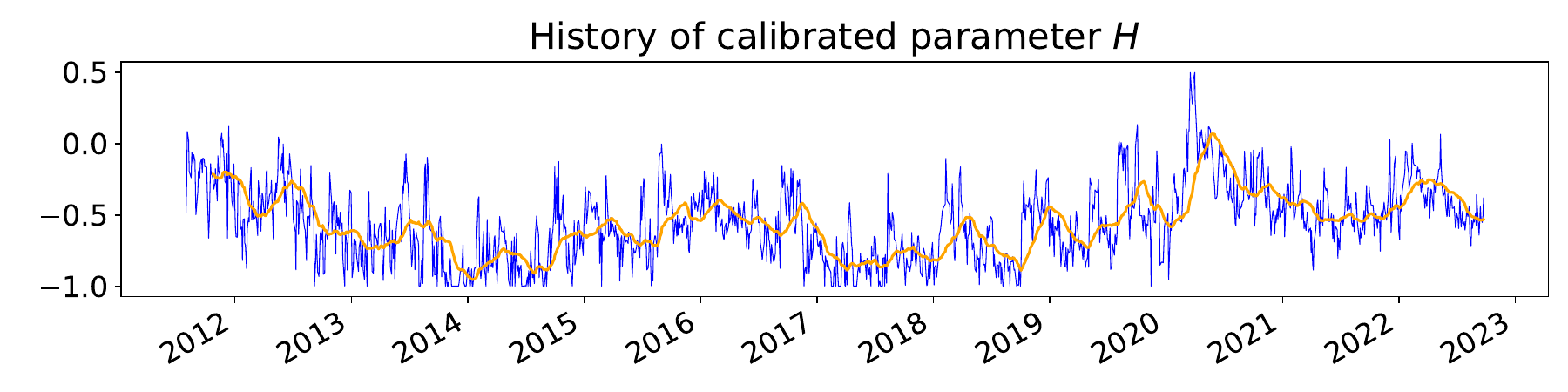}
    \caption{Shifted fractional kernel $K^{shift}$: Evolution of the calibrated parameters $\rho$ and $H$ under the shifted fractional kernel $K^{shift}$, the blue line is the actual value of the calibrated parameters in time, the orange line is the 30-day moving average. Note how $\rho$ is not saturated, similar to that of exponential kernel $K^{exp}$, with $H$ basically staying below zero the entire time series (except during 2020 Covid pandemic).}
    \label{shift_param_history}
  \end{figure}
}

\subsection{Exponential kernel \texorpdfstring{$K^{exp}$}{Lg} fit quality quantiles}\label{exp_ker_quantile}

\paragraph{0.3 Quantile:} \text{ }

{
  \begin{figure}[H]
    \centering
    
    \includegraphics[width=0.5\textwidth]{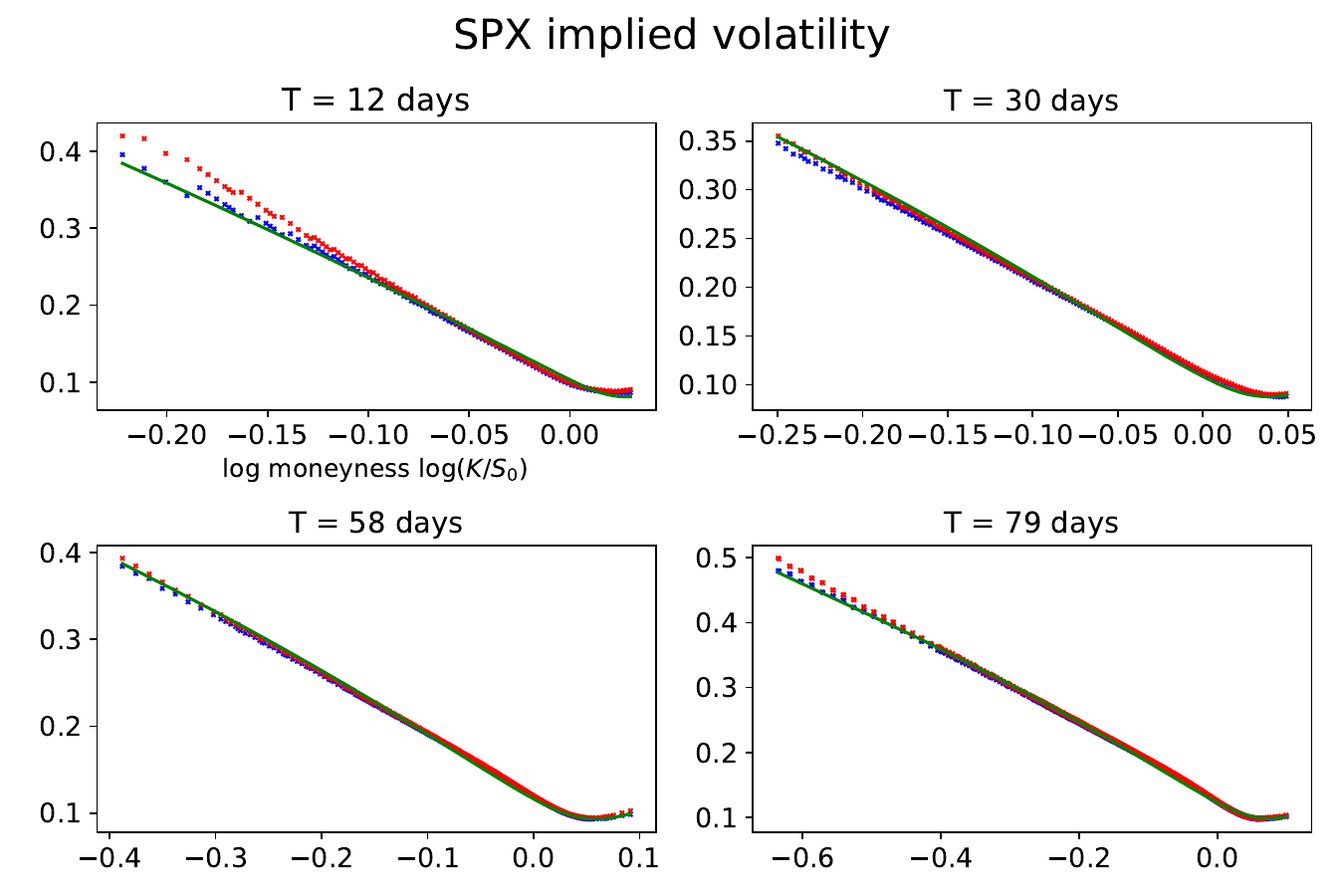}%
    \includegraphics[width=0.5\textwidth]{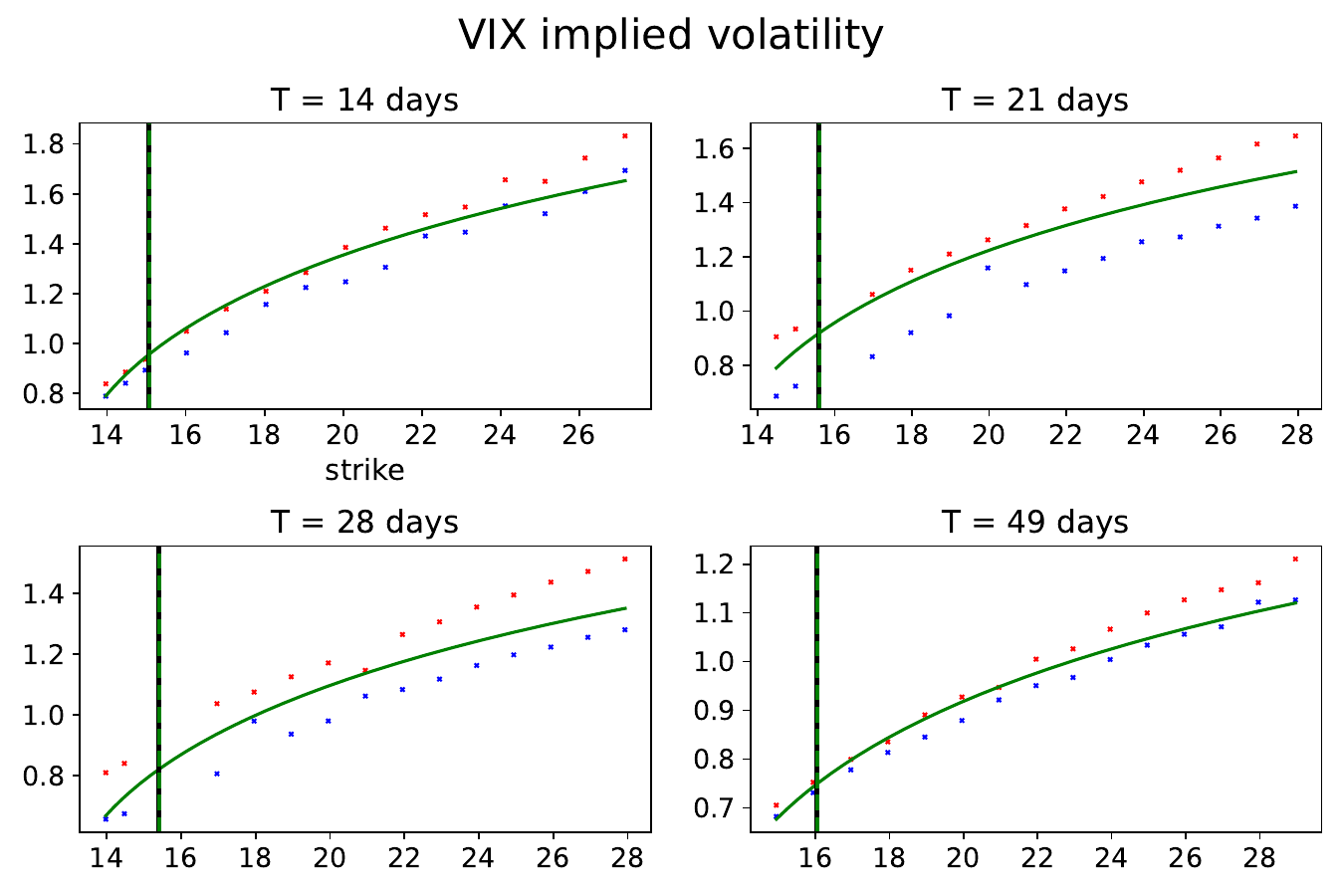}
    \caption{Exponential kernel $K^{exp}$: Joint calibration of SPX implied volatility, VIX implied volatility, and VIX Futures on 03 April 2019. The blue and red dots are bid/ask implied volatilities, and the green lines are model fit. The vertical bars represent the VIX Futures price. Calibrated parameters are: $\rho = -0.69148, H = 0.04743, (\alpha_0, \alpha_1, \alpha_3, \alpha_5) = (1,0.1,0.30223,0.04788)$.}
    \label{exp_30}
  \end{figure}
}

\paragraph{0.6 Quantile:} \text{ }

{
  \begin{figure}[H]
    \centering
    
    \includegraphics[width=0.5\textwidth]{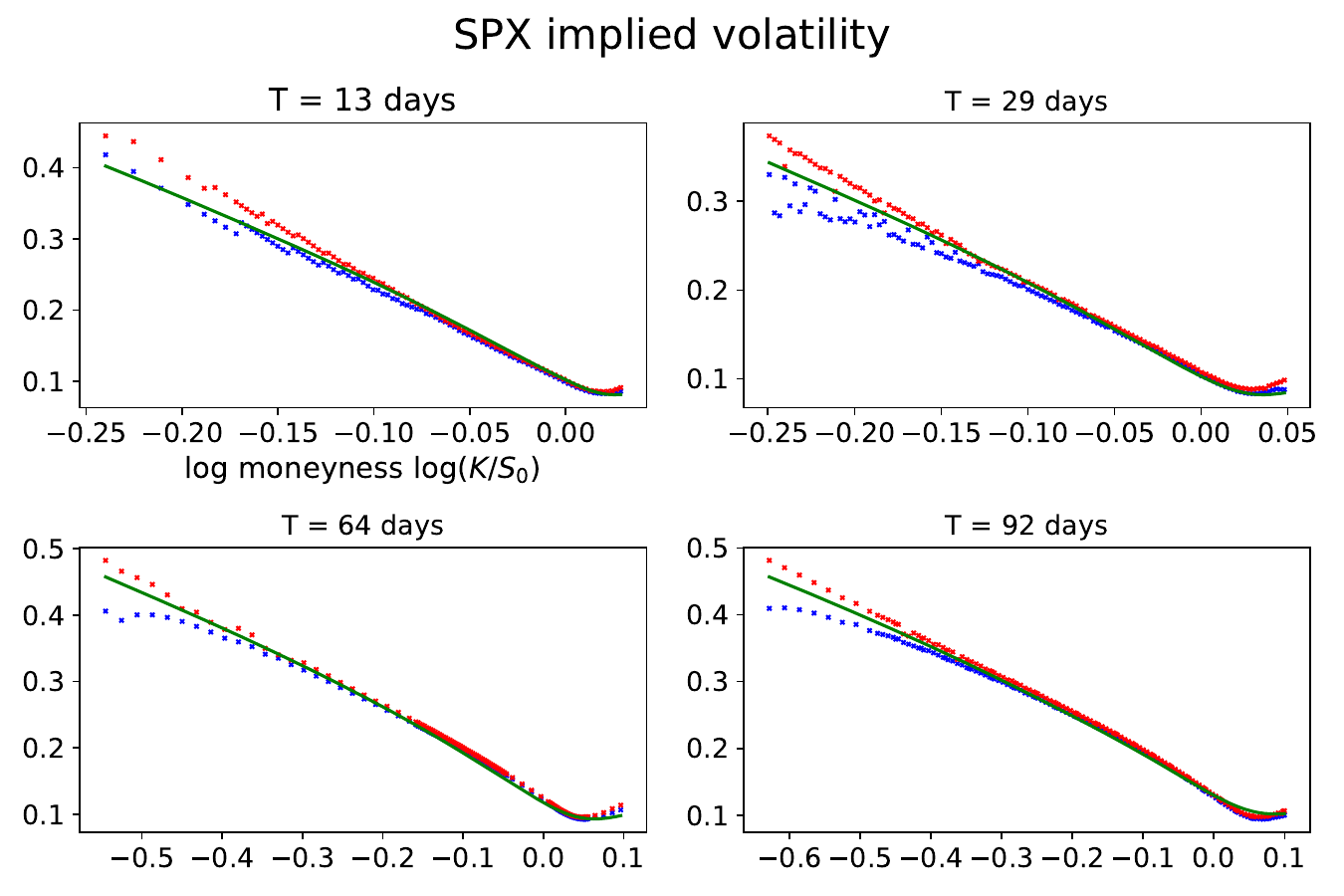}%
    \includegraphics[width=0.5\textwidth]{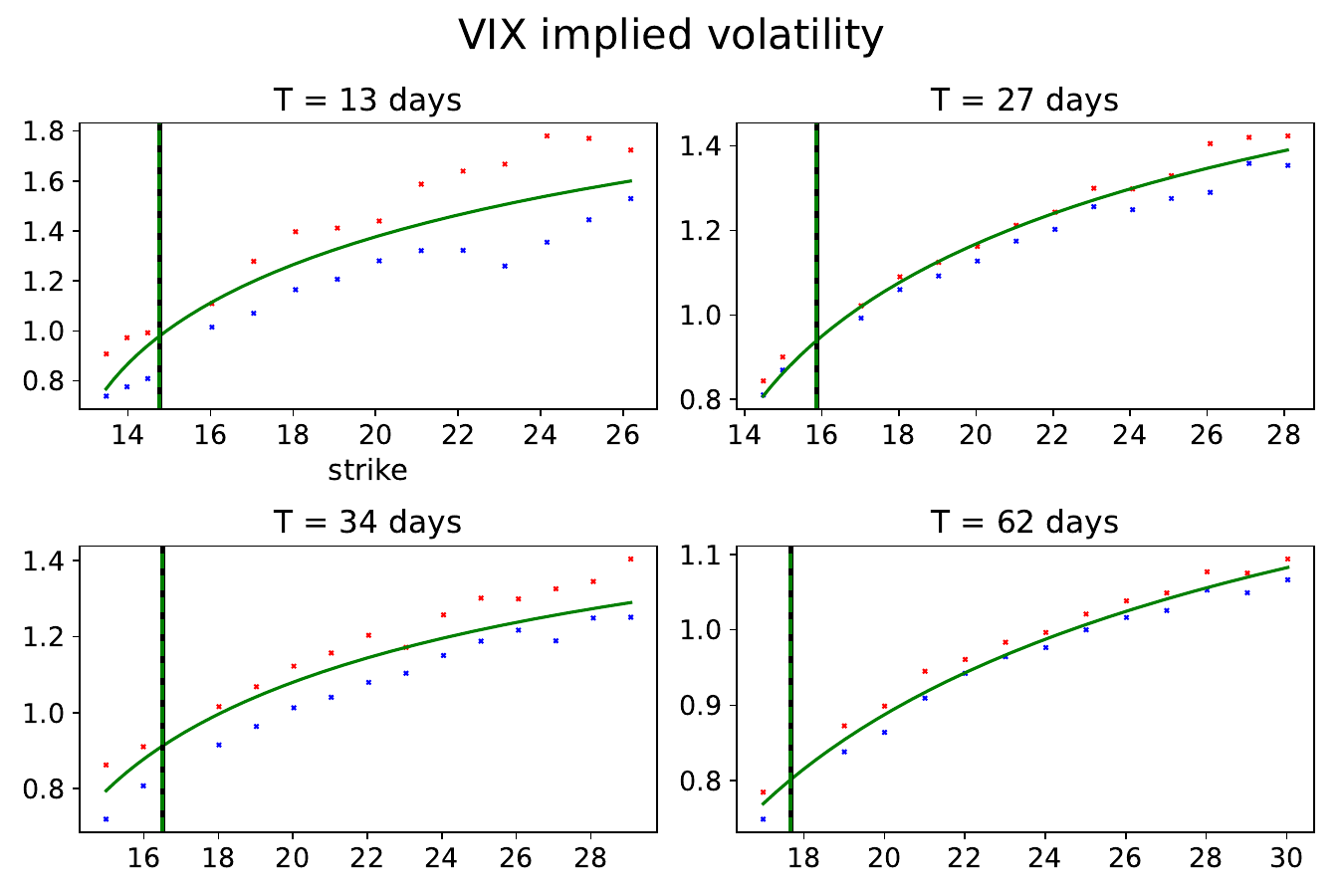}
    \caption{Exponential kernel $K^{exp}$: Joint calibration of SPX implied volatility, VIX implied volatility, and VIX Futures on 21 July 2016. The blue and red dots are bid/ask implied volatilities, and the green lines are model fit. The vertical bars represent the VIX Futures price. Calibrated parameters are: $\rho = -0.72793, H = 0.1373, (\alpha_0, \alpha_1, \alpha_3, \alpha_5) = (0.81221,0.73015,0.94968,0.02753)$.}
    \label{exp_60}
  \end{figure}
}

\paragraph{0.99 Quantile:} \text{ }

{
  \begin{figure}[H]
    \centering
    
    \includegraphics[width=0.5\textwidth]{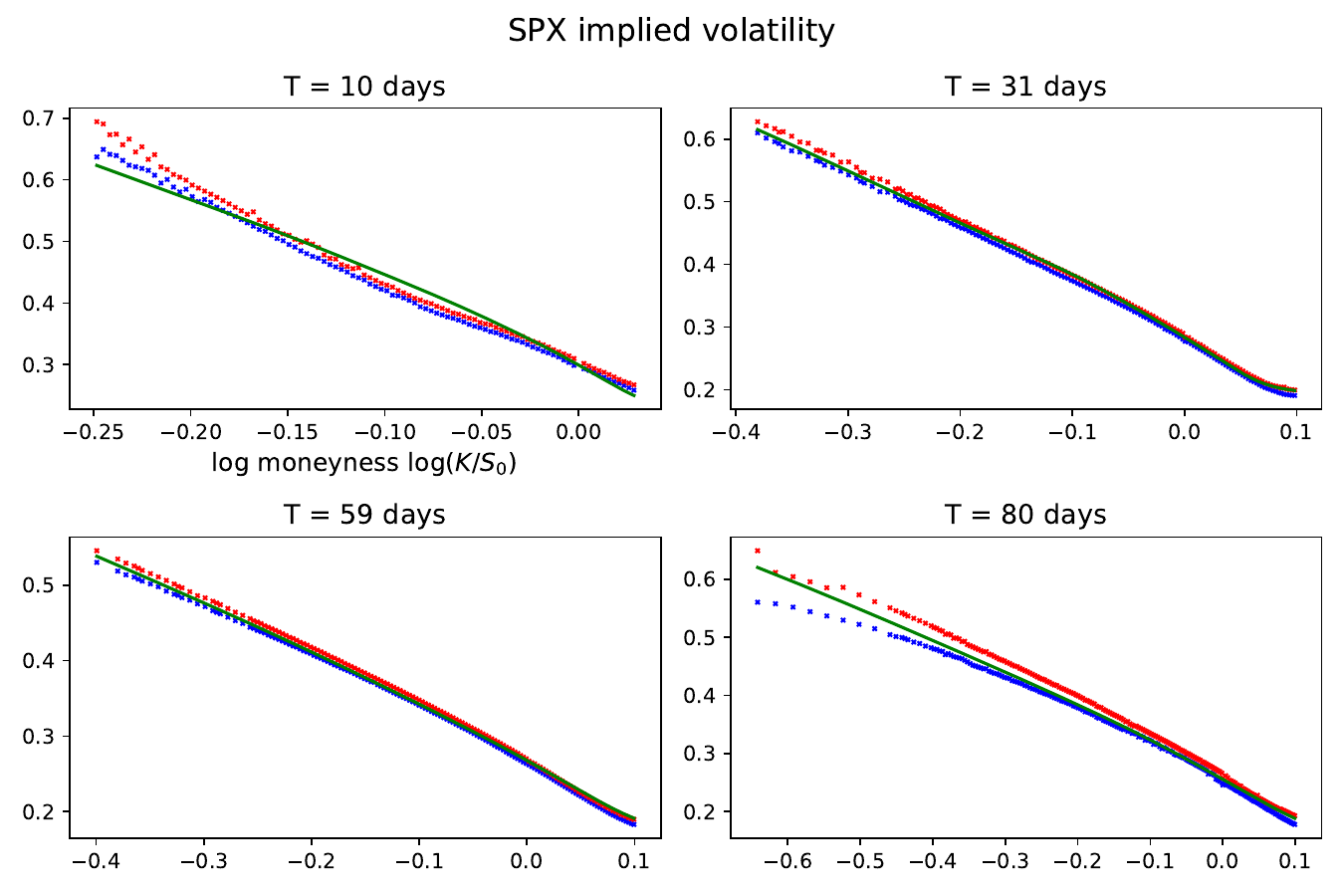}%
    \includegraphics[width=0.5\textwidth]{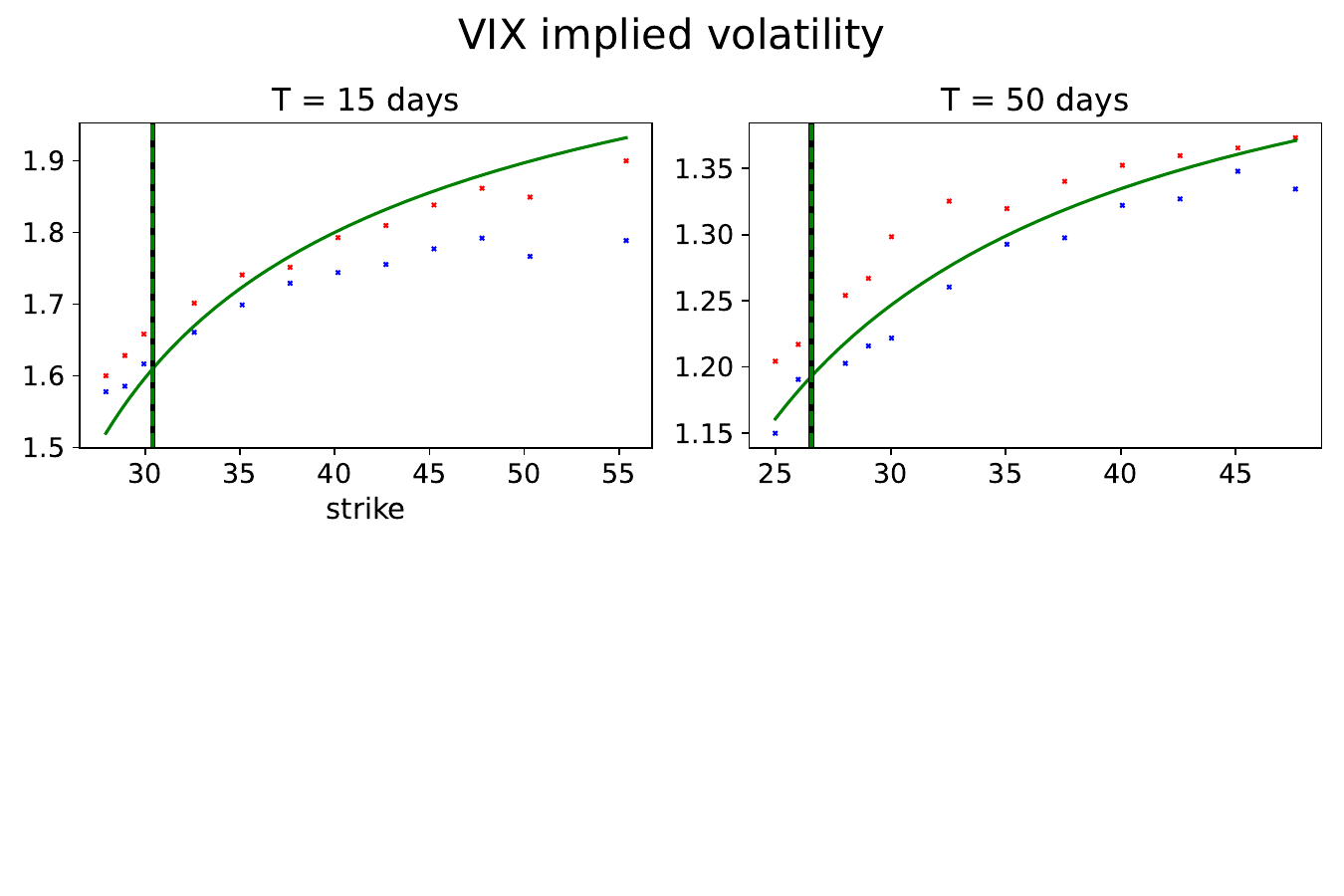}
    \caption{Exponential kernel $K^{exp}$: Joint calibration of SPX implied volatility, VIX implied volatility, and VIX Futures on 01 September 2015. The blue and red dots are bid/ask implied volatilities, and the green lines are model fit. The vertical bars represent the VIX Futures price. Calibrated parameters are: $\rho = -0.8171, H = 0.35267, (\alpha_0, \alpha_1, \alpha_3, \alpha_5) = (0.3021,0.8658,0.8522,0.0574)$.}
    \label{exp_99}
  \end{figure}
}

{
{\section{Monte Carlo variance reduction via Neural Networks }\label{C:nn_rss}

We performed an experiment to evaluate the Neural Networks' ability to reduce Monte Carlo noise reported in \cite{ferguson2018deeply} using the Volterra Stein-Stein model from \cite{abi2022characteristic} for which the prices of call options can be computed by Fourier inversion techniques without any Monte-Carlo simulation. Similar to the data generation process in  Section \ref{ss:data_gen}, we randomly sampled $100,000$ combinations of the rough Stein Stein model parameters for the following parameters: $\eta \sim \mathcal{U}[0.1,1], \rho \sim \mathcal{U}[-1,-0.3], X_0 \sim \mathcal{U}[0,0.2], \theta \sim \mathcal{U}[-0.1,0.1], H \sim \mathcal{U}[0.02,0.5], T \sim \mathcal{U}[0.01,0.5]$ with $\kappa=0$ and $\mathcal{U}$ denoting the uniform distribution. Next, we generated MC prices for each parameter combination using 80,000 simulations including antithetic variables, and 800 step size, split into 85\% training dataset and 15\% testing dataset. We applied functional quantization techniques as per Section \ref{S:SPXpricing} and trained our Neural Networks as per Section \ref{S:NN}.

Figures \ref{fig:nn_vs_fourier_maturity} and \ref{fig:mc_vs_fourier_maturity} indeed show that the prices produced by Quantization plus Neural Network are much closer to the actual Fourier prices compared to that of Monte Carlo on the testing dataset, despite being trained on Monte Carlo prices.

  \begin{figure}[H]
    \centering    \includegraphics[width=0.8\textwidth]{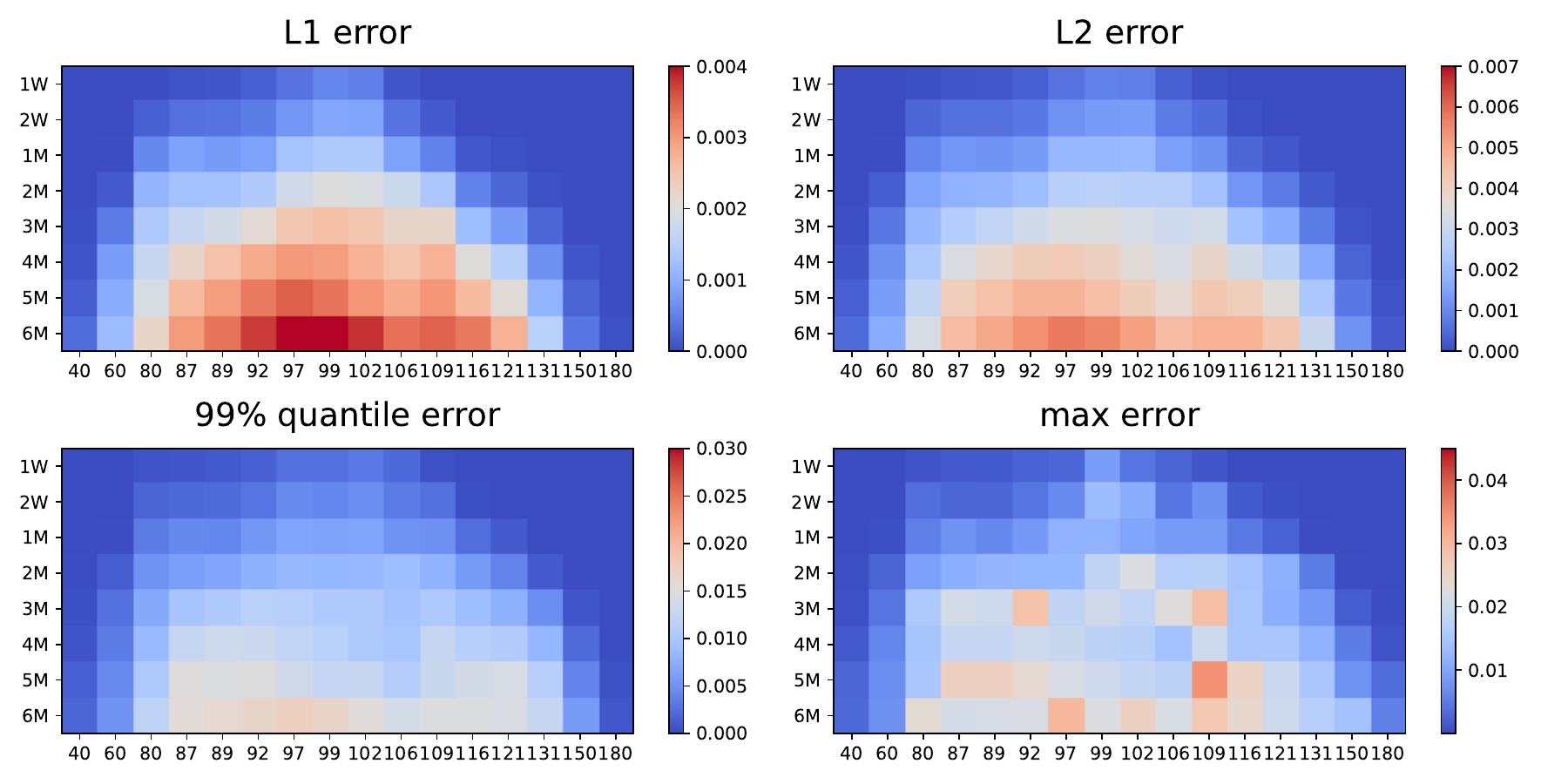}%
    \caption{Price error distribution of SPX call between trained Neural Networks' prices (trained using Monte Carlo data) and Fourier prices under the Volterra Stein-Stein model using the testing dataset. The horizontal axis is different strikes with $S_0 = 100$, the vertical axis is different maturities between 1 week to 6 months.}    \label{fig:nn_vs_fourier_maturity}
  \end{figure}
  
  \begin{figure}[H]
    \centering    \includegraphics[width=0.8\textwidth]{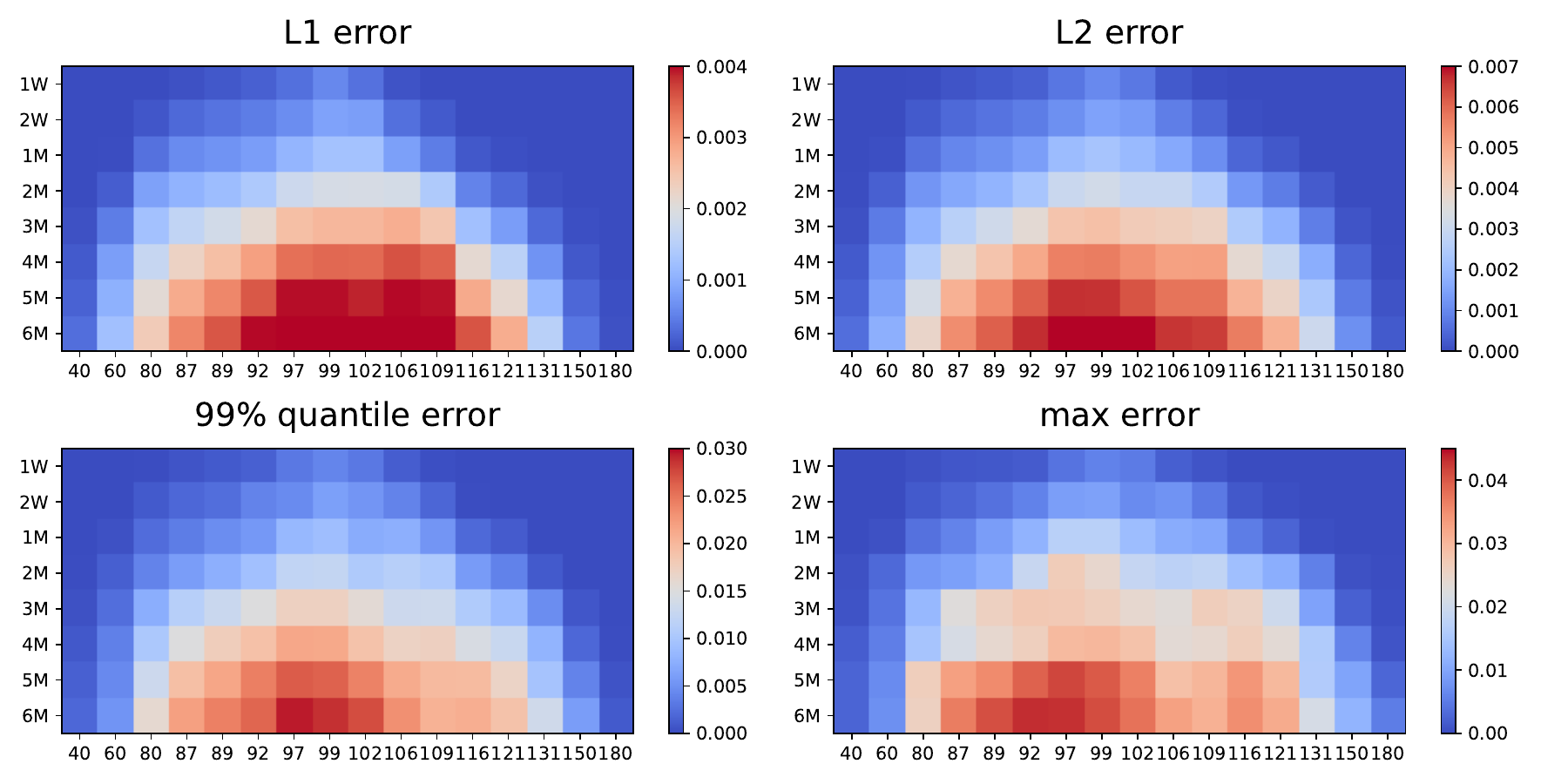}%
    \caption{Price error distribution of SPX call between MC prices and Fourier prices under the Volterra Stein-Stein model using the testing dataset. The horizontal axis is different strikes with $S_0 = 100$, the vertical axis is different maturities between 1 week to 6 months.}
    \label{fig:mc_vs_fourier_maturity}
  \end{figure}
}

}


\bibliographystyle{plainnat}
\bibliography{bibl.bib}

\end{document}